\documentclass[]{llncs}

\usepackage{amsmath,proof}
\usepackage{amssymb}
\usepackage{stmaryrd}
\usepackage{mathpartir}

\usepackage{color}
\usepackage[]{todonotes}
\usepackage{thm-restate}
\usepackage{hyperref}

\newcommand{\lft}[1]{{{#1}\mathsf{L}}}
\newcommand{\rgt}[1]{{{#1}\mathsf{R}}}

\newcommand{\B}[1]{#1}

\newcommand{\m}[1]{\mathsf{#1}}
\newcommand{\nub}{{\boldsymbol{\nu}}}

\newcommand{\zero}{\boldsymbol{0}}

\newcommand{\pcase}[1]{\m{case}\{#1\}}

\newcommand{\seq}{\vdash} 
\newcommand{\semi}{\mathrel{;}}
\newcommand{\lolli}{\multimap}
\newcommand{\with}{\mathbin{\binampersand}}
\newcommand{\tensor}{\otimes}
\newcommand{\one}{\mathbf{1}}
\newcommand{\bang}{{!}}

\newcommand{\cut}{\mathsf{cut}}

\newcommand{\lb}{\llbracket}
\newcommand{\rb}{\rrbracket}

\newcommand{\Ga}{\Gamma}
\newcommand{\D}{\Delta}

\newcommand{\ccase}[2]{\mathsf{case}\;#1\;\mathsf{of}\;#2}

\newcommand{\hPi}{\Pi}

\newcommand{\tra}[1]{\xrightarrow{#1}}

\newcommand{\type}{\m{type}}
\newcommand{\stype}{\m{stype}}
\newcommand{\ov}[1]{\overline{#1}}

\newcommand{\monad}[1]{\{#1\}}

\newcommand{\ltrue}{\mathtt{t}}
\newcommand{\lfalse}{\mathtt{f}}
\newcommand{\myparagraph}[1]{\paragraph{\bf #1}}

\makeatletter
\let\c@proposition\c@theorem
\let\c@corollary\c@theorem
\let\c@lemma\c@theorem
\let\c@definition\c@theorem
\let\c@example\c@theorem
\let\c@remark\c@theorem
\makeatother
\numberwithin{proposition}{section}
\numberwithin{corollary}{section}
\numberwithin{lemma}{section}
\numberwithin{theorem}{section}
\numberwithin{definition}{section}
\numberwithin{example}{section}
\numberwithin{remark}{section}

\begin{document}

\title{Depending on Session-Typed Processes} 
\author{Bernardo Toninho \and Nobuko Yoshida}
\institute{Imperial College London, United Kingdom}

\maketitle

\begin{abstract}
  %
  This work proposes a dependent type theory that combines functions
  and session-typed processes (with value dependencies) through a
  contextual monad, internalising typed processes in a
  dependently-typed $\lambda$-calculus. The proposed framework, by
  allowing session processes to depend on functions and vice-versa,
  enables us to specify and statically verify protocols where the
  choice of the next communication action can depend on specific
  values of received data. Moreover, the type theoretic nature of the
  framework endows us with the ability to internally describe and
  prove predicates on process behaviours.
  Our main results are type soundness of the framework, and a faithful
  embedding of the functional layer of the calculus within the
  session-typed layer, showcasing the expressiveness of dependent
  session types.

\end{abstract}


\section{Introduction}
\label{sec:intro}

Session types
\cite{DBLP:conf/parle/TakeuchiHK94,DBLP:conf/esop/HondaVK98} are a
typing discipline for communication protocols, whose simplicity
provides an extensible framework that allows for integration with a
variety of functional type features.  One useful instance arising from
the proof theoretic exploration of logical quantification is {\em
  value dependent session types}
\cite{Toninho:2011:DST:2003476.2003499}. In this work, one can express
properties of exchanged data in protocol specifications separately
from communication, but {\em cannot} describe protocols where
communication actions depend on the actual exchanged data 
(e.g.~\cite[\S~2]{DBLP:journals/pacmpl/IgarashiTVW17}).
Moreover, it does not allow functions or values to depend on protocols
(i.e.~sessions) or communication, thus preventing reasoning about
dependent process behaviours, exploring the proofs-as-programs
paradigm of dependent type theory,
e.g.~\cite{DBLP:conf/popl/KrishnaswamiPB15,DBLP:conf/esop/GeorgesMOP17}.

Our work addresses the limitations of existing formulations of
session types by proposing a type theory that
integrates dependent functions \emph{and} session types using a
\emph{contextual monad}. This monad internalises a session-typed calculus
within a dependently-typed $\lambda$-calculus. By allowing session
types to depend on $\lambda$-terms \emph{and} $\lambda$-terms to
depend on typed processes (using the monad), we are able to achieve
heightened degrees of expressiveness. Exploiting the former direction,
we enable writing actual data-dependent communication
protocols. Exploiting the latter, we can define and \emph{prove}
properties of linearly-typed objects (i.e. processes) within our
intuitionistic theory.


%

To informally demonstrate how our type theory goes beyond the
state of the art in order to represent
data-dependent protocols, consider the following session type
(we write $\tau \wedge A$ for $\exists x{:}\tau.A$ where $x$ does not
occur in $A$ and similarly $\tau \supset A$ for $\forall x{:}\tau.A$
when $x$ is not free in $A$),
$T \triangleq \m{Bool}\supset \oplus\{\ltrue : \m{Nat} \wedge \one ,
\lfalse : \m{Bool}\wedge \one\}$, representable in existing
session typing systems.  The type $T$ denotes a protocol which first,
inputs a boolean and then either emits the label $\ltrue$, which will
be followed by an output of a natural number; or emits the label
$\lfalse$ and a boolean. \B{The intended protocol described by} $T$ is to
take the $\ltrue$ branch if the received value is $\ltrue$ and the
$\lfalse$ branch otherwise, which we can implement as $Q$ with channel
$z$ typed by $T$ as follows:
\[
\begin{array}{lcl}
Q & \triangleq & z(x).\ccase{x}{(\m{true} \Rightarrow z.\ltrue;
  z\langle 23 \rangle.\zero, \ 
        \m{false} \Rightarrow z.\lfalse;z\langle \m{true}
        \rangle.\zero)}
\end{array}
\]
where $z(x).P$ denotes an input process, $z.\ltrue$ is a process which
selects label $\ltrue$ and $z\langle 23 \rangle.P$ is an output on
$z$. However, since the specification is imprecise, process
$z(x).\ccase{x}{(\m{false} \Rightarrow z.\ltrue; z\langle 23
  \rangle.\zero, \ \m{true} \Rightarrow z.\lfalse;z\langle \m{true}
  \rangle.\zero)}$ is also a type-correct implementation of $T$ that
does not adhere to the intended protocol.
Using our dependent type system, we can narrow the specification to
guarantee that the desired protocol is precisely enforced.  Consider
the following definition of a session-type level conditional where we
assume inductive definition and dependent pattern matching
mechanisms ($\stype$ denotes the \emph{kind} of session types):
\[
  \begin{array}{l}
    \mathtt{if} :: \m{Bool}\rightarrow \stype \rightarrow \stype
    \rightarrow \stype\\
    \mathtt{if}\ \m{true}\, A\, B \ = \ A \quad \quad 
    \mathtt{if}\ \m{false}\, A\, B \ = \ B
  \end{array}
\]

The type-level function above case analyses the boolean and produces
its first session type argument if the value is $\m{true}$ and the
second otherwise.
We may now specify a session type that faithfully implements the
protocol:
\[
  T' \triangleq \forall x{:}\m{Bool}.\mathtt{if}\,x\,(\m{Nat} \wedge \one)\,
  (\m{Bool}\wedge \one)
\]
A process $R$ implementing such
a type on channel $z$ is given below:
\[
\begin{array}{lcl}
R  \triangleq  z(x).\ccase{x}{(\m{true} \Rightarrow z\langle 23
  \rangle.\zero, \ 
\m{false} \Rightarrow z\langle \m{true} \rangle.\zero)}
\end{array}
\]
Note that if we flip the two branches of the case analysis \B{in $R$,
  the session is} no longer typable with $T'$, ensuring that the
protocol is implemented faithfully.

The example above illustrates a simple yet useful data-dependent
protocol.  When we further extend our dependent types with a
\emph{process} monad \cite{DBLP:conf/esop/ToninhoCP13}, where
$\{c \leftarrow P \leftarrow \ov{u_j};\ov{d_i}\}$ is a functional term
denoting a process that may be \emph{spawned} by other processes by
instantiating the names in $\ov{u_j}$ and $\ov{d_i}$,
we can provide more powerful reasoning on
processes, enabling refined specifications through
the use of type indices (i.e. type families) and an ability to
internally specify and verify predicates on process behaviours.
We also show that {\em all} functional types and terms can be
faithfully embedded in the process layer using the dependently-typed
sessions and process monads.

\paragraph{\bf Contributions.} \S~\ref{sec:sys} introduces
our dependent type theory, augmenting the example above by showing how
we can reason about process behaviour using type families and
dependently-typed functions (\S~\ref{sec:exreasoning}). We then
establish the soundness of the theory (\S~\ref{sec:meta}).
\S~\ref{sec:embed} develops a faithful embedding of the dependent
function space in the process layer (Theorem~\ref{thm:opcorrcomp}).
\S~\ref{sec:conc} concludes with related work.
This article is a long version of \cite{depstypes} containing omitted
definitions, proofs and additional examples.


\section{A Dependent Type Theory of Processes}
\label{sec:sys}
\begin{figure}[t]
$
\begin{array}{llcl}
\mbox{Kinds } & K,K' & ::= &\m{type} \mid \m{stype} \mid \Pi x{:}\tau.K \mid \Pi t{:}K . K'\\
\mbox{Functional} \quad & \tau,\sigma & ::= &\Pi x{:}\tau.\sigma \mid
\lambda x{:} \tau . \sigma \mid \tau \; M \mid
 \{\ov{u_j{:}B_j} ; \ov{d_i{:}A_i} \vdash c{:}A\} 
\mid \lambda t :: K.\tau \mid \tau\,\sigma \\
\mbox{Sessions} & A,B & ::= & \bang A \mid A \lolli B \mid A
 \tensor B \mid \forall x{:}\tau . A \mid \exists x{:}\tau. A \mid
 \one 
\\
&&\mid & 
\with\{
\overline{l_i : A_i} \} \mid \oplus\{\overline{l_i : A_i}\} \mid 
 \lambda x{:}\tau . A \mid  A\;M \mid 
\lambda t {::} K. A \mid A \, B\\
\mbox{Terms } & M,N & ::= & \lambda x {:}\tau. M \mid \{ c \leftarrow P \leftarrow  \ov{u_j} ; \ov{d_i}\} \mid 
M\, N \mid x  \\

\mbox{Processes } & P,Q & ::= & 
\ov{c}\langle d \rangle.P  \mid (\nub c)P \mid c(x).P \mid c\langle M
\rangle.P \mid \bang c(x).P
\\
&& \mid & c.\pcase{\ov{l_i \Rightarrow P_i}} \mid c.l ; P \mid
[c\leftrightarrow d] \mid \zero
\mid c \leftarrow M \leftarrow  \ov{u_j} ; \ov{d_i}  ; Q
\end{array}
$
\caption{Syntax of Kinds, Types, Terms and Processes\label{fig:syntax}}
\end{figure}

This section introduces our dependent type theory combining
session-typed processes and functions. The theory is a generalisation
of the line of work relating linear logic and session types
\cite{DBLP:conf/concur/CairesP10,Toninho:2011:DST:2003476.2003499,DBLP:conf/esop/ToninhoCP13},
considering type-level functions and dependent kinds
in an intensional type theory with full
\emph{mutual} dependencies between functions and processes.
This generalisation enables
us to express more sophisticated session types (such as those of
\S~\ref{sec:intro}) and also to define and \emph{prove} properties of
processes expressed as type families with proofs as their inhabitants.
%
We focus on the new rules and judgements,
pointing the interested reader to
\cite{Toninho:2011:DST:2003476.2003499,DBLP:journals/mscs/CairesPT16,toninhothesis}
for additional details on the base theory.

\subsection{Syntax}
\label{subsec:syntax}

The calculus is stratified into
two mutually dependent layers of processes and terms, which
we often refer to as the \emph{process} and \emph{functional}
layers, respectively. The syntax of the
theory is given in Fig.~\ref{fig:syntax} (we use $x,y$ for variables
ranging over terms and $t$ for variables ranging over types).

\myparagraph{Types and Kinds.} The process layer
is able to refer to terms of the functional layer via appropriate
(dependently-typed) communication actions and through a \emph{spawn}
construct, allowing for processes encapsulated as functional values to
be executed. Dually, the functional layer can refer to the process
layer via a \emph{contextual} monad \cite{DBLP:conf/esop/ToninhoCP13}
that internalises (open) typed processes as opaque functional
values.
This mutual dependency is also explicit in the type structure on
several axes: process channel usages are typed by a language of
session types, which specifies the communication protocols implemented
on the used channels, extended with two dependent communication
operations $\forall x{:}\tau.A$ and $\exists x{:}\tau.A$, where $\tau$
is a functional type and $A$ is a session type in which $x$ may
occur. Moreover, we also extend the language of session types with
type-level $\lambda$-abstraction over terms $\lambda x{:}\tau.A$
and session types $\lambda t {::} K.A$ (with the corresponding elimination forms
$A\,M$ and $A\,B$). As we show in \S~\ref{sec:intro}, the combination of
these features allows for a new degree of expressiveness, enabling us
to construct session types whose structure depends on previously
communicated values.

The remaining session constructs are standard, following
\cite{DBLP:journals/mscs/CairesPT16}: $\bang A$ denotes a \emph{shared} session of
type $A$ that may be used an arbitrary (finite) number of times;
$A \lolli B$ represents a session offering to input a session of type
$A$ to then offer the session behaviour $B$; $A \tensor B$ is the dual
operator, denoting a session that outputs $A$ and proceeds as $B$;
$\oplus\{\ov{l_i : A_i}\}$ and $\with\{\ov{l_i : A_i}\}$ represent
internal and external labelled choice, respectively;
$\one$ denotes the terminated session.

The functional layer is a $\lambda$-calculus with dependent
functions $\Pi x{:}\tau.\sigma$, type-level $\lambda$-abstractions
over terms and types (and respective type-level applications) and a
\emph{contextual monadic} type $\{\ov{u_j{:}B_j} ; \ov{d_i{:}A_i}
\vdash c{:}A\}$, denoting a
(quoted) process offering session $c{:}A$ by using the \emph{linear}
sessions $\ov{d_i{:}A_i}$ and \emph{shared} sessions
$\ov{u_j{:}B_j}$  \cite{DBLP:conf/esop/ToninhoCP13}. We often write
$\monad{A}$ for $\monad{\cdot;\cdot\vdash c{:}A}$.
The kinding system for our theory contains two base kinds $\type$ and
$\stype$ of functional and session types, respectively.
Type-level $\lambda$-abstractions require dependent kinds $\Pi
x{:}\tau.K$ and $\Pi t {::} K.K'$, respectively. We note that the
functional connectives form a standard dependent type theory \cite{DBLP:journals/jacm/HarperHP93,norell:thesis}.

\myparagraph{Terms and Processes.}
Terms include the standard
$\lambda$-abstractions $\lambda x{:}\tau.M$, applications $M\,N$ and
variables $x$. In order to internalise processes within the functional
layer we make use of a monadic process wrapper, written
$\monad{c \leftarrow P \leftarrow \ov{u_j} ; \ov{d_i}}$. In such a
construct, the channels $c$, $\ov{u_j}$ and $\ov{d_i}$ are bound in
$P$, where $c$ is the session channel being offered and $\ov{u_j}$ and
$\ov{d_i}$ are the session channels (linear and shared, respectively)
being used. We write $\monad{c \leftarrow P \leftarrow \epsilon}$ when
$P$ does not use any ambient channels, which we abbreviate to
$\monad{P}$.

The syntax of processes follows that of
\cite{DBLP:journals/mscs/CairesPT16} extended with the monadic
elimination form $c \leftarrow M \leftarrow  \ov{u_j} ; \ov{d_i}  ;
Q$. Such a process construct denotes a term $M$ that is to be 
evaluated to a monadic value of the form $\monad{c \leftarrow P
  \leftarrow \ov{u_j} ; \ov{d_i}}$ which will then be executed in
parallel with $Q$, sharing with it a session channel $c$ and using the
provided channels $\ov{u_j}$ and $\ov{d_i}$. We write
$c\leftarrow M \leftarrow \epsilon ; Q$ when no channels are provided
for the execution of $M$ and often abbreviate this to $c\leftarrow M ;
Q$. 
The process $\ov{c}\langle d \rangle.P$ denotes the output of the
\emph{fresh} channel $d$ along channel $c$ with continuation $P$,
which binds $d$; $(\nub c)P$ denotes channel hiding, restricting the
scope of $c$ to $P$; $c(x).P$ denotes an input along $c$, bound to $x$
in $P$; $c\langle M \rangle.P$ denotes the output of term $M$ along
$c$ with continuation $P$; $\bang c(x).P$ denotes a replicated input
which spawns copies of $P$; the construct $c.\pcase{\ov{l_i \Rightarrow
    P_i}}$ codifies a process that waits to receive some label $l_j$
along $c$, with continuation $P_j$; dually, $c.l ; P$ denotes a
process that emits a label $l$ along $c$ and continues as $P$;
$[c\leftrightarrow d]$ denotes a forwarder between $c$ and $d$, which
is operationally implemented as renaming; $P\mid Q$ denotes parallel
composition and $\zero$ the null process.

\subsection{A Dependent Typing System}
\label{subsec:typing}
We now introduce our typing system, defined by a series of mutually
inductive judgements, given in Fig.~\ref{fig:typjudg}. We use $\Psi$
to stand for a typing context for dependent $\lambda$-terms (i.e. 
assumptions of the form $x{:}\tau$ or $t :: K$, not subject to exchange), $\Ga$ for a typing
context for \emph{shared} sessions of the form $u{:}A$ (implicitly
subject to weakening and contraction) and $\D$ for a linear context of
sessions $x{:}A$.  The context well-formedness judgments $\Psi \vdash$
and $\Psi ; \D \vdash$  require that types  and kinds (resp. session
types) in $\Psi$ (resp. $\D$) are well-formed.
The judgments $\Psi \vdash K$, 
$\Psi \vdash \tau :: K$ and $\Psi \vdash A :: K$ codify
well-formedness of kinds, functional and session types (with
kind $K$), respectively. Their rules are standard. 


\begin{figure}[t]
  \[
\small
\begin{array}{ll}
  \Psi \vdash  & \mbox{Context $\Psi$ is well-formed.}\\
  \Psi ; \D \vdash & \mbox{Context $\D$ is well-formed, under
                     assumptions in $\Psi$.}\\
\Psi \vdash K & \mbox{$K$ is a kind in context $\Psi$.}\\
\Psi \vdash \tau :: K & \mbox{$\tau$ is a (functional) type of kind
                        $K$ in context $\Psi$.}\\
\Psi \vdash A :: K & \mbox{$A$ is a session type of kind $K$ in context $\Psi$.}\\
\Psi \vdash M : \tau & \mbox{$M$ has type $\tau$ in context $\Psi$.}\\
\Psi ; \Ga ; \D \vdash P :: z{:}A & \mbox{$P$ offers session $z{:}A$
                                   when composed with processes}\\
& \mbox{\B{offering sessions specified in} $\Ga$ and $\D$ in
                                   context $\Psi$.}\\
\Psi \vdash K_1 = K_2 & \mbox{Kinds $K_1$ and $K_2$ are equal.}\\
\Psi \vdash \tau = \sigma :: K & \mbox{Types $\tau$ and $\sigma$ are
                                 equal of kind $K$.}\\
\Psi \vdash A = B :: K & \mbox{Session types $A$ and $B$ are equal of
                         kind $K$.}\\
\Psi \vdash M = N : \tau & \mbox{Terms $M$ and $N$ are equal of type $\tau$.}\\
\Psi \vdash \D = \D' :: \stype & \mbox{Contexts $\D$ and $\D'$ are
                                   equal, under the assumptions in $\Psi$.}\\
  \Psi ; \Ga ; \D \vdash P = Q :: z{:}A \quad 
  & \mbox{Processes $P$ and $Q$ are
                                        equal with typing $z{:}A$.}\\

\end{array}
  \]
\caption{Typing Judgements\label{fig:typjudg}}
\vspace{-3.5ex}
\end{figure}



\myparagraph{Typing.}
An excerpt of the typing rules for terms
and processes is given in Fig.~\ref{fig:typlam} and~\ref{fig:typpi},
respectively, noting that typing enforces types to be of
base kind $\type$ (respectively $\stype$). The rules for
dependent functions are standard, including the type conversion rule
which internalises definitional equality of types.  We highlight the
introduction rule for the monadic construct, which requires the
appropriate session types to be well-formed and the process $P$ to
offer $c{:}A$ when provided with the appropriate session contexts.

In the typing rules for processes (Fig.~\ref{fig:typpi}), presented as
a set of right and left rules (the former identifying how to
\emph{offer} a session of a given type and the latter how to use such
a session), we highlight the rules for dependently-typed
communication and monadic elimination (\B{for type-checking purposes
  we annotate constructs with the respective dependent
  type -- this is akin to functional type theories}). To offer a session $c{:}\exists x{:}\tau.A$ we send a term $M$
of type $\tau$ and then offer a session $c{:}A\{M/x\}$; dually, to
use such a session we perform an input along $c$, bound to $x$ in $Q$,
warranting a use of $c$ as a session of (open) type $A$. The rules for
the universal are dual. Offering a session  $c{:}\forall
x{:}\tau.A$ entails receiving on $c$ a term of type $\tau$ and
offering $c{:}A$. Using a session of such a type requires sending
along $c$ a term $M$ of type $\tau$, warranting the use of $c$ as
a session of type $A\{M/x\}$.

The rule for the monadic elimination form requires that the term $M$
be of the appropriate monadic type and that the provided channels
$\ov{u_j}$ and $\ov{y_i}$ adhere to the typing specified in $M$'s
type. Under these conditions, the process $Q$ may then use the session
$c$ as session $A$. The type conversion rules reflect session type
definitional equality in typing.

\begin{figure}[t]
  \[
\small
    \begin{array}{c}
\inferrule[$(\Pi I)$]
{\Psi \vdash \tau :: \type \quad \Psi , x{:}\tau \vdash M : \sigma}
  {\Psi \vdash \lambda x {:} \tau . M : \Pi x{:}\tau.\sigma}
\quad
\inferrule[$(\Pi E)$]
{\Psi \vdash M : \Pi x{:}\tau.\sigma \quad \Psi \vdash N : \tau }
  {\Psi \vdash M\, N : \sigma\{N/x\}}
      \\[1.5em]
\inferrule[$(\{\}I)$]
  {\forall i,j . \Psi \vdash A_i,B_j :: \stype 
   \quad \Psi ; \ov{u_j{:}B_j} ; \ov{d_i{:}A_i} \vdash P :: c{:}A}
  {\Psi \vdash \{c \leftarrow P \leftarrow \ov{u_j} ; \ov{d_i}\} : \{\ov{u_j{:}B_j};\ov{d_i : A_i} \vdash c{:}A\} }    
\quad
      \inferrule[$(\m{Conv})$]
{\Psi \vdash M : \tau \quad \Psi \vdash \tau = \sigma :: \type}
  {\Psi \vdash M : \sigma}
    \end{array}
  \]
\caption{Typing for Terms (Excerpt -- See Appendix~\ref{app:typterm})\label{fig:typlam}}
\vspace{-3.5ex}
\end{figure}
\begin{figure}[t]
  \[
\small
    \begin{array}{c}
\inferrule[$(\rgt{\exists})$]
  {\Psi \vdash M {:} \tau \quad 
   \Psi ; \Ga ;  \Delta \seq P :: c {:} A\{M/x\}}
  {\Psi ; \Ga ;\Delta \seq  c\langle M \rangle_{\exists x{:}\tau.A} .P :: c {:}
    \exists x{:}\tau . A}
\quad
\inferrule[$(\lft{\exists})$]
  {\Psi \vdash \tau :: \type \quad \Psi, x{:}\tau \semi \Ga ; \Delta , c{:} A\seq Q :: d {:} D }
{\Psi \semi \Ga ; \Delta , c{:}\exists x{:}\tau. A \seq c(x{:}\tau).Q :: d {:} D}\\[1.5em]

\inferrule[$(\rgt{\forall})$]
  {\Psi \vdash \tau :: \type \quad \Psi, x{:}\tau \semi \Ga ; \Delta \seq P :: c {:} A}
{\Psi ; \Ga ; \Delta \seq c(x{:}\tau).P :: c {:} \forall x{:}\tau . A}
\quad
\inferrule[$(\lft{\forall})$]
  {\Psi \vdash M {:} \tau \quad
   \Psi ; \Ga ;  \Delta , c{:}A\{M/x\} \seq Q :: d {:} D}
  {\Psi ; \Ga ;\Delta , c{:}\forall x{:}\tau . A \seq  c\langle M \rangle_{\forall x{:}\tau.A} .Q :: d {:}
   D}
      \\[1.5em]
\inferrule[($\{\}E$)]
{\D' = \ov{d_i : B_i} \quad \ov{u_j{:}C_j} \subseteq \Ga \quad
\Psi \vdash M : \{\ov{u_j{:}C_j};\ov{d_i{:}B_i} \vdash c{:}A\}
\quad \Psi ; \Ga ; \D ,c{:}A\vdash Q :: z{:}C }
{\Psi ; \Ga ; \D' , \D \vdash 
      c \leftarrow M \leftarrow \ov{u_j};\ov{y_i} ; Q :: z{:}C }\\[1.5em]

\inferrule[($\rgt{\m{Conv}}$)]
{\Psi ; \Ga ; \D \vdash P :: z{:}A \quad \Psi \vdash A = B :: \stype }
{\Psi ; \Ga ; \D \vdash P :: z{:}B}

  \quad

  \inferrule[($\lft{\m{Conv}}$)]
  {\Psi ; \Ga' ; \D' \vdash P :: z{:}A \quad \Psi ; \Ga' ; \D' = \Psi ;
  \Ga ; \D}
      {\Psi ; \Ga ; \D \vdash P :: z{:}A}\\[1.5em]
      \inferrule*[left=$(\m{cut})$]
      {\Psi ; \Ga ; \D \vdash P :: c{:}A \quad
       \Psi ; \Ga ; \D', c{:}A\vdash Q :: d{:}D}
      { \Psi ; \Ga ; \D, \D' \vdash (\nub c)(P \mid Q) :: d{:}D       }
    \end{array}
  \]
  
\caption{Typing for Processes (Excerpt -- See Appendix~\ref{app:typproc})\label{fig:typpi}}
\vspace{-3.5ex}
\end{figure}

\myparagraph{Definitional Equality.} The crux of any dependent type
theory lies in its \emph{definitional equality}. 
Type equality relies on equality of terms which, by including
the monadic construct, necessarily relies on a notion of
\emph{process} equality.

Our presentation of an intensional definitional equality of terms
follows that of \cite{DBLP:journals/tocl/HarperP05}, where we 
consider an
intrinsically typed relation, including $\beta$ and $\eta$ conversion
(similarly for type equality which includes $\beta$ and $\eta$
principles for the type-level $\lambda$-abstractions).  An excerpt of
the rules for term equality is given in Fig.~\ref{fig:termeq}. The
remaining rules are congruence rules and closure under symmetry,
reflexivity and transitivity. Rule $(\m{TMEq}\beta)$ captures the
$\beta$-reduction, identifying a $\lambda$-abstraction applied to an
argument with the substitution of the argument in the function body
(typed with the appropriately substituted type). We highlight rule
$(\m{TMEq}\{\}\eta)$, which codifies a general $\eta$-like principle
for arbitrary terms of monadic type: We form a monadic term that
applies the monadic elimination form to $M$, forwarding the result
along the appropriate channel, which becomes a term equivalent to
$M$.

Definitional equality of processes is summarised in
Fig.~\ref{fig:proceq}. We rely on process reduction defined
below. 
Definitional
equality of processes \B{consists of the usual congruence rules,}
(typed) reductions and the commutting conversions of
linear logic and $\eta$-like principles, which allows for forwarding
actions to be equated with the primitive syntactic forwarding
construct. Commutting conversions amount to sound observational
equivalences between processes \cite{DBLP:conf/esop/PerezCPT12}, given
that session composition requires name restriction (embodied by the
$(\m{cut})$ rule): In rule $(\m{PEqCC}\forall)$, either process
can only be interacted with via channel $c$ and so postponing actions
of $P$ to after the input on $c$ (when reading the equality from left
to right) cannot impact the process' observable behaviours. 
While $P$ can in general interact with sessions in $\D$ (or with $Q$),
these interactions are unobservable due to hiding in the
$(\m{cut})$ rule.

\begin{figure}[t]
\[
\small
  \begin{array}{c}
\inferrule[$(\m{TMEq}\beta)$]
{\Psi \vdash \tau :: \type \quad \Psi ,x{:}\tau\vdash M : \sigma \quad
 \Psi \vdash N : \tau}
    {\Psi \vdash (\lambda x{:} \tau . M)\,N = M\{N/x\} : \sigma\{N/x\}}
    \quad
\inferrule[$(\m{TMEq}\eta)$]
  {\Psi \vdash M : \Pi x{:}\tau.\sigma \quad x \not\in fv(M)}
  {\Psi \vdash \lambda x{:}\tau. M\, x = M : \Pi x{:}\tau.\sigma }
    \\[1.5em]
 \inferrule[$(\m{TMEq}\{\}\eta)$]
  {\Psi \vdash M : \monad{\ov{u_j{:}B_j};\ov{d_i{:}A_i} \vdash c{:}A}}
  {\Psi \vdash \monad{c \leftarrow (y\leftarrow M ;\ov{u_j};\ov{d_i} ; [y\leftrightarrow c])
  \leftarrow\ov{u_j};\ov{d_i} } = M :  \monad{\ov{u_j{:}B_j};\ov{d_i{:}A_i} \vdash c{:}A} }
  \end{array}
\]
\vspace{-2ex}
\caption{Definitional Equality of Terms (Excerpt -- See
  Appendix~\ref{app:defeqterm})\label{fig:termeq}}
\vspace{-1ex}
\end{figure}
\begin{figure}[t]
\[
\small
  \begin{array}{c}
     \inferrule*[left=$(\m{PEqRed})$]
    {\Psi ; \Ga ; \D \vdash P :: z{:}A \quad
     P \tra{} Q \quad \Psi ; \Ga ; \D \vdash Q :: z{:}A}
    {\Psi ; \Ga ; \D \vdash P = Q :: z{:}A}\\[1em]

    \inferrule*[left=$(\m{PEq}\forall\eta)$]
    { }
    {\Psi ; \Ga ; d{:}\forall x{:}\tau.A \vdash  c(x).d\langle x \rangle.[d\leftrightarrow c] = [d\leftrightarrow c] :: c{:}\forall x{:}\tau.A }\\[1em]

    \inferrule*[left=$(\m{PEqCC}\forall)$]
    {\Psi ; \Ga ; \D \vdash P :: d{:}B \quad
     \Psi , x{:}\tau ; \Ga ; \D' , d{:}B \vdash Q :: c{:} A  }
    {\Psi ; \Ga ; \D  , \D' \vdash (\nub d)(P \mid c(x).Q) =
    c(x).(\nub d)(P \mid Q) :: c{:}\forall x{:}\tau.A}

  \end{array}
\]
\vspace{-2ex}
\caption{Definitional Equality of Processes (Excerpt -- See
  Appendix~\ref{app:defeqproc})\label{fig:proceq}}
\vspace{-3ex}
\end{figure}

\vspace{-3mm}
\myparagraph{Operational Semantics.}
The operational semantics for the $\lambda$-calculus is standard,
noting that no reduction can take place inside monadic terms. 
The operational (reduction) semantics for 
processes is presented 
below 
where we omit closure under
structural congruence and the standard congruence rules 
\cite{DBLP:conf/concur/CairesP10,Toninho:2011:DST:2003476.2003499,DBLP:conf/esop/ToninhoCP13}.
The last rule defines 
spawning a process in a monadic term. 
  \[
\small
    \begin{array}{ll}
      c\langle M \rangle.P \mid c(x).Q \tra{} P \mid Q\{M/x\} &
      \ov{c}\langle x\rangle.P \mid c(x).Q \tra{} (\nub x)(P \mid Q)\\[0.5em]
      \bang c(x).P \mid \ov{c}\langle x \rangle.Q \tra{}   \bang c(x).P \mid (\nub x)(P \mid Q)\quad
      & c.\m{case}\{\ov{l_i \Rightarrow P_i}\} \mid c.l_j;Q \tra{} P_j \mid Q\,\,\,\,(l_j \in \ov{l_i})\\[0.5em]
     (\nub c)(P \mid [c\leftrightarrow d]) \tra{} P\{d/c\} &
                                                              \hspace{-1cm}c\leftarrow
                                                              \monad{c\leftarrow
                                                              P
                                                              \leftarrow
                                                              \ov{u_j};\ov{d_i}}
                                                              \leftarrow \ov{u_j};\ov{d_i}
                                                              ; Q \tra{} (\nub c)(P \mid Q)
    \end{array}
  \]

\subsection{Example -- Reasoning about Processes using Dependent Types}
\label{sec:exreasoning}

The use of type indices
(i.e. type families) in dependently typed frameworks adds information to types
to produce more refined specifications.
Our framework enables us to do this at the level of session types.

Consider a session type that ``counts down'' on a natural number (we
assume inductive definitions and dependent pattern matching in the
style of \cite{norell:thesis}):
\[
  \begin{array}{lcl}
    \m{countDown} & :: & \Pi x{:}\m{Nat}.\stype \\
    \m{countDown}\,(\m{succ}(n)) &  = & 
 \exists y{:}\m{Nat}.\m{countDown}(n)\\
    \m{countDown}\,\,\,\m{z} & = &
          \one                     
  \end{array}
\]
The type family $\m{countDown}(n)$ denotes a session type that emits
exactly $n$ numbers and then terminates. We can now write a
(dependently-typed) function that produces processes with the
appropriate type, given a starting value: 


\[
  \begin{array}{lcl}
    \m{counter} & : & \Pi x{:}\m{Nat}.\monad{\m{countDown}(x)}\\
    \m{counter}\,\,\,(\m{succ}(n)) & = & \monad{c \leftarrow c\langle \m{succ}(n) \rangle.
                \, d \leftarrow \m{counter}(n) ;
                 [d\leftrightarrow c]}\\
    \m{counter}\,\,\,\m{z}     & = & \monad{c\leftarrow \zero}
  \end{array}
\]

Note how the type of $\m{counter}$, through the type family
$\m{countDown}$, allows us to specify exactly the number of times a
value is sent. This is in sharp contrast with existing recursive (or
inductive/coinductive
\cite{DBLP:conf/icfp/LindleyM16,DBLP:conf/tgc/ToninhoCP14}) session
types, where one may only specify the general iterative nature of the
behaviour (e.g. ``send a number and then recurse or terminate'').



The example above relies on session type indexing in order to provide
additional static guarantees about processes (and the functions that
generate them). An alternative way 
is to consider ``simply-typed'' programs
and then \emph{prove} that they satisfy the desired properties, using
the language itself.
%
Consider a
simply-typed version of the counter above described as an
inductive session type:
\[
\begin{array}{lcl}
  \m{simpleCounterT} & :: & \stype\\
  \m{simpleCounterT} & = &  \oplus\{\m{dec} : \m{Nat}\wedge \m{simpleCounterT} ,
                        \m{done} : \one\}
\end{array}
\]
There are many processes that correctly implement such a
type, given that the type merely dictates that the session outputs a
natural number and recurses (modulo the $\m{dec}$ and $\m{done}$
messages to signal which branch of the internal choice is taken).
A function that produces processes implementing such a session, mirroring those generated
by the $\m{counter}$ function above, is:
\[
  \begin{array}{lcl}
    \m{simpleCounter} & : & \m{Nat}\rightarrow \monad{\m{simpleCounterT}}\\
    \m{simpleCounter}\,\,\,(\m{succ}(n)) & = & \monad{c\leftarrow c.\m{dec};
                                               (\nub d)(d\langle  \m{succ}(n) \rangle.\zero \mid
                                               d(x).
                                               c\langle x\rangle.\\
    && \,\,\, d\leftarrow \m{simpleCounter}(n); [d\leftrightarrow c]     } \\
    \m{simpleCounter} \quad \m{z} & = & \monad{c\leftarrow c.\m{done};\zero} \\                           
                          
  \end{array}
\]
The process generated by $\m{simpleCounter}$, after emiting the
$\m{dec}$ label, spawns a process in parallel that sends the
appropriate number, which is received by the parallel thread and then
sent along the session $c$. 
Despite its simplicity, 
this example embodies a general pattern
where a computation is spawned in parallel (itself potentially
spawning many other threads) and the main thread then waits for
the result before proceeding.

While such a process is typable in most session typing frameworks, our
theory enables us to \emph{prove} that the counter
implementation above indeed counts down from a given number by
defining an appropriate (inductive) type family, indexed by \emph{monadic} values
(i.e. processes):
\[
  \begin{array}{lcl}
    \m{corrCount} & ::  & \Pi x{:}\m{Nat}.\Pi y{:}\monad{\m{simpleCounterT}}.\type\\
    \m{corr}_z & : & \m{corrCount}\,\m{z}\,\monad{c\leftarrow c.\m{done};\zero}\\
    \m{corr}_n & : & \Pi n {:}\m{Nat}.\Pi P{:}\monad{\m{simpleCounterT}}.
                     \m{corrCount}\,n\,P \rightarrow\\
                  & & \m{corrCount}\,(\m{succ}(n))\,\monad{c\leftarrow c.\m{dec};
                      c\langle \m{succ}(n) \rangle.d\leftarrow P ; [d\leftrightarrow c]}    
  \end{array}
\]
The type family $\m{corrCount}$, indexed by a natural number and a
monadic value implementing the session type $\m{simpleCounter}$, is
defined via two constructors: $\m{corr}_z$, which specifies that a correct
$0$ counter emits the $\m{done}$ label and terminates; and
$\m{corr}_n$, which given a monadic value $P$ that is a correct
$n$-counter, defines that a correct $(n+1)$-counter emits $n+1$ and then
proceeds as $P$ (modulo the label emission bookkeeping).

The proof of correctness of the $\m{simpleCounter}$ function above is
no more than a function of type $\Pi n{:}\m{Nat}.\m{corrCount}\,n$
$(\m{simpleCounter}(n))$, defined below:
\[
  \begin{array}{lcl}
    \m{prf} & : & \Pi n{:}\m{Nat}.\m{corrCount}\,n \,(\m{simpleCounter}(n))\\
    \m{prf}\quad\m{z} & = & \m{corr}_z\\
    \m{prf}\quad(\m{succ}(n)) & = & \m{corr}_n\,n\,(\m{simpleCounter}(n))\,(\m{prf}\,n)
  \end{array}
\]
\noindent Note that in this scenario, the processes that index the
$\m{corrCount}$ type family are not syntactically equal to those
generated by $\m{simpleCounter}$, but rather \emph{definitionally} equal.

Typically, the processes that index such correctness
specifications tend to be distilled versions of the actual
implementations, which often perform some additional internal
computation or communication steps. Since our notion of
definitional equality of processes includes reduction (and also
commuting conversions which account for type-preserving shuffling of
internal communication actions \cite{toninhothesis}), the type conversion mechanism allows
us to use the techniques described above to generally reason about
specification conformance.

We may also consider a variant of the example above which does
not force outputs to match precisely with the type index:
\[
  \begin{array}{lcl}
    \m{countDown}' & :: & \Pi x{:}\m{Nat}.\stype \\
    \m{countDown}'\,(\m{succ}(n)) &  = & 
 \exists y{:}\m{Nat}.\m{countDown}'(n)\\
    \m{countDown}'\,\,\,\m{z} & = &
          \one                     
  \end{array}
\]
The type $\m{countDown}'\, n$ will still require $n$ outputs to be
performed, but unlike with $\m{countDown}$ we do not enforce a
relation between the iteration and the number being sent.
An implementation of such a type is given below, using fundamentally the
same code as for $\m{counter}$:
\[
  \begin{array}{lcl}
    \m{counter}' & : & \Pi x{:}\m{Nat}.\monad{\m{countDown}'(x)}\\
    \m{counter}'\,\,\,(\m{succ}(n)) & = & \monad{c \leftarrow c\langle \m{succ}(n) \rangle.\\
                && \qquad\,\,\, d \leftarrow \m{counter}'(n) ;\\
                && \qquad\,\,\, [d\leftrightarrow c]}\\
    \m{counter}'\,\,\,\m{z}     & = & \monad{c\leftarrow \zero}
  \end{array}
\]

We may then use an heterogeneous equality
(a special case of the so-called \emph{John Major equality}
\cite{DBLP:conf/types/McBride00}) of the form
\[
  \begin{array}{lcl}
    \m{JMEq} & :: & \Pi A {:}\m{stype} . \Pi B {:}\m{stype}.
                       \Pi x{:}\monad{A}.\Pi y{:}\monad{B}.\type\\
    \m{JMEqRefl} & : & \lambda A {:}\m{stype}.\lambda x{:}\monad{A}.
                        \m{JMEq}\,A\,A\,x\,x
  \end{array}
\]
to inductively show that the processes produced by $\m{counter}$ and
$\m{counter}'$ are indeed the same.
\[
  \begin{array}{l}
    \m{eqs}  :  \Pi n{:}\m{Nat}.\m{JMEq}\,(\m{countDown}(n))\,(\m{countDown}'(n))\,
                  (\m{counter}(n))\,(\m{counter}'(n))\\
    \m{eqs}\,z  =  \m{JMEqRefl}\,\one\,\monad{c\leftarrow \zero}\\
    \m{eqs}\,(\m{succ}(n))  =  \m{case}\,(\m{eqs}\,n)\,\m{of}\,\{\_ \Rightarrow
    \m{JMEqRefl}\,(\m{countDown}(\m{succ}(n)))\\
    \qquad\qquad\qquad\qquad\qquad
    \qquad\qquad\qquad\qquad\quad\,\,\,(\m{counter}(\m{succ}(n))))\}\\
  \end{array}
\]
We note that the example above makes extensive use of dependent
pattern matching, using some implicit assumptions on its behaviour
that have not been formalised in this paper and are left for future work.



\subsection{Type Soundness of the Framework}
\label{sec:meta}
The main goal of this section is to present type soundness of our
framework through a subject reduction result. We also show that our
theory guarantees progress for terms and processes.
The
development requires a series of auxiliary results (detailed in
Appendix~\ref{app:soundness}) pertaining to the functional and process
layers which are ultimately needed to produce the inversion properties
necessary to establish subject reduction.  We note that strong
normalisation results for linear-logic based session processes are
known in the literature
\cite{DBLP:conf/esop/CairesPPT13,DBLP:conf/tgc/ToninhoCP14,toninhothesis},
even in the presence of impredicative polymorphism, restricted
corecursion and higher-order data. Such results are directly
applicable to our work using appropriate semantics preserving type erasures.

In the remainder we
often write $\Psi \vdash \mathcal{J}$ to stand for a well-formedness,
typing or definitional equality judgment of the appropriate form.
Similarly for $\Psi ; \Ga ; \D \vdash
\mathcal{J}$.
We begin with the substitution property, which naturally holds for
both layers, noting that the dependently typed nature of the framework
requires substitution in both contexts, terms and in types.

\begin{restatable}[Substitution]{lemma}{substlem}
  \label{lem:trans}
  Let $\Psi \vdash M :\tau$:
  \begin{enumerate}
  \item If $\Psi , x{:}\tau , \Psi' \vdash
    \mathcal{J}$ then $\Psi , \Psi'\{M/x\} \vdash \mathcal{J}\{M/x\}$;
  \item If $\Psi , x{:}\tau , \Psi' ; \Ga ;
    \D \vdash \mathcal{J}$ then $\Psi , \Psi'\{M/x\} ; \Ga \{M/x\} ; \D \{M/x\}\vdash \mathcal{J}\{M/x\}$
  \end{enumerate}
\end{restatable}

\noindent Combining substitution with a form of functionality for typing (i.e.
that substitution of equal terms in a well-typed term produces equal
terms) and for equality (i.e. that substitution of equal terms in a
definitional equality proof produces equal terms), we can establish
validity for typing and equality, which is a form of internal
soundness of the type theory stating that judgments are consistent
across the different levels of the theory.

\begin{lemma}[Validity for Typing]
{\em (1)} 
If $\Psi \vdash \tau :: K$ or $\Psi \vdash A :: K$ then $\Psi \vdash
K$; 
{\em (2)} 
If $\Psi \vdash M : \tau$ then $\Psi \vdash \tau :: \type$; and 
{\em (3)} 
If $\Psi ; \Ga ; \D \vdash P :: z{:}A$ then $\Psi \vdash A :: \stype$.
\end{lemma}

\begin{restatable}[Validity for Equality]{lemma}{valeqlem}
\begin{enumerate}
\item If $\Psi \vdash M = N : \tau$ then $\Psi \vdash M : \tau$, $\Psi
  \vdash N : \tau$ and $\Psi \vdash \tau :: \type$
\item If $\Psi \vdash \tau = \sigma :: K$ then $\Psi \vdash \tau ::
  K$, $\Psi \vdash \sigma :: K$ and $\Psi \vdash K$
\item If $\Psi \vdash A = B :: K$ then $\Psi \vdash A :: K$, $\Psi
  \vdash B :: K$ and $\Psi \vdash K$
\item If $\Psi \vdash K = K'$ then $\Psi \vdash K$ and $\Psi \vdash
  K'$
\item If $\Psi ; \Ga ; \D \vdash P = Q :: z{:}A$ then $\Psi ; \Ga ; \D
  \vdash P :: z{:}A$, $\Psi ; \Ga ; \D  \vdash Q:: z{:}A$ and $\Psi
  \vdash A :: \stype$
\end{enumerate}
\end{restatable}
With these results we establish the appropriate inversion and
injectivity properties which then enable us to show unicity of types
(and kinds).

\begin{restatable}[Unicity of Types and Kinds]{theorem}{unicitythm}
  \begin{enumerate}
\item If $\Psi \vdash M : \tau$ and $\Psi \vdash M : \tau'$ then $\Psi
  \vdash \tau = \tau' :: \type$
\item If $\Psi \vdash \tau :: K$ and $\Psi \vdash \tau :: K'$ then
  $\Psi \vdash K = K'$
\item If $\Psi; \Ga ; \D \vdash P :: z{:}A$ and $\Psi ; \Ga ; \D
  \vdash P :: z{:}A'$ then $\Psi \vdash A = A' :: \stype$
\item If $\Psi \vdash A :: K$ and $\Psi \vdash A :: K'$ then $\Psi
  \vdash K = K'$
  \end{enumerate}
\end{restatable}
All the results above, combined with the process-level properties
established in
\cite{depsesstr,toninhothesis,DBLP:journals/mscs/CairesPT16} enable us
to show the following:

\begin{restatable}[Subject Reduction -- Terms]{theorem}{srlamterms}
\label{thm:srlammain}
If $\Psi \vdash M : \tau$ and $M\tra{} M'$ then $\Psi \vdash M' : \tau$
\end{restatable}
\vspace{-1.5ex}
\begin{restatable}[Subject Reduction -- Processes]{theorem}{srlamproc}
\label{thm:srprocs}
If $\Psi ; \Ga ; \D \vdash P :: z{:}A$ and $P \tra{} P'$ then $\exists
Q$ such that $P' \equiv Q$ and $\Psi ; \Ga ; \D \vdash Q :: z{:}A$
\end{restatable}
\vspace{-1.5ex}
\begin{restatable}[Progress -- Terms]{theorem}{proglamthm}
\label{thm:proglam}
If $\Psi \vdash M : \tau$ then either $M$ is a value or $M \tra{} M'$
\end{restatable}

As common in logical-based session type theories, typing
enforces a strong notion of \emph{global} progress which states that
closed processes that are waiting to perform communication actions
cannot get stuck (this relies on a notion of \emph{live} process,
defined as $\m{live}(P)$ iff $P \equiv (\nub \tilde{n})(\pi.Q \mid R)$
for some process $R$, sequence of names $\tilde{n}$ and a
non-replicated guarded process $\pi.Q$). We note that the restricted
typing for $P$ is without loss of generality, due to the $(\m{cut})$ rule.

\begin{theorem}[Progress -- Processes]
If $\Psi;\cdot;\cdot \vdash P :: c{:}\one$ and $\m{live}(P)$ then
$\exists Q$ such that $P \tra{} Q$
\end{theorem}


\section{Embedding the Functional Layer in the Process Layer}
\label{sec:embed}
Having introduced our type theory and showcased some of its informal
expressiveness in terms of the ability to specify and
\emph{statically} verify true data dependent protocols, as well as the
ability to prove properties of processes, we now develop a formal
expressiveness result for our theory, showing that the process level
type constructs are able to encode the dependently-typed
functional layer, faithfully preserving type dependencies.

Specifically, we show that (1) the type-level constructs in the
functional layer can be represented by those in the process layer
combined with the contextual monad type, and (2) all term level
constructs can be represented by session-typed processes that exchange
monadic values. Thus, we show that both $\lambda$-abstraction and
application can be eliminated while still preserving non-trivial type
dependencies. 
Crucially, we note that the monadic construct \emph{cannot}
be fully eliminated due to the cross-layer nature of session type
dependencies: In the process layer, simply-kinded dependent types
(i.e. types with kind $\stype$) are of the form $\forall x{:}\tau.A$
where $\tau$ is of kind $\type$ and $A$ of kind $\stype$ (where $x$ may
occur). Operationally, such a session denotes an input of some term $M$ of
type $\tau$ with a continuation of type $A\{M/x\}$. Thus, to
faithfully encode type dependencies we cannot represent such a type
with a non-dependently typed input (e.g. a type of the form $A \lolli
B$).

\subsection{The Embedding}
\myparagraph{A first attempt.}
Given the observation above, a seemingly reasonable option would be to attempt
an encoding that maintains monadic objects solely at the level of
type indices and then exploits Girard's encoding
\cite{DBLP:journals/tcs/Girard87} of function types
$\tau \rightarrow \sigma$ as $\bang\lb\tau\rb \rightarrow
\lb\sigma\rb$, which is adequate for session-typed processes
\cite{DBLP:conf/fossacs/ToninhoCP12}. Thus a candidate encoding for
the type $\Pi x{:}\tau.\sigma$ would be $\forall x{:}\{\lb\tau\rb\}.
                                               \bang \lb\tau\rb \lolli
                                               \lb \sigma\rb$,
where $\lb{-}\rb$ denotes our encoding on types. If we then consider
the encoding at the level of terms, typing dictates the following (we
write $\lb M \rb_z$ for the process encoding of $M : \tau$, where $z$
is the session channel along which one may observe the ``result'' of
the encoding, typed with $\lb \tau\rb$):
\[
  \begin{array}{lcl}
     \lb \lambda x{:}\tau . M \rb_z & \triangleq & z(x).z(x').\lb M\rb_z \\
 \lb M\, N\rb_z & \triangleq & 
(\nub x)(\lb M \rb_x \mid x\langle
                              \{\lb N\rb_y\}\rangle.\ov{x}\langle x'\rangle.(\bang x'(y).\lb N\rb_y \mid [x\leftrightarrow z])\\
  \end{array}
\]

However, this candidate encoding breaks down once we consider
definitional equality. Specifically, compositionality (i.e. the
relationship between $\lb M\{N/x\}\rb_z$ and the encoding of $N$
substituted in that of $M$) requires us to relate $\lb M\{N/x\}\rb_z$
with
$(\nub x)(\lb M \rb_z\{\monad{\lb N\rb_y}/x\} \mid \bang x'(y).\lb
N\rb_y)$, which relies on reasoning up-to \emph{observational
  equivalence} of processes, a much stronger relation than our notion
of definitional equality. Therefore it is \emph{fundamentally}
impossible for such an encoding to preserve our definitional equality,
and thus it cannot preserve typing in the general case.

\myparagraph{A faithful embedding.}  We now develop our embedding of
the functional layer into the process layer 
which is compatible with definitional equality. 
Our target calculus is reminiscent of a higher-order (in
the sense of higher-order processes \cite{sangiorgipi}) session
calculus \cite{Mostrous07twosession}. Our encoding $\lb{-}\rb$ is
inductively defined on kinds, types, session types, terms and
processes. 
As usual in process encodings of
the $\lambda$-calculus, the encoding of a term $M$ is indexed by a
result channel $z$, written $\lb M \rb_z$, where the behaviour of $M$
may be observed. 

\begin{figure}[t]
\small

$\begin{array}{lcllcl}
\text{\bf Kind:}\\
\  \lb \type \rb & \triangleq & \stype & \lb \stype \rb & \triangleq & \stype\\
\  \lb \Pi x{:}\tau.K\rb & \triangleq & \Pi x{:}\monad{\lb \tau\rb}.\lb
  K\rb 
\quad\quad&
  \lb \Pi t :: K_1 . K_2 \rb & \triangleq & \Pi t {::} \lb K_1 \rb.\lb
  K_2\rb\\
\text{\bf Functional:}\\
\   \lb \Pi x{:}\tau.\sigma \rb & \triangleq & \forall x{:}\monad{\lb \tau \rb}.\lb \sigma \rb &
  \lb \monad{\ov{u_j {:}B_j} ; \ov{d_i {:}B_i} \vdash c{:}A} \rb & \triangleq &                  \ov{\bang \lb B_j\rb } \lolli \ov{\lb B_i \rb} \lolli \lb A \rb\\
\   \lb \lambda x{:}\tau.\sigma \rb & \triangleq & \lambda x{:} \monad{\lb \tau\rb}.\lb \sigma\rb
                                      & \lb \tau\, M\rb & \triangleq & \lb \tau\rb \,\monad{\lb M \rb_c}\\
\   \lb \lambda t {::} K . \tau \rb & \triangleq & \lambda t {::} \lb K \rb . \lb \tau \rb
                                      & \lb \tau \, \sigma \rb & \triangleq & \lb \tau \rb \, \lb \sigma \rb\\
\text{\bf Session:}\\
\  \lb \forall x {:}\tau. A \rb & \triangleq & \forall x{:}\monad{\lb \tau\rb} .\lb A \rb
  & \lb \exists x{:}\tau . A \rb & \triangleq & \exists x{:}\monad{\lb \tau \rb}.\lb A \rb\\
\  \lb \lambda x{:}\tau . A \rb & \triangleq & \lambda x{:}\monad{\lb \tau \rb}.\lb A \rb &
  \lb A \, M \rb & \triangleq & \lb A \rb \, \monad{\lb M \rb_c}                           
\end{array}
$
\\[1mm]
$
\begin{array}{l} 
\text{\bf Terms:}\\
\ \lb \lambda x{:}\tau . M \rb_z  \triangleq 
 z(x{:}\monad{\lb\tau\rb}).\lb M\rb_z  
\quad\quad
 \lb M\, N\rb_z  \triangleq  
(\nub x)(\lb M \rb_x \mid x\langle
                              \{\lb N\rb_y\}\rangle.[x\leftrightarrow z])\\

\ \lb x \rb_z  \triangleq  y \leftarrow x ; [y \leftrightarrow z]\quad\quad
\lb \{ z \leftarrow P \leftarrow \ov{u_j};\ov{d_i}\}\rb_z  \triangleq 
z(u_0).\dots.z(u_j).z(d_0). \dots . z(d_n).\lb P\rb
\end{array}
$
\\[1mm]
$
\begin{array}{lcl}
{\text{\bf Processes:}}\\
\ \lb(\nub x)(P \mid Q)\rb & \triangleq & (\nub x)(\lb P \rb \mid \lb Q \rb) 
\quad \lb \zero \rb  \triangleq  \zero\quad 
\lb \ov{x}\langle y \rangle.(P \mid Q) \rb  \triangleq  \ov{x}\langle
y \rangle.(\lb P \rb \mid \lb Q \rb)\\ 
\ \lb x\langle M \rangle.P \rb  & \triangleq  & x\langle \{\lb M\rb_y
\}\rangle.\lb P \rb \quad 
\lb x(y).P \rb  \triangleq  x(y).\lb P \rb\\
\ \lb c \leftarrow M \leftarrow \ov{u_j};\ov{y_i} ; Q \rb  & \triangleq &   
                                                                        (\nub c)(\lb M \rb_c \mid \ov{c}\langle v_1\rangle.(\ov{u_1}\langle a_1\rangle.[a_1\leftrightarrow v_1] \mid
\dots \mid        \\                                      
                    &&                                                    \ov{c}\langle d_1 \rangle.([y_1\leftrightarrow
      d_1] \mid \dots \mid \ov{c}\langle d_n\rangle.([y_n
                       \leftrightarrow d_n] \mid \lb Q\rb) \dots )
\end{array}
$
\caption{An embedding of dependent functions into processes\label{fig:enc}}
\end{figure}

The embedding is presented in Fig.~\ref{fig:enc}, noting that the
encoding extends straightforwardly to typing contexts, where
functional contexts $\Psi, x{:}\tau$ are mapped to
$\monad{\lb\Psi\rb},x{:}\monad{\lb\tau\rb}$.
The mapping of
base kinds is straightforward. Dependent kinds $\Pi x{:}\tau.K$ rely
on the monad for well-formedness and are encoded as (session) kinds
of the form $\Pi x{:}\monad{\lb\tau\rb}.\lb K \rb$. The higher-kinded
types in the functional layer are translated to the corresponding
type-level constructs of the process layer where all objects that must
be $\type$-kinded rely on the monad to satisfy this constraint. For
instance, $\lambda x{:}\tau.\sigma$ is mapped to the session-type
abstraction $\lambda x{:}\monad{\lb \tau\rb}.\lb \sigma\rb$ and the
type-level application $\tau\,M$ is translated to $\lb \tau\rb \,
\monad{\lb M\rb_c}$. Given the observation above on embedding the
dependent function type $\Pi x{:}\tau.\sigma$, we translate it
directly to $\forall x{:}\monad{\lb\tau\rb}.\lb\sigma\rb$, that is,
functions from $\tau$ to $\sigma$ are mapped to sessions that input
\emph{processes} implementing $\lb\tau\rb$ and then behave as $\lb\sigma\rb$
accordingly. The encoding for monadic types simply realises the
contextual nature of the monad by performing a sequence of inputs of
the appropriate types (with the shared sessions being of $\bang$
type).

The mutually dependent nature of the framework requires us to extend
the mapping to the process layer. Session types are mapped
homomorphically 
(e.g.~$\lb A \lolli B \rb  \triangleq  \lb A \rb \lolli \lb B \rb$) 
with the exception of dependent inputs and outputs
which rely on the monad, similarly for type-level functions
and application.

The encoding of $\lambda$-terms is guided by the embedding for types:
the abstraction $\lambda x{:}\tau.M$ is mapped to an input of a term
of type $\monad{\lb\tau\rb}$ with continuation $\lb M\rb_z$;
application $M\,N$ is mapped to the composition of the encoding of $M$
on a fresh name $x$ with the corresponding output of $\monad{\lb N
  \rb_y}$, which is then forwarded to the result channel $z$; monadic
expressions are translated to the appropriate sequence of inputs, as
dictated by the translation of the monadic type; and, the translation
of variables makes use of the monadic elimination form (since
the encoding enforces variables to always be of monadic type) combined
with forwarding to the appropriate result channel.

The mapping for processes is mostly homomorphic, using the monad
constructor as needed. The only significant exception is the encoding
for monadic elimination which must provide the encoded
monadic term $\lb M\rb_c$ with the necessary channels. Since the
session calculus does not support communication of free names this is
achieved by a sequence of outputs of fresh names combined with
forwarding of the appropriate channel. To account for replicated
sessions we must first trigger the replication via an output which is
then forwarded accordingly.


We can illustrate our encoding via a simple example of an encoded
function (we omit type annotations for conciseness):
\[
  \begin{array}{l}
    \lb (\lambda x.x)\,(\lambda x.\lambda y.y)\rb_z  =  (\nub c)(\lb \lambda x.x \rb_c \mid
       c\langle\monad{ \lb \lambda x.\lambda y.y \rb_w } \rangle.[c\leftrightarrow z]) =\\
    
       \qquad(\nub c)(c(x).y\leftarrow x;[y\leftrightarrow c] \mid
    c\langle\monad{w(x).w(y).d\leftarrow y;[d\leftrightarrow w]}\rangle.[c\leftrightarrow z])\\

  \tra{}^+ z(x).z(y).d\leftarrow y;[d\leftrightarrow z]
\ = \ \lb \lambda x.\lambda y.y\rb_z
  \end{array}
\]

\subsection{Properties of the Embedding}
We now state the key properties satisfied by our embedding, ultimately
resulting in type preservation and operational correspondence.
For conciseness, in the statements below we list only the cases for terms and
processes, omitting those for types and kinds (see Appendix~\ref{app:enc}).
The key property that is needed is a notion of compositionality, which
unlike in the sketch above no longer falls outside of definitional
equality.

\begin{restatable}[Compositionality]{lemma}{enccompos}
  \label{lem:comp}
  \begin{enumerate}
\item {\small$\Psi ; \Ga ; \D  \vdash \lb M \{N/x\}\rb_z = \lb M\rb_z\{\monad{\lb N \rb_y}/x\} :: z{:}\lb
    A\{N/x\}\rb$}
  \item {\small $\Psi ; \Ga ; \D \vdash \lb P\{M/x\}\rb :: z{:}\lb A
    \{M/x\}\rb$
    iff
    $\Psi ; \Ga ; \D \vdash \lb P \rb\{\monad{\lb M \rb_c}/x\}::
     z{:}\lb A \rb\{\monad{\lb M \rb_c}/x\}$}
\end{enumerate}
\end{restatable}






Given the dependently typed nature of the framework, establishing the
key properties of the encoding must be done simultaneously (relying on some auxiliary results -- see
Appendix~\ref{app:enc}).

\begin{restatable}[Preservation of Equality]{theorem}{encpreseq}
  \begin{enumerate}
\item If $\Psi \vdash M = N : \tau$ then $\monad{\lb \Psi \rb} ; \cdot
  ; \cdot \vdash \lb M \rb_z = \lb N \rb_z :: z{:}\lb \tau\rb$
\item If $\Psi ; \Ga ; \D \vdash P = Q :: z{:}A$ then $\monad{\lb \Psi
    \rb} ; \lb \Ga \rb ; \lb \D \rb\vdash \lb P \rb = \lb Q\rb ::
  z{:}\lb A \rb$
\end{enumerate}
\end{restatable}
\vspace{-3mm}
\begin{restatable}[Preservation of Typing]{theorem}{encprestyp}
  \label{lem:pwf}
  \begin{enumerate}
\item If $\Psi \vdash M : \tau$ then $\{\lb \Psi \rb\} ; \cdot
  ; \cdot \vdash \lb M \rb_z :: z{:}\lb \tau\rb$
\item If $\Psi ; \Ga ; \D \vdash P :: z{:}A$ then
 $\{\lb\Psi\rb\} ; \lb\Ga\rb  ; \lb\D\rb \vdash \lb P\rb :: z{:}\lb
 A \rb$
\end{enumerate}
\end{restatable}

\begin{theorem}[Operational Correspondence]
  \label{thm:opcorrcomp}
  If $\Psi ; \Ga ; \D \vdash P :: z{:}A$ and $\Psi \vdash M : \tau$ then:
  \begin{enumerate}
  \item (a) If $P \tra{} P'$ then
  $\lb P \rb \tra{}^+ Q$ with $\monad{\lb\Psi\rb} ; \lb\Ga\rb ;
  \lb\D\rb \vdash Q = \lb P'\rb :: z{:}\lb A\rb$ and
  (b) if $\lb P\rb \tra{} P'$ then $P \tra{}^+ Q$ with $\monad{\lb\Psi\rb} ; \lb\Ga\rb ;
    \lb\D\rb \vdash P' = \lb Q\rb :: z{:}\lb A\rb$
  \item (a) If $M \tra{} M'$ then $\lb M \rb_z
  \tra{}^+ N$ with $\monad{\lb \Psi \rb} ; \cdot ; \cdot \vdash N = \lb
  M '\rb_z ::z{:}\lb \tau\rb$ and
  (b) if $\lb M \rb_z \tra{} P$ then $M \tra{} N$ with $\monad{\lb\Psi\rb} ;\cdot;\cdot\vdash
      \lb N \rb_z = P ::z{:}\lb\tau\rb $
  \end{enumerate}
\end{theorem}

In Theorem~\ref{thm:opcorrcomp}, (a) is commonly
referred to as operational completeness, with (b) establishing
soundness. As exemplified above, our encoding satisfies a very
precise operational correspondence with the original $\lambda$-terms.




\section{Related and Future Work} 
\label{sec:conc}


\vspace{-2mm}
\myparagraph{Enriching Session Types via Type Structure.}
%
Exploiting the linear logical foundations of session types,
\cite{Toninho:2011:DST:2003476.2003499} considers a form of value
dependencies where session types can state properties of exchanged
data values, while the work \cite{DBLP:conf/esop/ToninhoCP13}
introduces the contextual monad in a simply-typed setting. Our
development not only subsumes these two works, but goes beyond simple
value dependencies by extending to a richer type structure and
integrating dependencies with the contextual monad. Recently,
\cite{DBLP:journals/pacmpl/BalzerP17} considers a non-conservative
extension of linear logic-based session types with sharing, allowing
true non-determinism.  Their work includes dependent quantifications
with shared channels, but their type syntax does {\em not} include
free type variables, so the actual type dependencies do not arise (see
\cite[37:8]{DBLP:journals/pacmpl/BalzerP17}). Thus none of the examples in
this paper can be represented in
\cite{DBLP:journals/pacmpl/BalzerP17}.  The work
\cite{DBLP:journals/pacmpl/IgarashiTVW17} studies gradual session
types. 
To the best of our
knowledge, the main example in \cite[\S~2]{DBLP:journals/pacmpl/IgarashiTVW17}
is \emph{statically} representable in our framework as in the example of
\S~\ref{sec:intro}, where protocol actions depend on values that are
communicated (or passed as function arguments).

In the context of multiparty session types, the theory of multiparty indexed
session types is studied in \cite{DYBH12}, and implemented
in a protocol description language \cite{NY15}.
The main aim of these works
is to use indexed  types to represent an arbitrary number of
session {\em participants}. The work \cite{TY17} extends
\cite{Toninho:2011:DST:2003476.2003499} to multiparty sessions in
order to treat value dependency across multiple
participants. Extending our framework to multiparty
\cite{DBLP:conf/popl/HondaYC08} or non-logic based session types
\cite{honda.vasconcelos.kubo:language-primitives} is an interesting
future topic.

\myparagraph{Combining Linear and Dependent Types.}  
Many works 
have
studied the various challenges of integrating linearity in dependent
functional type theories. 
We focus on the most closely
related works. 
The work 
\cite{DBLP:journals/iandc/CervesatoP02} introduced the Linear Logical
Framework (LLF), integrating linearity with the LF
\cite{DBLP:journals/jacm/HarperHP93} type theory, which was later
extended to the Concurrent Logical Framework
(CLF) \cite{DBLP:conf/types/WatkinsCPW03}, accounting for further
linear connectives.
Their theory is representable in our
framework through the contextual monad (encompassing full
intuitionistic linear logic), depending on
linearly-typed processes that can express dependently typed functions
(\S~\ref{sec:embed}).

The work of \cite{DBLP:conf/popl/KrishnaswamiPB15} integrates
linearity with type dependencies by extending LNL
\cite{DBLP:conf/csl/Benton94}. Their work is aimed at reasoning about
imperative programs using a form of Hoare triples, requiring
features that we do not study in this work such has proof irrelevance
and computationally irrelevant quantification. Formally, their type
theory is extensional which introduces significant technical
differences from our intensional type theory, such as a realisability
model in the style of NuPRL \cite{HARPER199271} to establish consistency.

Recently, \cite{DBLP:conf/esop/GeorgesMOP17} proposed an extension of
LLF with first-class contexts (which may contain both linear and
unrestricted hypotheses). While the contextual aspects of their theory
are reminiscent of our contextual monad, their framework differs
significantly from ours, since it is designed to enable higher-order
abstract syntax (commonplace in the LF family of type theories),
focusing on a type system for canonical LF objects with a
meta-language that includes contexts and context manipulation. They do
not consider additives since their integration with first-class
contexts can break canonicity.  

While none of the above works considers processes as primitive,
their techniques should be useful for, e.g.  developing
algorithmic type-checking and integrating inductive and coinductive
session types based on
\cite{toninhothesis,DBLP:conf/tgc/ToninhoCP14,DBLP:conf/icfp/LindleyM16}.


\myparagraph{Dependent Types and Higher-Order $\pi$-calculus.}

The work \cite{YH02} studies a form of dependent types where the type
of processes takes the form of a mapping $\Delta$ from channels $x$ to
channel types $T$ representing an interface of process $P$.  The
dependency is specified as $\Pi(x{:}T)\Delta$, representing a
channel abstraction of the environment. This notion is extended to an
existential channel dependency type $\Sigma(x{:}T)\Delta$ to address
fresh name creation \cite{Y04,HRY05}. Combining our process monad
with dependent types can be regarded as an ``interface'' which
describes explicit channel usages for processes. The main differences
are (1) our dependent types are more general, treating full dependent
families including terms and processes in types, while
\cite{YH02,Y04,HRY05} study only channel dependency to environments
(i.e. neither terms nor processes appear in types, only channels); and
(2) our calculus emits only fresh names, not needing to
handle the complex scoping mechanism treated in \cite{Y04,HRY05}. In
this sense, the process monad provides an elegant framework to
handle higher-order computations and assign non-trivial types to
processes.\\
\B{{\bf Acknowledgements.} 
The authors would like to thank the anonymous reviews for their
comments and suggestions. This work is partially supported by EPSRC EP/K034413/1,
EP/K011715/1, EP/L00058X/1, EP/N027833/1 and EP/N028201/1.}



\bibliographystyle{splncs03}

\newpage
\appendix

\section{Appendix -- Dependently-typed Calculus}

\subsection{Complete Rules for Dependently-Typed System}
\label{app:rules}

We recall the meaning of the several judgments of our type theory:
\[
\begin{array}{ll}
\Psi \vdash  & \mbox{Context $\Psi$ is well-formed.}\\
\Psi \vdash K & \mbox{$K$ is a kind in context $\Psi$.}\\
\Psi \vdash \tau :: K & \mbox{$\tau$ is a (functional) type of kind
                        $K$ in context $\Psi$.}\\
\Psi \vdash A :: K & \mbox{$A$ is a session type of kind $K$ in context $\Psi$.}\\
\Psi \vdash M : \tau & \mbox{$M$ has type $\tau$ in context $\Psi$.}\\
\Psi ; \Ga ; \D \vdash P :: z{:}A & \mbox{$P$ offers session $z{:}A$
                                   when composed with processes
                                   according}\\
& \mbox{to $\Ga$ and $\D$ in
                                   context $\Psi$.}\\
\Psi \vdash K_1 = K_2 & \mbox{Kinds $K_1$ and $K_2$ are equal.}\\
\Psi \vdash \tau = \sigma :: K & \mbox{Types $\tau$ and $\sigma$ are
                                 equal of kind $K$.}\\
\Psi \vdash A = B :: K & \mbox{Session types $A$ and $B$ are equal of
                         kind $K$.}\\
\Psi \vdash M = N : \tau & \mbox{Terms $M$ and $N$ are equal of type $\tau$.}\\
\Psi ; \Ga ; \D \vdash P = Q :: z{:}A & \mbox{Processes $P$ and $Q$ are
                                       equal with typing $z{:}A$.}\\
\end{array}
\]

\subsubsection{Well-formed Contexts}

We write $\cdot$ for the empty context. We write $\Psi , x{:}\tau$ for
the extension of context $\Psi$ with the binding $x{:}\tau$. We assume
that $x$ does not occur in $\Psi$. We use a similar notation for the
session typing contexts $\D$ and $\Ga$.

\[
\begin{array}{c}
\infer[]{\cdot \vdash}{ }
\quad
\infer[]{\Psi , x{:}\tau \vdash}
  {\Psi \vdash & \Psi \vdash \tau :: \type}
                 \quad
\infer[]{\Psi , t {::} K \vdash}{\Psi \vdash & \Psi \vdash K}
\quad
\infer[]{\Psi ; \D , x {:} A \vdash}
{\Psi \vdash  & \Psi ; \D \vdash &  \Psi \vdash A :: \stype}\\[0.5em]
\infer[]{\Psi ; \Ga , x {:} A \vdash}
  {\Psi \vdash  & \Psi ; \Ga \vdash &  \Psi \vdash A :: \stype}
\end{array}                                      
\]

\subsubsection{Well-formed Kinds}

\[
\begin{array}{c}
\infer[]{\Psi \vdash \type}{\Psi \vdash}
\quad
\infer[]{\Psi \vdash \stype}{\Psi \vdash}
\quad
\infer[]
{\Psi \vdash \Pi x {:}\tau . K }
  {\Psi , x{:}\tau \vdash K & \Psi \vdash \tau :: \type}
                              \quad
\infer[]
{\Psi \vdash \hPi x {:}\tau . K }
                              {\Psi , x{:}\tau \vdash K & \Psi \vdash \tau :: \stype}\\[0.5em]
  \infer[]
  {\Psi \vdash \Pi t{::}K.K'}
  {\Psi \vdash K \quad \Psi , t {::} K \vdash K'}
\end{array}
\]

\subsection{Well-formed (Functional) Types}

\[
\begin{array}{c}
\infer[]
{\Psi \vdash \Pi x{:}\tau.\sigma :: \type}
{\Psi \vdash \tau :: \type & \Psi, x{:}\tau \vdash \sigma :: \type}
\quad
\infer[]
{\Psi \vdash \lambda x{:}\tau .\sigma :: \Pi x{:}\tau . K}
{\Psi \vdash \tau :: \type & \Psi , x{:} \tau \vdash \sigma :: K}\\[0.5em]
\infer[]
{\Psi \vdash \tau \, M :: K\{M/x\}}
  {\Psi \vdash \tau :: \Pi x{:}\sigma . K \quad \Psi \vdash M : \sigma}
\quad
\infer[]
{\Psi \vdash \{\ov{u_j{:}B_j};\ov{d_i{:}A_i} \vdash c{:}A\} :: \type}
  {\forall i,j .\Psi \vdash A_i :: \stype & \Psi \vdash B_j :: \stype  & \Psi \vdash A :: \stype}\\[0.5em]

  \infer[]
  {\Psi \vdash \lambda t{::}K.\sigma :: \Pi t {::}K.K'}
  {\Psi \vdash K \quad\Psi,t{::}K \vdash \sigma :: K'}
 \quad
  \infer[]{\Psi \vdash \tau\,\sigma :: K'\{\sigma/t\}}
  {\Psi \vdash \tau :: \Pi t {::}K.K' \quad
  \Psi \vdash \sigma :: K}\\[0.5em]

  \infer[]{\Psi \vdash t :: K}{t {::} K \in \Psi \quad \Psi \vdash}
\end{array} 
\]

\subsection{Well-formed Session Types}

\[
\begin{array}{c}
\infer[]
{\Psi \vdash \one :: \stype}{\Psi\vdash }
\quad
\infer[]
{\Psi \vdash \bang A :: \stype}
{\Psi \vdash A :: \stype}
\quad
\infer[]
{\Psi \vdash A \lolli B :: \stype}
{\Psi \vdash A :: \stype & \Psi \vdash B :: \stype}\\[0.5em]
\infer[]
{\Psi \vdash A \tensor B :: \stype}
  {\Psi \vdash A :: \stype & \Psi \vdash B :: \stype}
\quad                             
\infer[]
{\Psi \vdash \forall x{:}\tau . A :: \stype}
{\Psi \vdash \tau :: \type & \Psi , x{:}\tau \vdash A :: \stype}\\[0.5em]
\quad
\infer[]
{\Psi \vdash \exists x{:}\tau . A :: \stype}
{\Psi \vdash \tau :: \type & \Psi , x{:}\tau \vdash A :: \stype}
\quad
\infer[]
{\Psi \vdash \with\{\ov{l_i : A_i}\} :: \stype}
{\forall i . \Psi \vdash A_i :: \stype}\\[0.5em]
\infer[]
{\Psi \vdash \oplus\{\ov{l_i : A_i}\} :: \stype}
{\forall i . \Psi \vdash A_i :: \stype}
\quad
\infer[]
{\Psi \vdash \lambda x{:}\tau . A :: \hPi x{:}\tau.K} 
  {\Psi \vdash \tau :: \type \quad \Psi , x{:}\tau \vdash A :: K}\\[0.5em]
\infer[]
{\Psi \vdash A \, M :: K\{M/x\}}
  {\Psi \vdash A :: \hPi x{:}\tau.K & \Psi \vdash M : \tau}
\quad                                     
\infer[]
{\Psi \vdash A :: K'}
  {\Psi \vdash A :: K & \Psi \vdash K = K'}\\[0.5em]

  \infer[]{\Psi \vdash \lambda t {::}K . A :: \Pi t{::}K.K'}
  {\Psi , t{::}K \vdash A :: K'}
  \quad
  \infer[]{\Psi \vdash A\,B :: K'\{B/x\}}
  {\Psi \vdash A ::  \Pi t{::}K.K' \quad \Psi \vdash B :: K}
  \quad
  \infer[]{\Psi \vdash t :: K}{\Psi \vdash \quad t{::}K \in \Psi}
\end{array}
\]

\subsection{Typing for $\lambda$-Terms}
\label{app:typterm}

\[
\begin{array}{c}
\inferrule[$(\Pi I)$]
{\Psi \vdash \tau :: \type \quad \Psi , x{:}\tau \vdash M : \sigma}
  {\Psi \vdash \lambda x {:} \tau . M : \Pi x{:}\tau.\sigma}
\quad
\inferrule[$(\Pi E)$]
{\Psi \vdash M : \Pi x{:}\tau.\sigma \quad \Psi \vdash N : \tau }
  {\Psi \vdash M\, N : \sigma\{N/x\}}
\\[1em]
\inferrule[$(\m{var})$]
{\Psi \vdash \quad x{:}\tau \in \Psi}
  {\Psi \vdash x {:}\tau}
  
\quad                 
\inferrule[$(\{\}I)$]
  {\forall i,j . \Psi \vdash A_i,B_j :: \stype 
   \quad \Psi ; \ov{u_j{:}B_j} ; \ov{d_i{:}A_i} \vdash P :: c{:}A}
  {\Psi \vdash \{c \leftarrow P \leftarrow \ov{u_j} ; \ov{d_i}\} : \{\ov{u_j{:}B_j};\ov{d_i : A_i} \vdash c{:}A\} }
  \\[1em]
\inferrule[$(\m{Conv})$]
{\Psi \vdash M : \tau \quad \Psi \vdash \tau = \sigma :: \type}
  {\Psi \vdash M : \sigma}
\end{array}
\]

\subsection{Typing for Processes}
\label{app:typproc}

 \[
\begin{array}{c}
\inferrule[$(\rgt{\exists})$]
  {\Psi \vdash M {:} \tau \quad 
   \Psi ; \Ga ;  \Delta \seq P :: c {:} A\{M/x\}}
  {\Psi ; \Ga ;\Delta \seq  c\langle M \rangle_{\exists x{:}\tau.A} .P :: c {:}
    \exists x{:}\tau . A}
\quad
\inferrule[$(\lft{\exists})$]
  {\Psi \vdash \tau :: \type \quad \Psi, x{:}\tau \semi \Ga ; \Delta , c{:} A\seq Q :: d {:} D }
{\Psi \semi \Ga ; \Delta , c{:}\exists x{:}\tau. A \seq c(x{:}\tau).Q :: d {:} D}\\[0.5em]

\inferrule[$(\rgt{\forall})$]
  {\Psi \vdash \tau :: \type \quad \Psi, x{:}\tau \semi \Ga ; \Delta \seq P :: c {:} A}
{\Psi ; \Ga ; \Delta \seq c(x{:}\tau).P :: c {:} \forall x{:}\tau . A}
\quad
\inferrule[$(\lft{\forall})$]
  {\Psi \vdash M {:} \tau \quad
   \Psi ; \Ga ;  \Delta , c{:}A\{M/x\} \seq Q :: d {:} D}
  {\Psi ; \Ga ;\Delta , c{:}\forall x{:}\tau . A \seq  c\langle M \rangle_{\forall x{:}\tau.A} .Q :: d {:}
   D}
\\[0.5em]
\inferrule[$(\m{id})$]
 { \Psi ;\Ga \vdash \quad [\Psi \vdash A :: \stype] }
  {\Psi ; \Ga ; d {:} A \seq [d\leftrightarrow c] :: c {:} A}
\quad
\inferrule[$(\rgt{\one})$]
  { \Psi ; \Ga \vdash  }
{\Psi ; \Ga; \cdot \seq \mathbf{0} :: c {:} \one}
\quad
\inferrule[$(\lft{\one})$]
  {\Psi ; \Ga ; \Delta \seq P :: d {:} D}
  {\Psi ; \Ga;\Delta, c{:}\one \seq P :: d {:}
    D}\\[1.5em]
\inferrule[$(\rgt\bang)$]
{ \Psi ; \Ga ; \cdot \seq P :: x {:} A}
{\Psi; \Ga ; \cdot \seq \bang c(x).P ::
    c {:} \bang A}
\quad
\inferrule[$(\lft\bang)$]
{\Psi ; \Ga , u{:}A ; \D \seq P :: d {:} D}
{\Psi ; \Ga ; \D , c{:}\bang A \seq P\{c/u\} :: d {:} D}
\quad
\inferrule[(copy)]
{\Psi ; \Ga , u{:}A ; \D , x{:}A \vdash P :: d{:}D}
{\Psi ; \Ga , u{:}A ; \D \vdash (\nub x)u\langle x \rangle.P :: d{:}D}
\\[1.5em]
\inferrule[$(\rgt\tensor)$]
{\Psi ; \Ga ; \D_1 \vdash P_1 :: x{:}A \quad \Psi ; \Ga ; \D
  \vdash P_2 :: c{:}B}
{\Psi ; \Ga ; \D_1 , \D_2 \vdash (\nub x)c\langle x \rangle.(P_1 \mid P_2) ::
  c{:}A\tensor B}
\quad
\inferrule[$(\lft\tensor)$]
{\Psi ; \Ga ; \D , x{:}A , c{:}B \vdash Q :: d{:}D }
{\Psi ; \Ga ; \D , c{:}A\tensor B \vdash c(x).Q :: d{:}D }\\[1.5em]
\inferrule[$(\rgt\lolli)$]
{\Psi ; \Ga ; \D , x{:}A \vdash P :: c{:}B}
{\Psi ; \Ga ; \D \vdash c(x).P :: c{:}A\lolli B}
\quad
\inferrule[$(\lft\lolli)$]
{\Psi ; \Ga ; \D_1 \vdash Q_1 :: x{:}A \quad \Psi ; \Ga ; \D_2 , c{:}B
  \vdash Q_2 :: d{:}D}
{\Psi ; \Ga ; \D_1 , \D_2 , c{:}A\lolli B\vdash (\nub x)c\langle x \rangle.(Q_1
  \mid Q_2) :: d{:}D}
\\[1.5em]
\inferrule[$(\rgt\with)$]
{\Psi ; \Ga ; \D \vdash P_1 :: c{:}A_1 \; \dots \; \Psi ; \Ga ; \D
  \vdash P_n :: c{:}A_n }
{\Psi ; \Ga ; \D \vdash c.\mathsf{case}(\overline{l_j \Rightarrow
    P_j}) :: c{:}\with\{\overline{l_j {:}A_j}\}}
\quad
\inferrule[$(\lft\with)$]
{\Psi ; \Ga ; \D , c{:}A_i \vdash Q :: d{:}D}
{\Psi ; \Ga ; \D , c{:}\with\{\overline{l_j {:}A_j}\} \vdash c.l_i
; Q :: d{:}D}
\\[1.5em]
\inferrule[$(\rgt\oplus)$]
{\Psi ; \Ga ; \D \vdash P :: c{:}A_i}
{\Psi ; \Ga ; \D \vdash c.l_i ; P :: c{:}\oplus\{\overline{l_j
    {:}A_j}\}}
\quad
\inferrule[$(\lft\oplus)$]
{\Psi ; \Ga ;\D , c{:}A_1 \vdash Q_1 :: d{:}D \; \dots \; \Psi ; \Ga
  ; \D , c{:}A_n \vdash Q_n :: d{:}D }
{\Psi ; \Ga ; \D , c{:}\oplus\{\overline{l_j {:}A_j}\} \vdash
  c.\mathsf{case}(\overline{l_j \Rightarrow Q_j}) :: d{:}D}
\\[1.5em]
 \inferrule[(cut)]
 {\Psi ; \Ga ; \D_1 \vdash P :: x{:}A \quad \Psi ; \Ga ; \D_2 , x{:}A
   \vdash Q :: d{:}D}
 {\Psi ; \Ga ; \D_1 , \D_2 \vdash (\nub x)(P \mid Q) :: d{:}D}
 \quad
 \inferrule[(cut${}^\bang$)]
 {\Psi ; \Ga ; \cdot \vdash P :: x{:}A \quad \Psi ; \Ga , u{:}A ; \D
 \vdash Q :: d{:}D}
 {\Psi ; \Ga ; \D \vdash (\nub u)(\bang u(x).P \mid Q) :: d{:}D}\\[0.5em]
\inferrule[($\{\}E$)]
{\D' = \ov{d_i : B_i} \quad \ov{u_j{:}C_j} \subseteq \Ga \quad
\Psi \vdash M : \{\ov{u_j{:}C_j};\ov{d_i{:}B_i} \vdash c{:}A\}
\quad \Psi ; \Ga ; \D ,x{:}A\vdash Q :: z{:}C }
{\Psi ; \Ga ; \D' , \D \vdash 
x \leftarrow M \leftarrow \ov{u_j};\ov{y_i} ; Q :: z{:}C }\\[1em]
\inferrule[($\rgt{\m{Conv}}$)]
{\Psi ; \Ga ; \D \vdash P :: z{:}A \quad \Psi \vdash A = B :: \stype }
{\Psi ; \Ga ; \D \vdash P :: z{:}B}

  \quad

  \inferrule[($\lft{\m{Conv}}$)]
  {\Psi ; \Ga' ; \D' \vdash P :: z{:}A \quad \Psi ; \Ga' ; \D' = \Psi ;
  \Ga ; \D}
  {\Psi ; \Ga ; \D \vdash P :: z{:}A}
  
\end{array}
\]

\subsection{Definitional Equality for Kinds}

\[
\begin{array}{c}
\inferrule[$(\m{KEqR})$]
{\Psi \vdash K }
  {\Psi \vdash K = K}

\quad
\inferrule[$(\m{KEqS})$]
{\Psi \vdash K_1 = K_2 \quad \Psi \vdash K_2 = K_3}
  {\Psi \vdash K_1 = K_3}

\quad
\inferrule[$(\m{KEqT})$]
{\Psi \vdash K_2 = K_1}{\Psi \vdash K_1 = K_2}
\\[0.5em]

\inferrule[$(\m{KEq}\Pi)$]
{\Psi \vdash \tau = \sigma :: \type \quad
 \Psi , x{:}\tau \vdash K_1 = K_2}
  {\Psi \vdash \Pi x{:}\tau . K_1 = \Pi x{:}\sigma . K_2}
  \quad
  \inferrule[$(\m{KEqK}\Pi)$]
  {\Psi \vdash K_1 = K_3 \quad \Psi , t {::}K_1 \vdash K_2 = K_4}
  {\Psi \vdash \Pi t :: K_1 . K_2 = \Pi t :: K_3 . K_4}
\end{array}
\]

\subsection{Definitional Equality for (Functional) Types}

\[
\begin{array}{c}
\inferrule[$(\m{TEqR})$]
{\Psi \vdash \tau :: \type}
{\Psi \vdash \tau = \tau :: \type}
  \quad
\inferrule[$(\m{TEqT})$]
{\Psi \vdash \tau_1 = \tau_2 :: \type \quad \Psi \vdash \tau_2 = \tau_3 :: \type}
  {\Psi \vdash \tau_1 = \tau_3 :: \type}\\[0.5em]
\inferrule[$(\m{TEqS})$]
{\Psi \vdash \sigma = \tau :: \type}
  {\Psi \vdash \tau = \sigma :: \type}
\quad
\inferrule[$(\m{TEq}\Pi)$]
{\Psi \vdash \tau = \tau' :: \type \quad
 \Psi , x{:}\tau \vdash \sigma = \sigma' ::\type}
{\Psi \vdash \Pi x{:}\tau . \sigma = \Pi x{:}\tau' . \sigma' :: \type}
  \\[0.5em]
\inferrule[$(\m{TEq}\lambda)$]
{\Psi \vdash \tau = \tau' :: \type \quad \Psi , x{:}\tau \vdash \sigma = \sigma' :: K }
  {\Psi \vdash \lambda x{:}\tau.\sigma = \lambda x{:}\tau'.\sigma' :: \Pi x{:}\tau.K }
\quad
\inferrule[$(\m{TEqApp})$]
  {\Psi \vdash \tau = \sigma :: \Pi x{:}\tau'.K
  \quad
 \Psi \vdash M = N : \tau'}
  {\Psi \vdash \tau\,M = \sigma\,N :: K\{M/x\}}
\\[1em]
\inferrule[$(\m{TEq}\beta)$]
{\Psi ,x{:}\tau \vdash \sigma :: K \quad \Psi \vdash M : \tau}
  {\Psi \vdash (\lambda x{:}\tau.\sigma)\,M = \sigma\{M/x\} :: K\{M/x\}}
\quad
\inferrule[$(\m{TEq}\eta)$]
{\Psi \vdash \sigma :: \Pi x{:}\tau.K \quad x\not\in fv(\sigma)}
  {\Psi \vdash \lambda x{:}\tau.\sigma\,x = \sigma :: \Pi x{:}\tau.K  }
\\[1em]
\inferrule[$(\m{TEq}\{\})$]
{\forall i,j.\quad \Psi \vdash A_i = B_i :: \stype \quad \Psi \vdash C_j = D_j :: \stype \quad
 \Psi \vdash A = B :: \stype}
  {\Psi \vdash \{ \ov{u_j{:}C_j};\ov{d_i{:}A_i} \vdash c{:}A \}
  = \{ \ov{u_j{:}D_j};\ov{d_i{:}B_i} \vdash c{:}B \} :: \type}\\[1em]
\inferrule[$(\m{TEqT}\lambda)$]
{\Psi \vdash K_1 = K_2 \quad \Psi , t{::}K_1 \vdash \tau = \sigma :: K_3 }
  {\Psi \vdash \lambda t{::}K_1.\tau = \lambda t{::}K_2.\sigma :: \Pi x{:}K_1.K_3 }
\quad
  \inferrule[$(\m{TEqTApp})$]
  {\Psi \vdash \tau_1 = \sigma_1 :: \Pi t{::}K_1.K_2
  \quad
 \Psi \vdash \tau_2 = \sigma_2 : K_1}
  {\Psi \vdash \tau_1\,\tau_2 = \sigma_1\,\sigma_2 :: K_2\{\tau_2/t\}}
\\[1em]
\inferrule[$(\m{TEqT}\beta)$]
{\Psi ,t{::}K \vdash \tau :: K' \quad \Psi \vdash \sigma :: K}
  {\Psi \vdash (\lambda t{::}K.\tau)\,\sigma = \tau\{\sigma/t\} :: K'\{\sigma/t\}}
\quad  
  \inferrule[$(\m{TEqConv})$]
  {\Psi \vdash \tau = \sigma :: K \quad \Psi \vdash K = K'}
  {\Psi \vdash \tau = \sigma :: K'}
\end{array}
\]

\subsection{Definitional Equality for Session Types}

\[
  \begin{array}{c}
    \inferrule[$(\m{STEqR})$]
    {\Psi \vdash A :: \stype}
    {\Psi \vdash A = A :: \stype}
    \quad
    \inferrule[$(\m{STEqS})$]
    {\Psi \vdash B = A :: \stype}
    {\Psi \vdash A = B :: \stype}\\[1em]
    \inferrule[$(\m{STEqT})$]
    {\Psi \vdash A = B :: \stype \quad \Psi \vdash B = C :: \stype}
    {\Psi \vdash A = C :: s\type    }
 \quad
\inferrule[$(\m{STEq}\bang)$]
    {\Psi \vdash A = B :: \stype }
    {\Psi \vdash \bang A = \bang B :: \stype }
\\[1em]
\inferrule[$(\m{STEq}\!\lolli)$]
{\Psi \vdash A = C :: \stype \quad \Psi \vdash B = D :: \stype}
    {\Psi \vdash A \lolli B = C \lolli D :: \stype}
\quad
\inferrule[$(\m{STEq}\tensor)$]
    {\Psi \vdash A = C :: \stype
    \quad
 \Psi \vdash B = D :: \stype}{\Psi \vdash A \tensor B = C \tensor D :: \stype}    \\[1em]
\inferrule[$(\m{STEq}\forall)$]
{\Psi \vdash \tau = \tau' :: \type \quad
 \Psi , x{:}\tau \vdash A = B :: \stype }
    {\Psi \vdash \forall x{:}\tau. A = \forall x{:}\tau'.B :: \stype}
\quad                                   
\inferrule[$(\m{STEq}\exists)$]
{\Psi \vdash \tau = \tau' :: \type \quad
 \Psi , x{:}\tau \vdash A = B :: \stype }
{\Psi \vdash \exists x{:}\tau. A = \exists x{:}\tau'.B :: \stype}
\\[1em]
\inferrule[$(\m{STEq}\with)$]
    {\forall i . \Psi \vdash A_i = B_i :: \stype}
    {\Psi \vdash \with \{\ov{l_i {:} A_i}\} = \with \{\ov{l_i {:} B_i}\} :: \stype }
\quad
\inferrule[$(\m{STEq}\oplus)$]
{\forall i . \Psi \vdash A_i = B_i :: \stype}
    {\Psi \vdash \oplus \{\ov{l_i {:} A_i}\} = \oplus \{\ov{l_i {:} B_i}\} :: \stype }
\\[1em]
\inferrule[$(\m{STEq}\lambda)$]
{\Psi \vdash \tau = \tau' :: \type \quad
 \Psi , x{:}\tau \vdash A = B :: K}
    {\Psi \vdash \lambda x{:} \tau. A = \lambda x{:}\tau' . B :: \Pi x{:}\tau.K}
\quad
\inferrule[$(\m{STEqApp})$]
{\Psi \vdash A = B :: \Pi x{:}\tau . K \quad
 \Psi \vdash M = N : \tau}
{\Psi \vdash A \, M = B \, N :: K\{M/x\} }
\\[1em]
\inferrule[$(\m{STEq}\beta)$]
{\Psi , x{:}\tau \vdash A :: K \quad \Psi \vdash M : \tau}
    {\Psi \vdash (\lambda x{:}\tau.A)\,M = A\{M/x\} :: K\{M/x\}}
\quad
\inferrule[$(\m{STEq}\eta)$]
{\Psi \vdash A :: \Pi x{:}\tau.K \quad x \not\in fv(A)}
    {\Psi \vdash \lambda x{:}\tau.A\,x = A :: \Pi x{:}\tau.K}\\[1em]
\inferrule[$(\m{STEqT}\lambda)$]
{\Psi \vdash K_1 = K_2 \quad \Psi , t{::}K_1 \vdash A = B :: K_3 }
  {\Psi \vdash \lambda t{::}K_1.A = \lambda t{::}K_2.B :: \Pi x{:}K_1.K_3 }
\quad
  \inferrule[$(\m{STEqTApp})$]
  {\Psi \vdash A = C :: \Pi t{::}K_1.K_2
  \quad
 \Psi \vdash B = D : K_1}
  {\Psi \vdash A\,B = C\,D :: K_2\{B/t\}}
\\[1em]
\inferrule[$(\m{STEqT}\beta)$]
{\Psi ,t{::}K \vdash A :: K' \quad \Psi \vdash B :: K}
  {\Psi \vdash (\lambda t{::}K.A)\,B = A\{B/t\} :: K'\{B/t\}}
    \quad

    \inferrule[$(\m{STEqConv})$]
    {\Psi \vdash A = B :: K \quad \Psi \vdash K = K'}
    {\Psi \vdash A = B :: K'}
\end{array}
\]

\subsection{Definitional Equality for $\lambda$-Terms}
\label{app:defeqterm}

\[
\begin{array}{c}
\inferrule[$(\m{TMEqR})$]
{\Psi \vdash M : \tau}
  {\Psi \vdash M = M : \tau}
\quad
\inferrule[$(\m{TMEqS})$]
{\Psi \vdash N = M : \tau}
{\Psi \vdash M = N : \tau}
\quad
  \inferrule[$(\m{TMEqT})$]
  {\Psi \vdash M = N' : \tau \quad
  \Psi \vdash N' = N : \tau}
  {\Psi \vdash M = N : \tau}\\[1em]

  \inferrule[$(\m{TMEqVar})$]
  {\Psi \vdash \quad x{:}\tau \in \Psi}
  {\Psi \vdash x = x : \tau}
  \quad
  \inferrule[$(\m{TMEq}\lambda)$]
  {\begin{array}{c}\Psi \vdash \lambda x{:} \tau . M : \Pi x{:}\tau.\sigma \quad
 \Psi \vdash \lambda x{:} \tau' . N : \Pi x{:}\tau'.\sigma' \\
 \Psi \vdash \Pi x{:}\tau.\sigma = \Pi x{:}\tau'.\sigma' :: \type \quad
 \Psi , x{:}\tau \vdash M = N : \sigma\end{array}}
{\Psi \vdash \lambda x {:} \tau . M = \lambda x{:} \tau' . N : \Pi x{:}\tau.\sigma}
\\[1em]
\inferrule[$(\m{TMEqApp})$]
{\Psi \vdash M = M' : \Pi x{:}\tau.\sigma \quad
 \Psi \vdash N = N' : \tau}
  {\Psi \vdash M\,N = M'\, N' : \sigma\{N/x\}}
\quad
\inferrule[$(\m{TMEq}\beta)$]
{[\Psi \vdash \tau :: \type] \quad \Psi ,x{:}\tau\vdash M : \sigma \quad
 \Psi \vdash N : \tau}
{\Psi \vdash (\lambda x{:} \tau . M)\,N = M\{N/x\} : \sigma\{N/x\}}\\[1em]
\inferrule[$(\m{TMEq}\eta)$]
  {\Psi \vdash M : \Pi x{:}\tau.\sigma \quad x \not\in fv(M)}
  {\Psi \vdash \lambda x{:}\tau. M\, x = M : \Pi x{:}\tau.\sigma }
\\[1em]
  \inferrule[$(\m{TMEq}\{\}\eta)$]
  {\Psi \vdash M : \monad{\ov{u_j{:}B_j};\ov{d_i{:}A_i} \vdash c{:}A}}
  {\Psi \vdash \monad{c \leftarrow (y\leftarrow M ;\ov{u_j};\ov{d_i} ; [y\leftrightarrow c])
  \leftarrow\ov{u_j};\ov{d_i} } = M :  \monad{\ov{u_j{:}B_j};\ov{d_i{:}A_i} \vdash c{:}A} }
  \\[1em]
\inferrule[$(\m{TMEq}\{\})$]
{[\forall i,j . \Psi \vdash B_j :: \stype \quad \Psi \vdash A_i :: \stype]\quad \Psi ; \ov{u_j{:}B_j} ; \ov{d_i{:}A_i} \vdash P = Q :: c{:}A }
  {\Psi \vdash \{c \leftarrow P \leftarrow \ov{u_j};\ov{d_i} \} = 
  \{c \leftarrow Q \leftarrow \ov{u_j};\ov{d_i}\} : \{\ov{u_j{:}B_j};
  \ov{d_i{:}A_i} \vdash c{:}A\} }\\[1em]
  \inferrule[$(\m{TMEqConv})$]
  {\Psi \vdash M = N : \tau \quad \Psi \vdash \tau = \sigma :: \type}
  {\Psi \vdash M = N : \sigma}
\end{array}
\]

\subsection{Definitional Equality for Processes}
\label{app:defeqproc}

\[
  \begin{array}{c}
    \inferrule[$(\m{PEqRefl})$]
    {\Psi ; \Ga ; \D \vdash P :: z{:}A}
    {\Psi ; \Ga ; \D \vdash P = P :: z{:}A}
    \quad
    \inferrule[$(\m{PEqS})$]
{\Psi ; \Ga ; \D \vdash Q = P :: z{:}A}
    {\Psi ; \Ga ; \D \vdash P = Q :: z{:}A}\\[1em]
    \inferrule[$(\m{PEqT})$]
    {\Psi ; \Ga ; \D \vdash P = Q :: z{:}A \quad
    \Psi ; \Ga ; \D \vdash Q = R :: z{:}A }    
{\Psi ; \Ga ; \D \vdash P = R :: z{:}A}
    \\[1em]
    \inferrule[]
    {\Psi ; \Ga ; \D \vdash P :: z{:}A \quad
     P \tra{} Q \quad \Psi ; \Ga ; \D \vdash Q :: z{:}A}
    {\Psi ; \Ga ; \D \vdash P = Q :: z{:}A}\\[1em]

     \inferrule*[left=$(\m{PEq}\forall\eta)$]
    { }
    {\Psi ; \Ga ; d{:}\forall x{:}\tau.A \vdash  c(x).d\langle x \rangle.[d\leftrightarrow c] = [d\leftrightarrow c] :: c{:}\forall x{:}\tau.A }\\[1em]

    \inferrule*[left=$(\m{PEqCC}\forall)$]
    {\Psi ; \Ga ; \D \vdash P :: d{:}B \quad
     \Psi , x{:}\tau ; \Ga ; \D' , d{:}B \vdash Q :: c{:} A  }
    {\Psi ; \Ga ; \D  , \D' \vdash (\nub d)(P \mid c(x).Q) =
    c(x).(\nub d)(P \mid Q) :: c{:}\forall x{:}\tau.A}
   \\[1em]

    \inferrule[$(\m{PEq}\rgt\forall)$]
    {\begin{array}{c}
       \Psi ; \Ga ; \D \vdash z(x{:}\tau).P :: z{:}\forall x{:}\tau.A
       \quad \Psi ; \Ga ; \D \vdash z(x{:}\tau').Q :: z{:}\forall x{:}\tau'.B\\
       \Psi \vdash \forall x{:}\tau.A = \forall x{:}\tau'.B :: \stype \quad \Psi , x{:}\tau ; \Ga ; \D \vdash
       P = Q :: z{:}A
     \end{array}}
    {\Psi ; \Ga ; \D \vdash z(x{:}\tau).P = z(x{:}\tau').Q ::
    z{:}\forall x{:}\tau.A    }\\[1em]
     \mbox{(Other congruence, $\eta$ and CC rules)}
  \end{array}
\]

\section{Type Soundness}
\label{app:soundness}

We use $\Psi \vdash \mathcal{J}$ to signify any of the
judgments $\Psi \vdash K$, $\Psi \vdash A :: K$, $\Psi \vdash \tau ::
K$ and respective definitional equality judgments. We use $\Psi ; \Ga
; \D \vdash \mathcal{J}$ in a similar fashion.

\begin{lemma}[Subderivation Properties]\label{lem:subderiv}
~
  \begin{enumerate}
\item Every derivation of $\Psi \vdash \mathcal{J}$ has a proof of $\Psi \vdash$
  as a sub-proof.
\item Every derivation of $\Psi , x{:}\tau \vdash$ has a proof of $\Psi
  \vdash \tau :: \type$ as a sub-proof.
\item Every derivation of $\Psi , t {::} K \vdash$ has a proof of
  $\Psi \vdash K$ as a sub-proof.
\item Every derivation of $\Psi , x{:} K \vdash$ has a proof of $\Psi
  \vdash K$ as a sub-proof.
\item If $\Psi \vdash \tau :: K$ or $\Psi \vdash A :: K$ then $\Psi \vdash K$
\item If $\Psi \vdash M : \tau$ then $\Psi \vdash \tau :: \type$
\item If $\Psi ; \Ga ; \D \vdash P :: z{:}A$ then $\Psi \vdash A :: \stype$

  \end{enumerate}
\end{lemma}

\begin{proof}
  By induction on the given derivation.

  \begin{description}
  \item[Case:] Kind well-formedness

    Straightforward by induction.

  \item[Case:] Functional type well-formedness

    Straightforward by induction.

  \item[Case:] Session type well-formedness

    Straightforward by induction.

  \item[Case:] Typing for terms

    Straightforward by induction. Base-case is immediate.

  \item[Case:] Typing for processes

    Straightforward by induction. Base-cases are immediate.

  \item[Case:] Kind equivalence

    Straightforward, base case is reflexivity (from i.h. for well-formedness).

  \item[Case:] Type Equivalence

    As above.

  \item[Case:] Session type equivalence

    As above.
    
  \end{description}

\end{proof}

\begin{lemma}[Weakening]\label{lem:weaken}
~
If $ \Psi \vdash$ and $\Psi' \vdash$ and $\Psi \subseteq \Psi'$
  then:
\begin{enumerate}
\item $\Psi \vdash \mathcal{J}$ implies $\Psi' \vdash \mathcal{J}$
\item $\Psi ; \Ga ; \D \vdash \mathcal{J}$ implies $\Psi' ; \Ga ; \D \vdash \mathcal{J}$
\end{enumerate}
\end{lemma}
\begin{proof}
Straightforward induction on the given derivation.
\end{proof}

\substlem*

 \begin{proof}
   By induction on the second given derivation. We show only some
   illustrative cases.

   \begin{description}
   \item[Case:] $\m{TypeAppWF}$
     \begin{tabbing}
$\Psi , x{:}\tau , \Psi' \vdash \tau' :: \Pi y{:}\sigma . K$ and $\Psi
, x{:}\tau , \Psi' \vdash M' : \sigma$ \` by inversion\\
$\Psi , \Psi' \vdash \tau'\{M/x\} :: \Pi y{:}\sigma\{M/x\}.K\{M/x\}$
\` by i.h.\\
$\Psi , \Psi' \vdash M'\{M/x\} : \sigma\{M/x\}$ \` by i.h.\\
$\Psi , \Psi' \vdash \tau'\{M/x\}\,M'\{M/x\} : K\{M'/y\}\{M/x\}$ \` by $\m{TypeAppWF}$
     \end{tabbing}
     
   \item[Case:] $\m{KindConv}$

     \begin{tabbing}
       $\Psi , x{:}\tau , \Psi' \vdash \tau :: K$ and
       $\Psi , x{:}\tau , \Psi' \vdash K = K'$ \` by inversion\\
       $\Psi , \Psi' \vdash \tau\{M/x\} :: K\{M/x\}$ \` by i.h.\\
       $\Psi , \Psi' \vdash K\{M/x\} = K'\{M/x\}$ \` by i.h.\\
       $\Psi , \Psi' \vdash \tau\{M/x\} :: K'\{M/x\}$ \` by  $\m{KindConv}$
      
     \end{tabbing}

   \item[Case:] $\m{Var}$

     \begin{tabbing}
       {\bf Subcase:} $x = y$\\
       $\Psi \vdash M : \tau$ \` by assumption\\
       $\Psi,\Psi'\{M/x\} \vdash M : \tau$ \` by weakening\\
       {\bf Subcase:} $x \neq y$\\
       $\Psi , x{:} \tau , \Psi',y{:}\tau' \vdash y{:}\tau'$ \` by
       weakening and $\m{Var}$
       
     \end{tabbing}

\item[Case:] $\m{TEq}\beta$     

  \begin{tabbing}
$\Psi,x{:}\tau,\Psi',y{:}\tau' \vdash \sigma :: K$ and
$\Psi,x{:}\tau,\Psi'\vdash M' : \tau'$ \` by inversion\\
$\Psi,\Psi'\{M/x\},y{:}\tau'\{M/x\} \vdash \sigma\{M/x\} :: K\{M/x\}$
\` by i.h.\\
$\Psi,\Psi'\{M/x\}\vdash M'\{M/x\} : \tau'\{M/x\}$ \` by i.h.\\
$\Psi , \Psi'\{M/x\} \vdash (\lambda y{:}\tau'\{M/x\}
. \sigma\{M/x\})\,M'\{M/x\} = \sigma\{M/x\}\{M'\{M/x\}/y\} ::
K\{M'/x\}\{M\{M/x\}/y\}$
\\ \` by $\m{TEq}\beta$ 
  \end{tabbing}

\item[Case:] $\m{TEq}\eta$

  \begin{tabbing}
$\Psi , x{:}\tau , \Psi' \vdash \sigma :: \Pi y{:}\tau'.K$ and $y\not\in fv(\sigma)$ \` by
inversion\\
$\Psi  , \Psi'\{M/x\} \vdash \sigma\{M/x\} :: \Pi
y{:}\tau'\{M/x\}.K\{M/x\}$ \` by i.h\\
$\Psi , \Psi'\{M/\} \vdash \lambda y{:}\tau'\{M/x\}.(\sigma\{M/x\}\,y)
= \sigma\{M/x\} :: \Pi y{:}\tau'\{M/x\}.K\{M/x\}$ \` by $\m{TEq}\eta$
  \end{tabbing}

\item[Case:] $\m{PEqRed}$

  \begin{tabbing}
$\Psi , x{:}\tau , \Psi' ; \Ga ; \D \vdash P :: z{:}A$, $P\tra{}^* Q$ and $\Psi , x{:}\tau , \Psi' ; \Ga ; \D
\vdash Q :: z{:}A$ \` by inversion\\
$\Psi , \Psi'\{M/x\} ; \Ga\{M/x\} ; \D \{M/x\} \vdash P\{M/x\} ::
z{:}A\{M/x\}$ \` by i.h.\\
$\Psi , \Psi'\{M/x\} ; \Ga\{M/x\} ; \D \{M/x\} \vdash Q\{M/x\} ::
z{:}A\{M/x\}$ \` by i.h.\\
$P\{M/x\} \tra{}^* Q\{M/x\}$ \` by compatibility of reduction with
substitution\\
$\Psi , \Psi'\{M/x\} ; \Ga\{M/x\} ; \D \{M/x\}  \vdash P\{M/x\} =
Q\{M/x\} :: z{:}A\{M/x\}$ \` by  $\m{PEqRed}$
  \end{tabbing}
  
   \end{description}

 \end{proof}

 \begin{lemma}[Type Substitution]
  \begin{enumerate}
  \item If $\Psi \vdash \tau :: K$ and $\Psi , t{::}K , \Psi' \vdash
    \mathcal{J}$ then $\Psi , \Psi'\{\tau/t\} \vdash \mathcal{J}\{\tau/t\}$;
  \item If $\Psi \vdash \tau :: K$ and $\Psi , t{::}K , \Psi' ; \Ga ;
    \D \vdash \mathcal{J}$ then $\Psi , \Psi'\{\tau/t\} ; \Ga
    \{\tau/t\} ; \D \{\tau/t\}
    \vdash \mathcal{J}\{\tau/t\}$
  \item If $\Psi \vdash A :: K$ and $\Psi , t{::}K , \Psi' \vdash
    \mathcal{J}$ then $\Psi , \Psi'\{A/t\} \vdash \mathcal{J}\{A/t\}$;
   \item If $\Psi \vdash A :: K$ and $\Psi , t{::}K , \Psi' ; \Ga ;
    \D \vdash \mathcal{J}$ then $\Psi , \Psi'\{A/t\} ; \Ga
    \{A/t\} ; \D \{A/t\}$
  \end{enumerate}
  
\end{lemma}

\begin{lemma}[Context Conversion]\label{lem:ctxtconvpsi}

Let $\Psi ,x{:}\tau \vdash$ and $\Psi \vdash \tau' :: K$. If $\Psi ,
x{:}\tau \vdash \mathcal{J}$ and $\Psi \vdash \tau = \tau' :: K$ then
$\Psi, x{:}\tau' \vdash \mathcal{J}$.

\end{lemma}
\begin{proof} Straightforward from the properties above.
  
  \begin{tabbing}
    $\Psi , x{:}\tau' \vdash x{:}\tau'$ \` by variable rule\\
    $\Psi \vdash \tau' = \tau :: K$ \` by symmetry\\
    $\Psi , x{:}\tau' \vdash x{:}\tau$ \` by conversion\\
    $\Psi, x'{:}\tau \vdash \alpha\{x'/x\}$ \` renaming assumption\\
    $\Psi , x{:}\tau', x'{:}\tau \vdash \alpha\{x'/x\}$ \` by
    weakening\\
    $\Psi , x{:}\tau' \vdash \alpha \{x'/x\}\{x/x'\}$ \` by
    substitution\\
    $\Psi , x{:}\tau' \vdash \alpha$ \` by definition
  \end{tabbing}
\end{proof}

\begin{lemma}[Context Conversion --
  Processes]\label{lem:ctxtconvpsiproc}
Let $\Psi , x{:}\tau ; \D \vdash$, $\Psi , x{:}\tau ; \Ga \vdash$ and
$\Psi \vdash \tau :: \type$. If $\Psi , x{:}\tau ; \Ga ; \D \vdash
\mathcal{J}$ and $\Psi \vdash \tau = \tau' :: \type$ then $\Psi , x{:}\tau'
; \Ga ; \D \vdash \mathcal{J}$
\end{lemma}
\begin{proof}
Straightforward by Lemma~\ref{lem:ctxtconvpsi}.
\end{proof}

\begin{lemma}[Context Conversion -- Types]
Let $\Psi , t{::}K \vdash$ and $\Psi \vdash K'$. If $\Psi ,
t{::}K\vdash \mathcal{J}$ and $\Psi \vdash K = K'$ then $\Psi ,
t{::}K' \vdash \mathcal{J}$
\end{lemma}
\begin{proof}
Identical to Lemma~\ref{lem:ctxtconvpsi}
\end{proof}

\begin{lemma}[Functionality of Typing]\label{lem:functionalitytyping}
~

 Assume $\Psi \vdash M = N : \tau$, $\Psi \vdash M : \tau$ and
  $\Psi \vdash N : \tau$:
  \begin{enumerate}
    \item If $\Psi , x{:}\tau , \Psi' \vdash M' : \tau'$ then $\Psi ,
      \Psi'\{M/x\} \vdash M'\{M/x\} = M'\{N/x\} : \tau'\{M/x\}$
      
    \item If $\Psi , x{:}\tau , \Psi' \vdash \tau' :: K$ then
      $\Psi, \Psi'\{M/x\} \vdash \tau'\{M/x\} = \tau'\{N/x\} ::
      K\{M/x\}$
    \item If $\Psi , x{:}\tau , \Psi' \vdash A :: K$ then
      $\Psi, \Psi'\{M/x\} \vdash A\{M/x\} = A\{N/x\} ::
      K\{M/x\}$
    \item If $\Psi , x{:}\tau , \Psi' ; \Ga ; \D \vdash P :: z{:}A$
      then $\Psi ; \Psi'\{M/x\} ; \Ga\{M/x\} ; \D\{M/x\} \vdash
      P\{M/x\} = P\{N/x\} :: z{:}A\{N/x\}$

    \item If $\Psi , x{:}\tau , \Psi' \vdash K$ then
      $\Psi, \Psi'\{M/x\} \vdash K\{M/x\} = K\{N/x\}$ 
      
  \end{enumerate}
\end{lemma}

\begin{proof}

  By induction on the given typing derivation.

  \begin{description}
  \item[Case:] $\Psi , x{:}\tau , \Psi' \vdash x{:}\tau$ by variable rule

    \begin{tabbing}
      $\Psi \vdash M = N : \tau$ \` by assumption\\
      $\Psi , \Psi'\{M/x\} \vdash M = N : \tau$ \` by weakening
    \end{tabbing}

  \item[Case:] $\Psi , x{:}\tau , \Psi' \vdash y{:}\sigma$ with
    $y{:}\sigma \in \Psi$ or $\Psi'$

    \begin{tabbing}
$y : \sigma \in \Psi$ or $y{:}\sigma\{M/x\} \in \Psi'\{M/x\}$  \` by
definition\\
$\Psi , \Psi'\{M/x\} \vdash y = y : \sigma\{M/x\}$ \` by reflexivity
    \end{tabbing}

\item[Case:] $\Psi , x{:}\tau , \Psi' \vdash M_0 \, N_0 :
  \sigma_0\{N_0/y\}$ from $\Pi E$

  \begin{tabbing}
$\Psi , \Psi'\{M/x\} \vdash M_0\{M/x\} = M_0\{N/x\} : \Pi
y{:}\sigma_1\{M/x\}.\sigma_0\{M/x\}$ \` by i.h.\\
$\Psi , \Psi'\{M/x\} \vdash N_0 \{M/x\} = N_0 \{N/x\} :
\sigma_1\{M/x\}$ \` by i.h.\\
$\Psi , \Psi'\{M/x\} \vdash M_0\{M/x\}\,N_0\{M/x\} =
M_0\{N/x\}\,N_0\{N/x\} : (\sigma_0\{M/x\})\{(N_0\{M/x\})/y\}$ \\\` by
$\m{TMEqApp}$ rule\\
  \end{tabbing}

\item[Case:] $\Psi , x{:}\tau , \Psi' \vdash \lambda y{:}\tau_0 . M_0
  : \Pi y{:}\tau_0 . \tau_1$  by $\Pi I$ rule

  \begin{tabbing}
  $\Psi , \Psi'\{M/x\} \vdash \tau_0\{M/x\} = \tau_0\{N/x\} :: \type$
  \` by i.h.\\
  $\Psi , \Psi'\{M/x\} ,y{:}\tau_0\{M/x\} \vdash M_0\{M/x\} =
  M_0\{N/x\} : \tau_1\{M/x\}$ \` by i.h.\\
  $\Psi , \Psi'\{M/x\} \vdash \tau_0\{M/x\} :: \type$ \` by
  substitution lemma\\
  $\Psi , \Psi'\{M/x\} \vdash \tau_0\{M/x\} = \tau_0\{M/x\} :: \type$
  \` by reflexivity\\
  $\Psi , \Psi'\{M/x\} \vdash \tau_0\{N/x\} = \tau_0\{M/x\} :: \type$
  \` by symmetry\\
  $\Psi , \Psi'\{M/x\} \vdash \lambda y{:}\tau_0\{M/x\}.M_0\{M/x\} =
  \lambda y{:}\tau_0\{N/x\} . M_0\{N/x\} : \Pi
  y{:}\tau_0\{M/x\}.\tau_1\{M/x\}$ \\\` by $\m{TMEq}\lambda$ rule
  \end{tabbing}

  \item[Case:] $\Psi , x{:}\tau , \Psi' \vdash \monad{ c \leftarrow P
    \leftarrow \ov{u_j};\ov{d_i}} : \monad{\Ga ; \D \vdash c {:}A}$ by $\{\}I$

  \begin{tabbing}
$\Psi , \Psi'\{M/x\} ; \Ga\{M/x\} ; \D\{M/x\} \vdash P\{M/x\} =
P\{N/x\} :: c{:}A\{M/x\}$ \` by i.h.\\
$\Psi , \Psi'\{M/x\} \vdash \ov{A_j\{M/x\}} :: \stype$ \` by
substitution lemma\\
$\Psi , \Psi'\{M/x\} \vdash \ov{B_i\{M/x\}} :: \stype$ \` by
substitution lemma\\
Conclude by $\m{TMEq}\{\}$ rule
\end{tabbing}

\item[Case:] $\Psi , x{:}\tau , \Psi' \vdash M_0 :\tau_0$ by
  conversion rule

  \begin{tabbing}
$\Psi , \Psi'\{M/x\} \vdash M_0 \{M/x\} = M_0\{N/x\} : \tau_0'\{M/x\}$
\` by i.h.\\
$\Psi , \Psi'\{M/x\} \vdash \tau_0\{M/x\} = \tau_0'\{M/x\} :: \type$
\` by substitution lemma\\
$\Psi , \Psi'\{M/x\} \vdash M_0\{M/x\} = M_0\{N/x\} : \tau_0\{M/x\}$ \`
by conversion rule
  \end{tabbing}

\item[Case:] $\Psi , x{:}\tau , \Psi' \vdash \Pi y{:}\tau_0.\tau_1 ::
  \type$ by $\Pi$ formation rule

  \begin{tabbing}
$\Psi , \Psi'\{ M /x \} \vdash \tau_0 \{M/x\} :: \type$ \` by
substitution\\
$\Psi , \Psi'\{M/x\},y{:}\tau_0\{M/x\} \vdash \tau_1\{M/x\} =
\tau_1\{N/x\} :: \type$ \` by i.h.\\
$\Psi , \Psi'\{M/x\} \vdash \tau_0\{M/x\} = \tau_0\{N/x\} :: \type$ \`
by i.h.\\
$\Psi , \Psi'\{M/x\} \vdash \Pi y{:}\tau_0\{M/x\}.\tau_1\{M/x\} = \Pi
y{:}\tau_0\{N/x\}.\tau_1\{N/x\} :: \type$ \\\` by $\Pi$ formation rule
  \end{tabbing}

\item[Case:] $\Psi , x{:}\tau , \Psi' \vdash \lambda y{:}\tau_0.\sigma
  :: \Pi y{:}\tau_0.K_0$ by $\lambda$ formation rule

\begin{tabbing}
$\Psi , \Psi'\{M/x\} \vdash \tau_0 \{M/x\} = \tau_0 \{N/x\} :: \type$
\` by i.h.\\
$\Psi , \Psi'\{M/x\} , y{:}\tau_0\{M/x\} \vdash \sigma\{M/x\} =
\sigma\{N/x\} :: K_0\{M/x\}$ \` by i.h.\\
$\Psi , \Psi'\{M/x\} \vdash \lambda y{:}\tau_0\{M/x\}.\sigma\{M/x\} =
\lambda y{:}\tau_0\{N/x\}.\sigma\{N/x\} :: \Pi
y{:}\tau_0\{M/x\}.K_0\{M/x\}$ \\ \` by $\lambda$ formation rule
\end{tabbing}

\item[Case:] $\Psi , x{:}\tau,\Psi' \vdash \tau_0 \, M_0 ::
  K_0\{M/y\}$ by type application formation rule
  
\begin{tabbing}
$\Psi , \Psi' \{M/x\} \vdash \tau_0\{M/x\} = \tau_0\{N/x\} :: \Pi
y{:}\tau_1\{M/x\}.K_0\{M/x\}$ \` by i.h.\\
$\Psi , \Psi'\{M/x\} \vdash M_0\{M/x\} = M_0 \{N/x\} :\tau_1\{M/x\}$
\` by i.h.\\
$\Psi , \Psi'\{M/x\} \vdash \tau_0\{M/x\}\, M_0\{M/x\} =
\tau_0\{M/x\}\, M_0\{M/x\}
:: K_0\{M_0/y\}\{M/x\}$ \\\` by type app. formation rule and def. of substitution
\end{tabbing}

\item[Case:] $\{\}$ formation rule

\begin{tabbing}
Straightforward by i.h.
\end{tabbing}

\item[Case:] $\Psi , x{:}\tau , \Psi' \vdash \tau_0 :: K_0$ by
  conversion rule

\begin{tabbing}

$\Psi , \Psi'\{M/x\} \vdash \tau_0 \{M/x\} = \tau_0\{N/x\} ::
K_1\{M/x\}$ \` by i.h.\\
$\Psi , \Psi'\{M/x\} \vdash K_1 \{M/x\} = K_0 \{M/x\}$ \` by
substitution lemma\\
$\Psi ,\Psi'\{M/x\} \vdash \tau_0 \{M/x\} = \tau_0\{N/x\} ::
K_0\{M/x\}$ \` by conversion

\end{tabbing}

\item[Case:] $\Psi , x{:}\tau , \Psi' ; \Ga ; \D \vdash c\langle
  M_0\rangle.P_0 :: c{:}\exists y{:}\tau_0.A_0$ by $\rgt\exists$

\begin{tabbing}
$\Psi , \Psi'\{M/x\} \vdash M_0\{M/x\} = M_0 \{N/x\} : \tau_0\{M/x\}$
\` by i.h.\\
$\Psi , \Psi' \{M/x\} ; \Ga\{M/x\} ; \D\{M/x\} \vdash P_0\{M/x\} =
P_0\{N/x\} :: c{:}A_0\{M_0/y\}\{M/x\}$\\ \` by i.h.\\
Conclude by $\m{PEq}\rgt\exists$
\end{tabbing}

\item[Case:] $\Psi , x{:}\tau ,\Psi' ; \Ga ; \D , y{:}\exists w
  {:}\sigma.A \vdash y(w{:}\sigma).P_0 :: z{:}C$ by $\lft\exists$
\begin{tabbing}
$\Psi , \Psi'\{M/x\} , w{:}\sigma\{M/x\} ; \Ga\{M/x\} ; \D\{M/x\} ,
y{:}A\{M/x\} \vdash P_0\{M/x\} = P_0\{N/x\} :: z{:}C\{M/x\}$ \\\` by
i.h.\\
$\Psi , \Psi'\{M/x\} \vdash \sigma\{M/x\} :: \type$ \` by substitution
lemma\\
Conclude by $\m{PEq}\lft\exists$
\end{tabbing}

\item[Case:] $\Psi , x{:}\tau , \Psi' ; \Ga ; \D \vdash P :: z{:}B$ by $\rgt{\m{Conv}}$

\begin{tabbing}
$\Psi , \Psi'\{M/x\} ; \Ga \{M/x\} ; \D\{M/x\} \vdash P\{M/x\} =
P\{N/x\} :: z{:}A\{M/x\}$ \` by i.h.\\
$\Psi , \Psi'\{M/x\} \vdash A\{M/x\} = B\{M/x\} :: \stype$ \` by
substitution lemma\\
$\Psi , \Psi'\{M/x\} ; \Ga\{M/x\} ; \D\{M/x\} \vdash P \{M/x\} =
P\{N/x\} :: z{:}B\{M/x\}$\\ \` by conversion
\end{tabbing}

Remaining cases follow similar patterns, relying on the inductive
hypothesis and the substitution lemmata.
  \end{description}

We omit the analogue functionality property for type substitution.

\end{proof}

\begin{lemma}[Inversion for Products]
\begin{enumerate}
\item If $\Psi \vdash \Pi x{:}\tau.\sigma :: K$ then $\Psi \vdash \tau
  :: \type$ and $\Psi ,x{:}\tau \vdash \sigma ::\type$
\item If $\Psi \vdash \Pi x{:}\tau.K$ then $\Psi \vdash \tau :: \type$
  and $\Psi , x{:}\tau \vdash K$
\end{enumerate}
\end{lemma}
\begin{proof}
$(1)$ follows straightforwardly by induction on the given
derivation. $(2)$ is immediate by inversion.
\end{proof}

\begin{lemma}[Inversion for $\forall\exists$]
\begin{enumerate}
\item If $\Psi \vdash \forall x{:}\tau.A :: K$ then $\Psi \vdash \tau
  :: \type$ and $\Psi , x{:}\tau \vdash A ::\stype$
\item If $\Psi \vdash \exists x{:}\tau.A : K$ then $\Psi \vdash \tau
  :: \type$ and $\Psi , x{:}\tau \vdash A :: \stype$
\end{enumerate}
\end{lemma}

\begin{proof}
Straightforwardly by induction on the given
derivation.
\end{proof}

\valeqlem*


\begin{proof}
By simultaneous induction on the given derivation.

\begin{description}

\item[Case:] $\m{TMEqR}$

\begin{tabbing}
$\Psi \vdash M : \tau$ \` by inversion\\
$\Psi \vdash \tau :: \type$ \` by subderivation lemma
\end{tabbing}

\item[Case:] $\m{TMEqS}$ and $\m{TMEqT}$

  \begin{tabbing}
Immediate by i.h.
\end{tabbing}

\item[Case:] $\m{TMEq}\lambda$

  \begin{tabbing}
    $\Psi \vdash \lambda x{:}\tau.M : \Pi x{:}\tau.\sigma$,
    $\Psi \vdash \lambda x{:}\tau'.N : \Pi x{:}\tau'.\sigma'$,
    $\Psi \vdash \Pi x{:}\tau.\sigma = \Pi x{:}\tau'.\sigma' :: \type$
    \\
    and $\Psi , x{:}\tau \vdash M = N : \sigma$ \` by inversion\\
    $\Psi , x{:}\tau \vdash M : \sigma$, $\Psi , x{:}\tau \vdash N :
    \sigma$ and $\Psi , x{:}\tau \vdash \sigma :: \type$ \` by i.h.\\
    $\Psi \vdash \Pi x{:}\tau.\sigma :: \type$, $\Psi \vdash \Pi
    x{:}\tau'.\sigma' :: \type$ and $\Psi \vdash \type$ \` by i.h.\\
    $\Psi \vdash \lambda x{:}\tau.M : \Pi x{:}\tau.\sigma$ \` by $(\Pi
    I)$\\
    $\Psi \vdash \lambda x{:}\tau'.N : \Pi x{:}\tau'.\sigma'$ \` by
    ($\Pi I)$\\
    $\Psi \vdash \lambda x{:}\tau'.N : \Pi x{:}\tau.\sigma$ \` by
    conversion (and symmetry)
  \end{tabbing}

\item[Case:] $\m{TMEqApp}$

  \begin{tabbing}
$\Psi \vdash M = M' : \Pi x{:}\tau . \sigma$ and $\Psi \vdash N = N' :
  \tau$ \` by inversion\\
$\Psi \vdash M : \Pi x{:}\tau . \sigma$, $\Psi \vdash M' : \Pi
x{:}\tau.\sigma$ and $\Psi \vdash \Pi x{:}\tau.\sigma :: \type$ \` by
i.h.\\
$\Psi \vdash N : \tau$, $\Psi \vdash N' : \tau$ and $\Psi \vdash \tau
:: \type$ \` by i.h.\\
$\Psi , x{:}\tau \vdash \sigma :: \type$ \` by inversion for
products\\
$\Psi \vdash \sigma\{N/x\} :: \type$ \` by substitution\\
$\Psi \vdash M \, N : \sigma\{N/x\}$ \` by $(\Pi E)$\\
$\Psi \vdash M' \, N' : \sigma \{N'/x\}$ \` by $(\Pi E)$\\
$\Psi \vdash \sigma\{N/x\} = \sigma\{N'/x\} :: \type$ \` by
functionality\\
$\Psi \vdash M' \, N' : \sigma\{N/x\}$ \` by conversion (and symmetry)
  \end{tabbing} 

\item[Case:] $\m{TMEq}\beta$

  \begin{tabbing}
$\Psi \vdash \lambda x{:}\tau.M : \Pi x{:}\tau.\sigma$, $\Psi
\vdash N : \tau$ and $\Psi , {:}\tau \vdash M :\sigma$ \` by inversion\\
$\Psi \vdash (\lambda x{:}\tau. M)\,N : \sigma\{N/x\}$ \` by $(\Pi
E)$\\
$\Psi \vdash \Pi x{:}\tau.\sigma :: \type$ \` by subderivation lemma\\
$\Psi \vdash \tau :: \type$ and $\Psi , x{:}\tau \vdash \sigma ::
\type$ \` by inversion for products\\
$\Psi \vdash \sigma\{N/x\} :: \type$ \` by substitution\\
$\Psi \vdash M\{N/x\} : \sigma\{N/x\}$ \` by substitution
  \end{tabbing}

\item[Case:] $\m{TMEq}\eta$

  \begin{tabbing}
    $\Psi \vdash M : \Pi x{:}\tau.\sigma$ \` by inversion\\
    $\Psi \vdash \lambda x{:}\tau.(M\,x) : \Pi x{:}\tau,\sigma$ \` by
      $(\Pi E)$, $(\m{var})$ and $(\Pi I)$\\
    $\Psi \vdash \Pi x{:}\tau.\sigma :: \type$ \` by subderivation lemma
  \end{tabbing}

\item[Case:] $\m{TMEq}\{\}$

  \begin{tabbing}
    $\Psi ; \Ga ; \D \vdash P = Q :: c{:}A$ \` by inversion\\
    $\Psi ; \Ga ; \D \vdash P :: c{:}A$, $\Psi ; \Ga ; \D \vdash Q ::
    c{:}A$ and $\Psi \vdash A :: \stype$ \` by i.h.\\
    $\Psi \vdash \monad{c \leftarrow P \leftarrow \dots} : \{\Ga ; \D
    \vdash c{:}A\}$ \` by $\{\}I$\\
    $\Psi \vdash \monad{c \leftarrow Q \leftarrow \dots} : \{\Ga ; \D
    \vdash c{:}A\}$ \` by $\{\}I$\\
    $\Psi ; \Ga \vdash$ and $\Psi ; \D\vdash$ \` by subderivation
    lemma\\
    $\Psi \vdash  \{\Ga ; \D \vdash c{:}A\}$ \` by $\{\}$ well-formedness
  \end{tabbing}

\item[Case:] $\m{TMEq}\{\}\eta$

  \begin{tabbing}
    $\Psi \vdash M : \{\Ga ; \D \vdash c{:}A\}$ \` by inversion\\
    $\Psi \vdash \{\Ga ; \D \vdash c{:}A\}$ \` by subderivation
    lemma\\
    Typing follows straightforwardly
  \end{tabbing}

\item[Case:] $\m{TEqR}$

  \begin{tabbing}
Straightforward by subderivation lemma.
\end{tabbing}

\item[Case:] $\m{TEqS}$ and $\m{TEqT}$

  \begin{tabbing}
Straightforward by i.h.
  \end{tabbing}

\item[Case:] $\m{TEq}\Pi$

  \begin{tabbing}
$\Psi \vdash \tau = \tau' :: \type$ and $\Psi , x{:}\tau \vdash \sigma
= \sigma' :: \type$ \` by inversion\\
$\Psi \vdash \tau :: \type$, $\Psi \vdash \tau' :: \type$ and $\Psi
\vdash \type$ \` by i.h.\\
$\Psi , x{:}\tau \vdash \sigma :: \type$, $\Psi , x{:}\tau \vdash
\sigma' :: \type$ and $\Psi , x{:}\tau \vdash \type$ \` by i.h.\\
$\Psi \vdash \Pi x{:}\tau.\sigma :: \type$ \` $\Pi$ rule\\
$\Psi , x{:}\tau'\vdash \sigma' :: \type$ \` by context conversion\\
$\Psi \vdash \Pi x{:}\tau.\sigma' :: \type$ \` by $\Pi$ rule\\
  \end{tabbing}

\item[Case:] $\m{TEq}\lambda$

  \begin{tabbing}
$\Psi \vdash \tau = \tau' :: \type$ and $\Psi , x{:}\tau \vdash \sigma
= \sigma' :: K$ \` by inversion\\
$\Psi \vdash \tau :: \type$, $\Psi \vdash \tau' :: \type$ and $\Psi
\vdash \type$ \` by i.h.\\
$\Psi , x{:}\tau \vdash \sigma :: K$, $\Psi ,x{:}\tau \vdash \sigma'
:: K$ and $\Psi , x{:}\tau \vdash K$  \` by i.h.\\
$\Psi \vdash \lambda x{:}\tau.\sigma :: \Pi x{:}\tau.K$ \` by
$\lambda$ rule\\
$\Psi , x{:}\tau' \vdash \sigma' :: K$ \` by context conversion\\
$\Psi \lambda x{:}\tau'.\sigma' :: \Pi x{:}\tau'.K$ \` by $\lambda$
rule\\
$\Psi \vdash \lambda x{:}\tau'.\sigma' :: \Pi x{:}\tau.K$  \` by
conversion\\
$\Psi \vdash \Pi x{:}\tau . K$ \` by $\Pi$ well-formedness rule
\end{tabbing}

\item[Case:] $\m{TEqApp}$

  \begin{tabbing}
$\Psi \vdash \tau = \sigma :: \Pi x{:}\tau'.K$ and $\Psi \vdash M = N
: \tau'$ \` by inversion\\
$\Psi \vdash \tau :: \Pi x{:}\tau'.K$, $\Psi \vdash \sigma :: \Pi
x{:}\tau'.K$ and $\Psi \vdash \Pi x{:}\tau'.K$ \` by i.h.\\
$\Psi \vdash M : \tau'$, $\Psi \vdash N : \tau'$ and $\Psi \vdash
\tau' :: \type$ \` by i.h.\\
$\Psi \vdash \tau\, M : K\{M/x\}$ \` by app. wf rule\\
$\Psi \vdash \sigma\, N : K\{N/x\}$ \` by app. wf rule\\
$\Psi , x{:}\tau' \vdash K$ \` by inversion for products\\
$\Psi \vdash K\{M/x\} = K\{N/x\}$ \` by functionality\\
$\Psi \vdash \sigma\, N : K\{M/x\}$ \` by conversion\\
$\Psi \vdash K \{M/x\}$ \` by substitution
  \end{tabbing}

\item[Case:] $\m{TEq}\beta$

  \begin{tabbing}
$\Psi , x{:}\tau \vdash \sigma :: K$ and $\Psi \vdash M : \tau$ \` by
inversion\\
$\Psi \vdash \lambda x{:}\tau . \sigma :: \Pi x{:}\tau.K$ \` by $\Pi$
rule\\
$\Psi \vdash (\lambda x{:}\tau.\sigma)\, M :: K\{M/x\}$ \` by
app. rule\\
$\Psi \vdash \sigma\{M/x\} :: K\{M/x\}$ \` by substitution\\
$\Psi , x{:}\tau \vdash K$ \` by subderivation lemma\\
$\Psi \vdash K \{M/x\}$ \` by substitution
\end{tabbing}

\item[Case:] $\m{TEq}\eta$

  \begin{tabbing}
    $\Psi \vdash \sigma :: \Pi x{:}\tau.K$ \` by inversion\\
    $\Psi \vdash \lambda x{:}\tau.(\sigma\,x) :: \Pi x{:}\tau.K$ \` by
    wf rules\\
    $\Psi \vdash \Pi x{:}\tau.K$ \` by subderivation lemma
  \end{tabbing}

\item[Case:] $\m{TEq}\{\}$

  \begin{tabbing}
Straightforward by i.h.
\end{tabbing}

\item[Case:] (3) is identical to (2), appealing to inversion for
  $\forall\exists$ as needed.

\item[Case:] $\m{PEqRefl}$

\begin{tabbing}
Immediate + subderivation lemma.
\end{tabbing}

\item[Case:] $\m{PEqT}$ and $\m{PEqS}$

  \begin{tabbing}
i.h.
  \end{tabbing}

\item[Case:] $\m{PEqRed}$

  \begin{tabbing}
$\Psi ; \Ga ; \D \vdash P :: z{:}A$, $P\tra{} Q$ and $\Psi ; \Ga ; \D
\vdash Q :: z{:}A$ \` by inversion\\
$\Psi \vdash A :: \stype$ \` by subderivation lemma
\end{tabbing}

\item[Case:] $\m{PEq}\rgt\forall$

  \begin{tabbing}
Straightforward by i.h.
  \end{tabbing}

\item[Case:] $\m{PEq}\lft\forall$

  \begin{tabbing}
    $\Psi \vdash M_0 = M_1 : \tau$ and
    $\Psi ; \Ga ; \D , x{:}A\{M_0/y\} \vdash P_0 = Q_0 :: z{:}C$ \` by
    inversion\\
    $\Psi ; \Ga ; \D , x{:}A\{M_0/y\} \vdash P_0 :: z{:}C$,
        $\Psi ; \Ga ; \D , x{:}A\{M_0/y\} \vdash Q_0 :: z{:}C$\\ and
        $\Psi \vdash C :: \stype$ \` by i.h.\\
    $\Psi \vdash M_0 : \tau$, $\Psi \vdash M_1 : \tau$ and $\Psi
    \vdash \tau :: \type$ \` by i.h.\\
    $\Psi ; \Ga ; \D, x{:}\forall y{:}\tau.A \vdash x \langle M_0
    \rangle.P_0 :: z{:}C$ \` by $\lft\forall$\\
    $\Psi ;  \D, x{:}\forall y{:}\tau.A \vdash$ \` by subderivation
    lemma\\
    $\Psi \vdash \forall y{:}\tau.A :: \stype$ \` by definition\\
    $\Psi , y{:}\tau \vdash A :: \stype$ \` by inversion for
    $\forall\exists$\\
    $\Psi \vdash A\{M_0/y\} = A\{M_1/y\} :: \stype$ \` by
    functionality\\
    $\Psi \vdash A\{M_1/y\} :: \stype$ \` by substitution\\
    $\Psi ; \Ga ; \D , x{:}A\{M_1/y\} \vdash Q_0 :: z{:}C$ \` by
    context conversion rule\\
    $\Psi ; \Ga ; \D , x{:}\forall y{:}\tau.A \vdash x\langle M_1
    \rangle . Q_0 :: z{:}C$ \` by $\lft\forall$
  \end{tabbing}

\item[Case:] $\rgt{\m{PEqConv}}$

  \begin{tabbing}
$\Psi ; \Ga ; \D \vdash P = Q :: z{:}A$ and $\Psi \vdash A = B ::
\stype$ \` by inversion\\
$\Psi ; \Ga ; \D \vdash P :: z{:}A$, $\Psi ; \Ga ; \D \vdash Q ::
z{:}A$,\\
$\Psi \vdash A :: \stype$ and $\Psi \vdash B :: \stype$ \` by i.h.\\
$\Psi ; \Ga ; \D \vdash P :: z{:}B$ \` by $\rgt{\m{PEqConv}}$\\
$\Psi ; \Ga ; \D \vdash Q :: z{:}B$ \` by $\rgt{\m{PEqConv}}$\\
  \end{tabbing}

Remaining cases are identical.  
  
\end{description}
\end{proof}

\begin{theorem}[Functionality for Equality]
  Assume $\Psi \vdash M = N : \tau$:

  \begin{enumerate}
\item If $\Psi , x{:}\tau \vdash M_0 = M_1 : \sigma$ then $\Psi \vdash
M_0\{M/x\} = M_1\{N/x\} : \sigma\{M/x\}$
\item If $\Psi , x{:}\tau \vdash \sigma_1 = \sigma_2 :: K$ then $\Psi
  \vdash \sigma_1 \{M/x\} = \sigma_2 \{N/x\} :: K\{M/x\}$
\item If $\Psi ,x{:}\tau \vdash A = B :: K$ then $\Psi \vdash A
  \{M/x\} = B\{N/x\} :: K\{M/x\}$
\item If $\Psi , x{:}\tau \vdash K_1 = K_2$ then $\Psi \vdash
  K_1\{M/x\} = K_2\{N/x\}$
\item If $\Psi , x{:}\tau ; \Ga ; \D \vdash P = Q :: z{:}A$ then $\Psi
  ; \Ga\{M/x\} ; \D\{M/x\} \vdash P\{M/x\} = Q\{N/x\} :: z{:}A\{M/x\}$
  \end{enumerate}
\end{theorem}

\begin{proof}
  {\bf (1)}

  \begin{tabbing}
    $\Psi , x{:}\tau \vdash M_0 = M_1 : \sigma$ \` assumption\\
    $\Psi \vdash M = N : \tau$ \` assumption\\
    $\Psi \vdash M : \tau$, $\Psi \vdash N : \tau$ and $\Psi \vdash
    \tau :: \type$ \` by validity\\
    $\Psi , x{:}\tau \vdash M_0 : \sigma$,     $\Psi , x{:}\tau \vdash
    M_1 : \sigma$ and $\Psi , x{:}\tau \vdash \sigma :: \type$ \` by
    validity\\
    $\Psi \vdash M_0\{M/x\} = M_1\{M/x\} : \sigma \{M/x\}$ \` by
    substitution\\
    $\Psi \vdash M_1 \{M/x\} = M_1\{N/x\} : \sigma\{M/x\}$ \` by
    functionality\\
    $\Psi \vdash M_0 \{M/x\} = M_1\{N/x\} : \sigma\{M/x\}$ \` by transitivity
  \end{tabbing}

  {\bf (2)}

  \begin{tabbing}
    $\Psi , x{:} : \tau ; \Ga ; \D \vdash P = Q :: z{:}A$ \` assumption\\
    $\Psi \vdash M = N : \tau$ \` assumption\\
    $\Psi \vdash M : \tau$, $\Psi \vdash N : \tau$ and $\Psi \vdash
    \tau :: \type$ \` by validity\\
    $\Psi , x{:}\tau ; \Ga ; \D \vdash P :: z{:}A$, $\Psi , x{:}\tau ;
    \Ga ; \D \vdash Q :: z{:}A$\\
    and $\Psi , x{:}\tau \vdash A :: \stype$ \` by validity\\
    $\Psi ; \Ga \{M/x\} ; \D\{M/x\} \vdash P\{M/x\} = Q \{M/x\} ::
    z{:}A\{M/x\}$ \` by substitutition\\
    $\Psi ; \Ga \{M/x\} ; \D\{M/x\} \vdash Q\{M/x\} = Q\{N/x\} ::
    z{:}A\{M/x\}$ \` by functionality\\
    $\Psi ; \Ga \{M/x\} ; \D\{M/x\} \vdash P\{M/x\} = Q\{N/x\} ::
    z{:}A\{M/x\}$ \` by transitivity
  \end{tabbing}

 Remaining cases are identical, appealing to validity, substitution
 and functionality of typing.
\end{proof}

We omit the analogue functionality property for type substitution.

\begin{lemma}[Inversion]

  \begin{enumerate}
\item If $\Psi \vdash x {:} \tau$ then $x{:}\sigma \in \Psi$ with
  $\Psi \vdash \tau = \sigma :: \type$
\item If $\Psi \vdash M_1 \, M_2 : \sigma$ then $\Psi \vdash M_1 : \Pi
  x{:}\tau_1.\tau_2$, $\Psi \vdash M_2 : \tau_1$ and $\Psi \vdash
  \sigma\{M_2/x\} = \tau_2 :: \type$
\item If $\Psi \vdash \lambda x{:}\tau.M : \sigma$ then $\Psi \vdash
  \sigma = \Pi x{:}\tau.\sigma' :: \type$, $\Psi \vdash \tau :: \type$
  and $\Psi , x{:}\tau \vdash M : \sigma'$
\item If $\Psi \vdash \Pi x{:}\tau_1.\tau_2 :: K$ then $\Psi \vdash K
  = \type$, $\Psi \vdash \tau_1 :: \type$ and $\Psi , x{:}\tau_1
  \vdash \tau_2 :: \type$
\item If $\Psi \vdash \lambda x{:}\tau.\sigma :: K$ then $\Psi \vdash
  K = \Pi x{:}\tau.K'$, $\Psi \vdash \tau :: \type$ and $\Psi,
  x{:}\tau \vdash \sigma :: K'$
\item If $\Psi \vdash \tau\, M :: K$ then $\Psi \vdash \tau :: \Pi
x{:}\tau_0.K_1$, $\Psi \vdash M : \tau_0$ and $\Psi \vdash K =
K_1\{M/x\}$
\item If $\Psi \vdash \lambda t{::}K .\tau :: K'$ then $\Psi \vdash
  K' = \Pi t{::}K.K''$, $\Psi \vdash K$ and $\Psi,
  t{::}K \vdash \tau :: K''$
\item If $\Psi \vdash \tau\, \sigma :: K$ then $\Psi \vdash \tau :: \Pi
t{::}K_0.K_1$, $\Psi \vdash \sigma : K_0$ and $\Psi \vdash K =
K_1\{\sigma/t\}$

\item If $\Psi \vdash \Pi x{:}\tau.K$ then $\Psi \vdash \tau :: \type$
  and $\Psi , x{:}\tau \vdash K$

\item If $\Psi \vdash \Pi t{::}K_1.K_2$ then $\Psi \vdash K_1$
  and $\Psi , t{::}K_1 \vdash K_2$

\item If $\Psi \vdash \monad{\Ga ; \D \vdash c{:}A} :: K$ then $\Psi
  \vdash K = \type$, $\Psi \vdash \Ga :: \stype$, $\Psi \vdash \D ::
  \stype$ and $\Psi \vdash A :: \stype$
\item If $\Psi ; \Ga ; \D \vdash z(x{:}\tau).P :: z{:}A$ then $\Psi
  \vdash A = \forall x{:}\tau.A'$ and $\Psi \vdash \tau :: \stype$ and
  $\Psi , x{:}\tau ; \Ga ; \D \vdash P :: z{:}A'$
\item If $\Psi ; \Ga ; \D , x{:}A \vdash x\langle M \rangle_{\forall x{:}\tau.A'}.P ::
  z{:}C$ then $\Psi \vdash A = \forall y{:}\tau.A' :: \stype$, $\Psi
  \vdash \tau :: \type$, $\Psi \vdash M : \tau$ and $\Psi ; \Ga ; \D ,
  x{:}A'\{M/y\} \vdash P :: z{:}C$
\item If $\Psi ; \Ga ; \D \vdash z\langle M \rangle_{\exists x{:}\tau.A'}.P :: z{:}A$ then
  $\Psi \vdash A = \exists x{:}\tau.A' :: \stype$, $\Psi \vdash \tau
  :: \type$ and $\Psi , y{:}\tau ; \Ga ; \D , x{:}A' \vdash P ::
  z{:}C$
\item If $\Psi \vdash \forall x{:}\tau . A :: K$ then $\Psi \vdash K =
  \stype$, $\Psi \vdash \tau :: \type$, $\Psi , x{:}\tau \vdash A ::
  \stype$
\item If $\Psi \vdash \exists x{:}\tau . A :: K$ then $\Psi \vdash K =
  \stype$, $\Psi \vdash \tau :: \type$, $\Psi , x{:}\tau \vdash A ::
  \stype$
\item If $\Psi \vdash \lambda x{:}\tau.A :: K$ then $\Psi \vdash K =
  \hPi x{:}\tau.K'$, $\Psi \vdash \tau :: \type$ and $\Psi , x{:}\tau
  \vdash A :: K'$
\item If $\Psi \vdash A\,M :: K$ then $\Psi \vdash A :: \hPi
  x{:}\tau_0.K'$, $\Psi \vdash M : \tau_0$ and $\Psi \vdash K =
  K'\{M/x\}$
\item If $\Psi \vdash \lambda t{::}K.A :: K'$ then $\Psi \vdash K' =
  \Pi t{::}K.K''$, $\Psi \vdash K$ and $\Psi , t{::}K
  \vdash A :: K''$
\item If $\Psi \vdash A\, B :: K$ then $\Psi \vdash A :: \Pi
  t{:}K_0.K_1$, $\Psi \vdash B :: K_0$ and $\Psi \vdash K = K_1\{B/t\}$
  
\end{enumerate}
\end{lemma}

\begin{proof}
By induction on the given derivation. Most cases require validity.
\end{proof}

\begin{theorem}[Equality Inversion]
  \begin{enumerate}
\item If $\Psi \vdash \tau = \Pi x{:}\tau_0.\tau_1 :: \type$ then
  $\Psi \vdash \tau = \Pi x{:}\sigma_0.\sigma_1 :: \type$ with $\Psi \vdash \sigma_0 =
  \tau_0 :: \type$ and $\Psi , x{:}\sigma_0 \vdash \sigma_1 = \tau_1
  :: \type$
\item If $\Psi \vdash K = \type$ then $K = \type$
\item If $\Psi \vdash K = \Pi x{:}\tau_0.K'$ then $\Psi \vdash K = \Pi
  x{:}\sigma_0.K''$ with $\Psi \vdash \sigma_0 = \tau_0 :: \type$ and
  $\Psi ,x{:}\sigma_0 \vdash K'' = K'$
\item If $\Psi \vdash K = \Pi t {::}K_1.K_2$ then $\Psi \vdash K =
  \Pi t{::}K_1'.K_2'$ with $\Psi \vdash K_1' = K_1$ and $\Psi ,
  t{::}K_1' \vdash K_2' = K_2$
  
\item $\Psi \vdash A = \forall x{:}\tau_0.A_0 :: \stype$ then $\Psi
  \vdash A =
  \forall x{:}\sigma_0.B_0 :: \stype$ with $\Psi \vdash \sigma_0 = \tau_0 :: \type$ and
  $\Psi , x{:}\sigma_0 \vdash B_0 = A_0 :: \stype$
\item $\Psi \vdash A = \exists x{:}\tau_0.A_0 :: \stype$ then $\Psi
  \vdash A =
  \exists x{:}\sigma_0.B_0 :: \stype$ with $\Psi \vdash \sigma_0 = \tau_0 :: \type$ and
  $\Psi , x{:}\sigma_0 \vdash B_0 = A_0 :: \stype$
\item $\Psi \vdash \tau = \lambda x{:}\tau_0 . \sigma :: K$ then $\Psi
  \vdash \tau
  = \lambda x{:} \tau_1 . \sigma' :: \Pi x{:}\tau_1.K_0$ with $\Psi \vdash \tau_1 = \tau_0
  :: \type$ and $\Psi , x{:}\tau_1 \vdash \sigma' = \sigma :: K_0$,
  for some $K_0$
\item $\Psi \vdash \tau = \tau_0\,M :: K$ then
  $\Psi \vdash \tau = \tau_1\, N :: K$ with $\Psi \vdash \tau_1 = \tau_0 :: \Pi
  x{:}\sigma . K_0$, $\Psi \vdash N = M : \sigma$ and $K = K_0\{N/x\}$

\item $\Psi \vdash \tau = \lambda t{::}K . \sigma :: K'$ then $\Psi
  \vdash \tau
  = \lambda t{::} K_0 . \sigma' :: \Pi t{::}K_0.K''$ with $\Psi \vdash
  K_0 = K
 $ and $\Psi , t{::}K_0 \vdash \sigma' = \sigma :: K''$,
  for some $K''$
\item $\Psi \vdash \tau = \tau_0\,\sigma_0 :: K$ then
  $\Psi \vdash \tau = \tau_1\, \sigma_1 :: K$ with $\Psi \vdash \tau_1 = \tau_0 :: \Pi
  t{::}K_1 . K_0$, $\Psi \vdash \sigma_1 = \sigma_0 :: K_1$ and $K = K_0\{\sigma_1/t\}$
  
\item $\Psi \vdash A = \lambda x{:}\tau_0 . A_0 :: K$ then $\Psi
  \vdash A
  = \lambda x{:} \tau_1 . A_0' :: \Pi x{:}\tau_1.K_0$ with $\Psi \vdash \tau_1 = \tau_0
  :: \type$ and $\Psi , x{:}\tau_1 \vdash A_0' = A_0 :: K_0$,
  for some $K_0$
\item $\Psi \vdash A = A_0\,M :: K$ then
  $\Psi \vdash A = A_1\, N$ with $\Psi \vdash A_1 = A_0 :: \Pi
  x{:}\sigma . K_0$, $\Psi \vdash N = M : \sigma$ and $K = K_0\{N/x\}$

\item $\Psi \vdash B = \lambda t{::}K . A :: K'$ then $\Psi
  \vdash B
  = \lambda t{::} K_0 . A' :: \Pi t{::}K_0.K''$ with $\Psi \vdash K_0 
 = K$ and $\Psi , t{::}K_0 \vdash A' = A :: K''$,
  for some $K''$
\item $\Psi \vdash A = A_0\,B_0 :: K$ then
  $\Psi \vdash A = A_1\, B_1 :: K$ with $\Psi \vdash A_1 = A_0 :: \Pi
  t{::}K_1 . K_0$, $\Psi \vdash B_1 = B_0 :: K_1$ and $K = K_0\{B_1/t\}$
\end{enumerate}
\end{theorem}

\begin{proof}
  By induction on the given equality derivations.

    {\bf (1)}
  \begin{description}

  \item[Case:] $\m{TEqT}$

    \begin{tabbing}
$\Psi \vdash \tau = \tau' :: \type$ and $\Psi \vdash \tau' = \Pi
x{:}\tau_0.\tau_1 :: \type$ \` assumption\\
 $\tau' = \Pi x{:}\tau'_0.\tau'_1$ with $\Psi \vdash \tau'_0 =
  \tau_0 :: \type$ and $\Psi , x{:}\tau'_0 \vdash \tau'_1 = \tau_1
  :: \type$ \` by i.h.\\
$\tau = \Pi x{:}\sigma_0.\sigma_1$ with $\Psi \vdash \sigma_0 = \tau_0' ::
\type$ and $\Psi , x{:}\sigma_0 \vdash \sigma_1 = \tau'_1 :: \type$ \`
by i.h.\\
$\Psi \vdash \sigma_0 = \tau_0 :: \type$ \` by transivitity\\
$\Psi , x{:}\sigma_0 \vdash \tau_1' = \tau_1 :: \type$ \` by context
conversion\\
$\Psi , x{:}\sigma_0 \vdash \sigma_1 = \tau_1 :: \type$ \` by
transitivity
\end{tabbing}

\item[Case:] $\m{TEq}\beta$

  \begin{tabbing}
$\Psi , y{:}\tau \vdash \sigma :: \type$ and $\Psi \vdash
M : \tau$, $K\{M/y\} = \type$ and\\
$\Pi x{:}\tau_0.\tau_1 = \sigma\{M/y\}$ \` by inversion\\
$\Psi \vdash (\lambda y{:}\tau.\sigma)\, M = \Pi x{:}\tau_0.\tau_1 :: \type$ \`
assumption\\
$\sigma = \Pi x{:}\tau.\sigma$ \` by definition of substitution\\
$\Psi , y{:}\tau \vdash \Pi x{:}\tau . \sigma :: \type$ \` by def.\\
$\Psi \vdash (\lambda y{:}\tau.\sigma)\, M = \Pi
x{:}\tau.\sigma\{M/y\} :: \type$ \` by rule\\
$\Psi \vdash \tau :: \type$ \` by validity\\
$\Psi \vdash \tau = \tau :: \type$ \` by reflexivity\\
$\Psi \vdash \sigma\{M/y\} :: \type$ \` by substitution\\
$\Psi \vdash \sigma\{M/y\} = \sigma\{M/y\} :: \type$ \` by reflexivity
  \end{tabbing}

The other cases follow similar patterns.
  
  \end{description}
  
\end{proof}

\begin{lemma}[Injectivity of Products]
  \begin{enumerate}
\item If $\Psi \vdash \Pi x{:}\tau.\sigma = \Pi x{:}\tau'.\sigma' ::
  \type$ then $\Psi \vdash \tau = \tau' :: \type$ and $\Psi , x{:}\tau
  \vdash \sigma = \sigma' :: \type$
\item If $\Psi \vdash \Pi x{:}\tau_1 . K_1 = \Pi x{:}\tau_2 . K_2$
  then $\Psi \vdash \tau_1 = \tau_2 :: \type$ and $\Psi , x{:}\tau_1
  \vdash K_1 = K_2$
\item If $\Psi \vdash \forall x{:}\tau_1.A_1 = \forall x{:}\tau_2
  . A_2 :: \stype$ then $\Psi \vdash \tau_1 = \tau_2 :: \type$ and
  $\Psi , x{:}\tau_1 \vdash A_1 = A_2 :: \stype$
  \end{enumerate}
\end{lemma}

\begin{proof}
  By equality inversion.
  
\end{proof}

\unicitythm*


 

\begin{proof}
By induction on the structure of the given term/type/process.

\begin{description}
\item[Case:] $M$ is $\lambda x{:}\tau_0.M'$

  \begin{tabbing}
$\Psi , x{:}\tau_0 \vdash M' : \sigma$, $\Psi , x{:}\tau_0 \vdash
M' : \sigma'$ with $\Psi \vdash  \tau = \Pi x{:}\tau_0.\sigma :: \type$, \\
$\Psi \vdash \tau_0 :: \type$ and $\Psi \vdash \tau' = \Pi x{:}\tau_0.\sigma' :: \type$ \` by inversion\\
$\Psi , x{:}\tau_0 \vdash \sigma = \sigma' :: \type$ \` by i.h.\\
$\Psi \vdash \Pi x{:}\tau_0.\sigma = \Pi x{:}\tau_0.\sigma' :: \type$
\` by $\m{TEq}\Pi$ rule
\end{tabbing}

\item[Case:] $M$ is $M'\,N'$
\begin{tabbing}
$\Psi \vdash M' N' : \tau$ and $\Psi \vdash M' N' : \tau'$ \`
assumption\\
$\Psi \vdash M' : \Pi x{:}\tau_0.\sigma_0$ and $\Psi \vdash M' : \Pi
x{:}\tau_0'.\sigma_0'$ with $\Psi \vdash \tau = \sigma_0\{N'/x\} :: \type$,\\
$\Psi \vdash N' : \tau_0$, $\Psi \vdash N' : \tau_0'$ and $\Psi \vdash
\tau' = \sigma_0'\{N'/x\} :: \type$ \` by inversion\\
$\Psi \vdash \Pi x{:}\tau_0.\sigma_0 = \Pi x{:}\tau_0' . \sigma_0' ::
\type$ \` by i.h.\\
$\Psi \vdash \tau_0 = \tau_0' :: \type$ \` by i.h.\\
$\Psi , x{:}\tau_0 \vdash \sigma_0 = \sigma_0' :: \type$ \` by
injectivity\\
$\Psi \vdash \sigma_0\{N'/x\} = \sigma_0'\{N'/x\} :: \type$ \` by functionality
\end{tabbing}

\item[Case:] $M$ is $\{ c \leftarrow P \leftarrow \ov{u_j};\ov{d_i}\}$
\begin{tabbing}
  $\Psi \vdash M : \{\Ga ; \D \vdash c{:}A\}$ and $\Psi \vdash M : \{
  \Ga ; \D \vdash c{:}A'\}$ \` assumption\\
 $\Psi ; \Ga ; \D \vdash P :: c{:}A$ and  $\Psi ; \Ga ; \D \vdash P ::
 c{:}A'$ \` by inversion\\
 $\Psi \vdash A = A' :: \stype$ \` by i.h.\\
 Conclude by reflexivity and $\m{TEq}\{\}$
\end{tabbing}

\item[Case:] $M$ is $x$

Direct by inversion lemma.

{\bf (2)}

\item[Case:] $\tau$ is $\Pi x{:}\tau_0.\sigma_0$

  \begin{tabbing}
$\Psi \vdash \Pi x{:}\tau_0.\sigma_0 :: K$ and $\Psi \vdash \Pi
x{:}\tau_0.\sigma_0 :: K'$ \` assumption\\
$\Psi , x{:}\tau_0 \vdash \sigma_0 :: \type$ and $\Psi , x{:}\tau_0
\vdash \sigma_0 ::\type$ and $K = K' = \type$
\\\` by inversion lemma\\
  \end{tabbing}

\item[Case:] $\tau$ is $\lambda x{:}\tau_0.\sigma_0$

  \begin{tabbing}
$\Psi \vdash \lambda x{:}\tau_0.\sigma_0 :: K$ and $\Psi \vdash
\lambda x{:}\tau_0.\sigma_0 :: K'$ \` assumption\\
$\Psi , x{:}\tau_0 \vdash \sigma_0 :: K_0$, $\Psi , x{:}\tau_0 \vdash
\sigma_0 :: K_0'$, $\Psi \vdash \tau_0 :: \type$\\ with $K = \Pi x{:}\tau_0.K_0$ and $K' = \Pi
x{:}\tau_0.K_0'$ \` by inversion lemma\\
$\Psi , x{:}\tau_0 \vdash K_0 = K_0'$ \` by i.h.\\
$\Psi \vdash \Pi x{:}\tau_0 . K_0 = \Pi x{:}\tau_0 . K_0'$ \` by rule
  \end{tabbing}

\item[Case:] $\tau$ is $\tau_0\,M$

  \begin{tabbing}
$\Psi \vdash \tau_0 \, M :: K$ and $\Psi \vdash \tau_0 \, M :: K'$ \`
assumption\\
$\Psi \vdash \tau_0 :: \Pi x{:}\sigma.K_0$ and $\Psi \vdash \tau_0 ::
\Pi x{:}\sigma'.K_0'$, $\Psi \vdash M : \sigma$ and $\Psi \vdash M : \sigma'$ \\
with $K = K_0\{M/x\}$ and $K' = K_0'\{M/x\}$ \` by inversion lemma\\
$\Psi \vdash \Pi x{:}\sigma.K_0 = \Pi x{:}\sigma'.K_0'$ \` by i.h.\\
$\Psi \vdash \sigma = \sigma' :: \type$ \` by i.h.\\
$\Psi , x{:}\sigma \vdash K_0 = K_0'$ \` by injectivity\\
$\Psi \vdash K_0\{M/x\} = K_0'\{M/x\}$ \` by substitution
  \end{tabbing}

\item[Case:] $\tau$ is $\monad{\Ga ; \D \vdash c{:}A}$

  Straightforward by i.h.

  {\bf (3)}
  
\item[Case:] $P$ is $z(x).P_0$

  \begin{tabbing}
    $\Psi ; \Ga ; \D \vdash z(x{:}\tau_0).P_0 :: z{:}A$ and
    $\Psi ; \Ga ; \D \vdash z(x{:}\tau_0).P_0 :: z{:}A'$ \`
    assumption\\
    $A = \forall x{:}\tau_0.A_0$, $A' = \forall x{:}\tau_0.A_0'$,
    $\Psi , x{:}\tau_0 ; \Ga ; \D \vdash P_0 :: z{:}A_0$,\\ 
    $\Psi , x{:}\tau_0 ; \Ga ; \D \vdash P_0 :: z{:}A_0'$,
    and $\Psi \vdash \tau_0 :: \type$
    \` by
    inversion lemma\\
    $\Psi , x{:}\tau_0 \vdash A_0 = A_0'$ \` by i.h.\\
    $\Psi \vdash \forall x{:}\tau_0 A_0 = \forall x{:}\tau_0 A_0'$ \`
    by rule
  \end{tabbing}

\item[Case:] $P$ is $x\langle M \rangle_{\forall x{:}\tau_0.A_0}.P_0$

  \begin{tabbing}
    $\Psi ; \Ga ; \D , x{:}A \vdash x\langle M \rangle_{\forall x{:}\tau_0.A_0}.P_0 :: z{:}C$
    and\\     $\Psi ; \Ga ; \D , x{:}A \vdash x\langle M \rangle.P_0 ::
    z{:}C$ \` assumption\\
    $A = \forall x{:}\tau_0.A_0$, 
    $\Psi \vdash M : \tau_0$, 
    $\Psi ; \Ga ; \D , x{:} A_0\{M/x\} \vdash P_0 :: z{:}C$ \\and
    $\Psi ; \Ga ; \D , x{:} A_0\{M/x\} \vdash P_0 :: z{:}C'$ \` by
    inversion lemma\\
    $\Psi \vdash C = C' :: \stype$ \` by i.h.
  \end{tabbing}

\item[Case:] $P$ is $z\langle M \rangle_{\exists x{:}\tau_0.A_0}.P_0$

  \begin{tabbing}
    $\Psi ; \Ga ; \D \vdash z\langle M \rangle.P_0 :: z{:}A$ and
    $\Psi ; \Ga ; \D \vdash z\langle M \rangle.P_0 :: z{:}A'$ \`
    assumption\\
    $A = \exists x{:}\tau_0.A_0$, $A' = \exists x{:}\tau_0.A_0$,
    $\Psi \vdash M : \tau_0$, \\
    $\Psi ; \Ga ; \D \vdash P_0 :: z{:}A_0\{M/x\}$ and
    $\Psi ; \Ga ; \D \vdash P_0 :: z{:}A_0\{M/x\}$ \\\` by inversion
    lemma
   
  \end{tabbing}

Remaining cases follow similarly.

\end{description}

\end{proof}

\begin{theorem}\label{thm:sreq}
If $\Psi \vdash M : \tau$ and $M \tra{} M'$ then $\Psi \vdash M = M' : \tau$
\end{theorem}
\begin{proof}
By induction on $\tra{}$ relation.

\begin{description}
\item[Case:] 
\[
\inferrule[]
{M \tra{} M'}
{M\, N \tra{} M'\,N}
\]
\begin{tabbing}
  $\Psi \vdash M : \Pi x{:}\tau_0.\sigma_0$,
  $\Psi \vdash \tau_0 :: \type$, $\Psi \vdash N : \tau_0$ and $\Psi
  \vdash M\,N : \sigma_0\{N/x\}$\\ \` by inversion lemma\\
  $\Psi \vdash M = M' : \Pi x{:}\tau_0.\sigma_0$ \` by i.h.\\
  $\Psi \vdash \Pi x{:}\tau_0.\sigma_0 :: \type$ \` by validity\\
  $\Psi \vdash N = N : \tau_0$ \` by reflexivity\\
  $\Psi \vdash M\,N = M'\, N : \sigma_0\{N/x\}$ \` by $\m{TMEq}\Pi$
\end{tabbing}

\item[Case:]
  \[
\inferrule[]
{N \tra{} N'}
{M\, N \tra{} M\,N'}
  \]
  \begin{tabbing}
  $\Psi \vdash M : \Pi x{:}\tau_0.\sigma_0$,
  $\Psi \vdash \tau_0 :: \type$, $\Psi \vdash N : \tau_0$ and $\Psi
  \vdash M\,N : \sigma_0\{N/x\}$\\ \` by inversion lemma\\
  $\Psi \vdash N = N' : \tau_0$ \` by i.h.\\
  $\Psi \vdash M = M : \Pi x{:}\tau_0.\sigma_0$ \` by reflexivity\\
  $\Psi \vdash M\,N = M\,N' :  \sigma_0\{N/x\}$ \` by $\m{TMEq}\Pi$
  \end{tabbing}

  \[
    \inferrule[]
    { }
    { (\lambda x{:}\tau_0.M_0)\,N_0 \tra{} M_0\{N_0/x\}  }
  \]

  \begin{tabbing}
    $\Psi \vdash \lambda x{:}\tau_0.M_0 : \Pi x{:}\tau_0.\sigma_0$,
    $\Psi \vdash \tau_0 :: \type$, $\Psi \vdash N_0 : \tau_0$,\\
    $\Psi \vdash (\lambda x{:}\tau_0.M_0)\,N_0 : \sigma_0\{N_0/x\}$   \` by
    inversion lemma\\
    $\Psi , x{:}\tau_0 \vdash M_0 : \sigma_0$ \` by inversion lemma\\
    $\Psi \vdash (\lambda x{:}\tau_0.M_0)\,N_0 = M_0\{N_0/x\} :
    \sigma_0\{N/x\}$ \` by $\m{TMEq}\beta$
  \end{tabbing}
  
\end{description}
\end{proof}

\srlamterms*
\begin{proof}
Immediate from Theorem~\ref{thm:sreq} and validity for equality.
\end{proof}

\srlamproc*
\begin{proof}
  The proof follows by Theorem~\ref{thm:sreq} and
  a series of lemmas that relate typed processes
and their reducts under a cut (which now crucially rely on the
inversion lemmas and validity). See \cite{depsesstr,toninhothesis,DBLP:journals/mscs/CairesPT16}.
\end{proof}

\proglamthm*
\begin{proof}
By induction on typing, using the standard canonical forms-based
reasoning and noting that monadic terms are values.
\end{proof}

\section{Appendix -- Embedding}
\label{app:enc}

\begin{lemma}[Compositionality]
~\label{lem:comp}
\begin{enumerate}
\item $\Psi \vdash \lb K\{M/x\}\rb$ iff $\Psi \vdash \lb K \rb
  \{\monad{\lb M \rb_c}/x\}$
\item $\Psi \vdash \lb K_1 \{\tau / t \}\rb$ iff $\Psi \vdash \lb K_1\rb
  \{\lb \tau \rb / t\}$
\item $\Psi \vdash \lb K_1\{A/x\}\rb$ iff $\Psi \vdash \lb K_1\rb\{\lb
  A \rb /x\}$
\item $\Psi \vdash \lb \tau\{M/x\}\rb :: \lb K\{M/x\}\rb$ iff 
$\Psi \vdash \lb \tau\rb\{\monad{\lb M \rb_c}/x\} :: \lb K \rb
\{\monad{\lb M \rb_c}/x\}$
\item $\Psi \vdash \lb A\{M/x\}\rb :: \lb K\{M/x\}\rb$ iff 
$\Psi \vdash \lb A\rb\{\monad{\lb M \rb_c}/x\} :: \lb K \rb
\{\monad{\lb M \rb_c}/x\}$
\item $\Psi ; \Ga ; \D  \vdash \lb M \{N/x\}\rb_z = \lb M\rb_z\{\monad{\lb N \rb_y}/x\} :: z{:}\lb
    A\{N/x\}\rb$ 
  \item $\Psi ; \Ga ; \D \vdash \lb P\{M/x\}\rb :: z{:}\lb A
    \{M/x\}\rb$
     iff $\Psi ; \Ga ; \D \vdash \lb P \rb\{\monad{\lb M \rb_c}/x\}::
     z{:}\lb A \rb\{\monad{\lb M \rb_c}/x\}$
\end{enumerate}

\end{lemma}

\begin{proof}
By mutual induction on the structure of the given kind/type/session
type/etc.

\begin{description}

\item[Case:] $K = \m{type}$ or $K = \m{stype}$
\begin{tabbing}
Trivial.
\end{tabbing}
\noindent {\bf (1)}
\item[Case:] $K = \Pi y{:}\tau.K'$

\begin{tabbing}
{\bf Subcase: $\Rightarrow$} \\
$\lb \Pi y{:}\tau.K'\{M/x\} \rb = 
 \lb \Pi y{:}\tau\{M/x\}.K'\{M/x\}\rb = \Pi y{:}\{\lb \tau
 \{M/x\}\rb\}.\lb K'\{M/x\}\rb$ \` by definition\\
$\Pi y{:}\{\lb\tau\rb\{\{\lb M\rb_c\}/x\}\}.\lb K' \rb\{\{\lb
M\rb_c\}/x\}$ \` by i.h.(3) and i.h.(1)\\
$= (\Pi y{:}\{\lb\tau\rb\}.\lb K'\rb) \{\{\lb
M\rb_c\}/x\}$ \` by definition, satisfying $\Rightarrow$\\
{\bf Subcase: $\Leftarrow$}\\
$\lb \Pi y{:}\tau.K'\rb \{\{\lb M\rb_c\}/x\} =
 (\Pi y{:} \{\lb \tau \rb\}.\lb K'\rb)\{\{\lb M\rb_c\}/x\} =$\\
$ \Pi y{:}\{\lb \tau \rb\}\{\{\lb M\rb_c\}/x\} .\lb K'\rb \{\{\lb M\rb_c\}/x\}$
\` by definition\\
$\Pi y{:}\{\lb\tau\{M/x\}\rb\}.\lb K'\{M/x\}\rb$ \` by i.h.(3) and
i.h.(1)\\
$= \lb\Pi y{:}\tau.K' \{M/x\}\rb$ \` by definition, satisfying $\Leftarrow$
\end{tabbing}

\item[Case:] $K = \Pi t{:}K_1 . K_2$

\begin{tabbing}
Same argument as above, appealing to i.h.(1)
\end{tabbing}

\noindent {\bf (2)}

As above, appealing to i.h.(2)

\noindent {\bf (3)}

As above, appealing to i.h.(3)

\noindent {\bf (4)}

\item[Case:] $\tau = \Pi y{:}\tau'.\sigma$

\begin{tabbing}
{\bf Subcase: $\Rightarrow$}\\
$\lb \Pi y{:}\tau'.\sigma\{M/x\}\rb = 
 \lb \Pi y{:}\tau'\{M/x\}.\sigma\{M/x\}\rb =
  \forall y{:}\{\lb\tau'\{M/x\}\rb\}.
  \lb\sigma\{M/x\}\rb  $ 
\` by definition\\
$\forall y{:}\{\lb\tau'\rb\{\{\lb M \rb_c\}/x\}\}.
  \lb\sigma\rb\{\{\lb M
 \rb_c\}/x\}$ \` by i.h.(3)\\
$= (\forall y{:}\{\lb\tau'\rb\}.
\lb\sigma\rb)\{\{\lb M \rb_c\}/x\}$ \` by definition, satisfying
$\Rightarrow$\\
{\bf Subcase: $\Leftarrow$}\\
$\lb\Pi y{:}\tau'.\sigma\rb\{\{\lb M \rb_c\}/x\} =
(\forall y{:}\{\lb\tau'\rb\}.
\lb\sigma\rb)\{\{\lb M \rb_c\}/x\} =
\forall y{:}\{\lb\tau'\rb\{\{\lb M \rb_c\}/x\}\}.
\lb\sigma\rb\{\{\lb M \rb_c\}/x\}$\\ \` by definition\\
$\forall y{:}\{\lb\tau'\{M/x\}\rb\}.
\lb\sigma\{M/x\}\rb$ \` by i.h.(3)\\
$ = \lb\Pi y{:}(\tau'\{M/x\}).(\sigma\{M/x\})\rb
  = \lb \Pi y{:}\tau'.\sigma \{M/x\}\rb$ \` by definition, satisfying $\Leftarrow$
\end{tabbing}

\item[Case:] $\tau = \lambda y{:}\tau'.\sigma$

\begin{tabbing}
As above.
\end{tabbing}

\item[Case:] $\tau = \tau'\,M'$

\begin{tabbing}
$\lb\tau'\,M' \{M/x\}\rb = \lb \tau'\{M/x\}\,M'\{M/x\}\rb
= \lb\tau'\{M/x\}\rb\,\{\lb M'\{M/x\}\rb_c\}$ \` by definition\\
$\lb\tau'\rb\{\{\lb M\rb_d\}/x\} \, 
\{\lb M'\rb_c\{\monad{\lb M \rb_d}/x\} \}$
\` by i.h.\\
$\lb \tau'\,M'\rb\{\monad{\lb M\rb_d}/x\} = 
( \lb \tau'\rb\,\monad{\lb M'\rb_c}) \{\monad{\lb M\rb_d}/x\} = $
$\lb\tau'\rb\{\{\lb M\rb_d\}/x\} \, 
\{\lb M'\rb_c\{\monad{\lb M \rb_d}/x\} \}$\\
\` by definition
\end{tabbing}

\item[Case:] $\tau = \lambda y :: K.\tau'$
  \begin{tabbing}
$\lb\lambda y :: K . \tau'\{M/x\}\rb = \lambda y :: L\{M/x\}
. \tau'\{M/x\}\rb = \lambda y :: \lb K \{M/x\}\rb . \lb
\tau'\{M/x\}\rb$ \` by definition\\
$\lb\lambda y :: K . \tau'\rb\{\{M\}_c/x\} = \lambda y :: \lb K \rb\{\{M\}_c/x\}
. \lb \tau'\rb\{\{M\}_c/x\}$ \` by definition\\
$= \lambda y :: \lb K \{M/x\}\rb . \lb
\tau'\{M/x\}\rb$ \` by i.h.
\end{tabbing}

\item[Case:] $\tau = \tau'\,\sigma$

Straightforward by i.h. as above.

\noindent{\bf (5)}

\item[Case:] $A = \one$

Trivial.

\item[Case:] $A = A_1 \lolli A_2$

  \begin{tabbing}

$\lb A_1 \lolli A_2 \{M/x\}\rb = \lb A_1 \{M/x\}\rb \lolli \lb A_2
\{M/x\}\rb$ \` by definition\\
$\lb A_1 \rb\{\monad{\lb M \rb_c}/x\} \lolli \lb A_2 \rb\{\monad{\lb M
  \rb_c}/x\}$ \` by i.h.\\
$\lb A_1 \lolli A_2\rb\{\monad{\lb M \rb_c}/x\} =
\lb A_1 \rb\{\monad{\lb M \rb_c}/x\} \lolli \lb A_2 \rb\{\monad{\lb M
  \rb_c}/x\}$ \` by definition
  \end{tabbing}

\item[Case:] $A = A_1 \tensor A_2$

Identical to $\lolli$ case.
  
\item[Case:] $A = \with\{\ov{l_i {:} A_i}\}$

See above.
  
\item[Case:] $A = \oplus\{\ov{l_i {:} A_i}\}$

See above.
  
\item[Case:] $A = \forall x{:}\tau.A_0$ 

  \begin{tabbing}
$\lb\forall x{:}\tau . A_0 \{M/x\}\rb = \forall x{:}\monad{\lb \tau
  \{M/x\}\rb} . \lb A_0 \{M/x\}\rb$ \` by definition\\
$\forall x{:}\monad{\lb\tau \rb\{\monad{\lb M\rb_c}/x\}} . \lb A_0 \rb
\{\monad{\lb M\rb_c}/x\}$ \` by i.h.\\
$\lb \forall x{:}\tau.A_0\rb\{\monad{\lb M\rb_c}/x\} = \forall
x{:}\monad{\lb \tau\rb} . \lb A_0\rb \{\monad{\lb M\rb_c}/x\} =$\\ 
$\forall x{:}\monad{\lb\tau \rb\{\monad{\lb M\rb_c}/x\}} . \lb A_0 \rb
\{\monad{\lb M\rb_c}/x\}$ \` by definition
  \end{tabbing}
  
\item[Case:] $A = \exists x{:}\tau.A_0$ 

As above.
  
\item[Case:] $A = \lambda x{:}\tau.A_0$

As above.

\item[Case:] $A = A_0\,M'$

  \begin{tabbing}
$\lb A_0\,M'\{M/x\}\rb = \lb A_0\{M/x\}\rb \, \monad{\lb
  M'\{M/x\}\rb_c}$ \` by definition\\
$(\lb A_0\rb \{\monad{\lb M\rb_c}/x\})\,\monad{\lb M'\rb_d\{\monad{\lb
    M\rb_c}/x\}}$ \` by i.h.\\
$\lb A_0\,M'\rb\{\monad{\lb M\rb_c}/x\} =
( \lb A_0 \rb \, \monad{\lb M'\rb_d})\{\monad{\lb M\rb_c}/x\}$ \\
$= (\lb A_0\rb \{\monad{\lb M\rb_c}/x\})\,\monad{\lb M'\rb_d\{\monad{\lb
    M\rb_c}/x\}}$ \` by definition
  \end{tabbing}

\item[Case:] $A = A_0\,A_1$

Straightforward by i.h.

\item[Case:] $A = \lambda t :: K . A_0$

Straightforward by i.h.

\noindent{\bf (6)} $\Psi ; \Ga ; \D  \vdash \lb M \{N/x\}\rb_z =  \lb M\rb_z\{\monad{\lb N \rb_y}/x\} :: z{:}\lb
A\{N/x\}\rb$
  
\item[Case:] $M = \lambda y{:}\tau.M_0$

  \begin{tabbing}
    $\lb \lambda y{:}\tau.M_0 \{N/x\}\rb_z :: z{:}\lb \Pi y{:}\tau.\sigma\{N/x\}\rb$\\
    $\lb \lambda y{:}\tau.M_0 \{N/x\}\rb_z =
    \lb \lambda y{:}\tau\{N/x\}.M_0\{N/x\}\rb_z =
    z(y).\lb M_0\{N/x\}\rb_z :: z{:}\forall y{:}\monad{\lb \tau\{N/x\}\rb}.\lb\sigma\{N/x\}\rb$\\
    \` by definition\\
    $z(y).\lb M_0\{N/x\}\rb_z = z(y).\lb M_0\rb_z\{\monad{\lb N \rb_c}/x\} :: z{:}\forall
    y{:}\monad{\lb\tau\rb\{\monad{\lb N\rb_c}/x\}}.\lb \sigma\rb\{\monad{\lb N\rb_c}/x\}$ \` by i.h.\\
    $\lb \lambda y{:}\tau.M_0 \rb_z\{\monad{\lb N \rb_c}/x\} =
    z(y).\lb M_0\rb_z\{\monad{\lb N \rb_c}/x\}$ \` by definition
  \end{tabbing}

\item[Case:] $M = M_1\,M_2$

  \begin{tabbing}
    $\lb M_1\,M_2\{N/x\}\rb_z :: z{:} \lb
    (\sigma\{M_2/y\})\{N/x\}\rb$\\
    $= (\nub y)(\lb M_1\{N/x\}\rb_y \mid
    y\langle \monad{ \lb M_2\{N/x\} \rb_y}\rangle.
    [y\leftrightarrow z])$ \` by definition\\
    $= (\nub y)(\lb M_1\rb_y \{\monad{\lb N\rb_c}/x\} \mid
    y\langle \monad{ \lb M_2\rb_y \{\monad{\lb N\rb_c}/x\}}\rangle.
    [y\leftrightarrow z])$ \` by i.h.\\
    $\lb M_1 \, M_2\rb_z \{\monad{\lb N\rb_c}/x\} =
    (\nub y)(\lb M_1\rb_y \mid y\langle \monad{ \lb
      M_2\rb_y}\rangle.[y\leftrightarrow z] )
    \{\monad{\lb N\rb_c}/x\}$\\
    $=(\nub y)(\lb M_1\rb_y \{\monad{\lb N\rb_c}/x\} \mid
    y\langle \monad{ \lb M_2\rb_y \{\monad{\lb N\rb_c}/x\}}\rangle.
    [y\leftrightarrow z])$    \` by definition\\
  \end{tabbing}
 
\item[Case:] $M = \monad{z \leftarrow P \leftarrow \ov{u_j};\ov{d_i}}$

\begin{tabbing}
$\lb M\{N/x\}\rb_z = 
z(u_0).\dots.z(u_j).z(d_0). \dots . z(d_n).\lb P \{N/x\}\rb$ \` by
definition\\
$z(u_0).\dots.z(u_j).z(d_0). \dots . z(d_n).\lb P\rb\{\monad{\lb
  N\rb_c}/x\}$ \` by i.h.\\
$\lb M \rb_z\{\monad{\lb N\rb_c}/x\} = 
z(u_0).\dots.z(u_j).z(d_0). \dots . z(d_n).\lb P \rb\{\monad{\lb
  N\rb_c}/x\}$
\` by definition
\end{tabbing}

\item[Case:] $M = y$ with $y \neq x$

\begin{tabbing}
$\lb y \{N/x\}\rb_z = w \leftarrow y ; [w \leftrightarrow z]$ \` by
definition\\
$\lb y \rb_z\{\monad{\lb N \rb_c}/x\} = w \leftarrow y ; [w
\leftrightarrow z]$ \` by definition
\end{tabbing}

\item[Case:] $M = y$ with $y = x$

\begin{tabbing}
$\lb x \{N/x\}\rb_z = \lb N \rb_z$ \` by
definition\\
$\lb x \rb_z\{\monad{\lb N \rb_c}/x\} = w \leftarrow \monad{\lb N \rb_c} ; [w
\leftrightarrow z]$ \` by definition\\
$w \leftarrow \monad{\lb N \rb_c} ; [w
\leftrightarrow z] \tra{}^+ \lb N \rb_z$ \` by reduction semantics\\
$w \leftarrow \monad{\lb N \rb_c} ; [w
\leftrightarrow z] = \lb N \rb_z$ \` by $\m{PEqRed}$
\end{tabbing}

{\bf (7)} 
$\Psi ; \Ga ; \D \vdash \lb P\{M/x\}\rb :: z{:}\lb A
    \{M/x\}\rb$
     iff $\Psi ; \Ga ; \D \vdash \lb P \rb\{\monad{\lb M \rb_c}/x\}::
     z{:}\lb A \rb\{\monad{\lb M \rb_c}/x\}$

\item[Case:] $P = (\nub y)(P_1 \mid P_2)$

\begin{tabbing}
$\lb (\nub y)(P_1 \mid P_2) \{M/x\}\rb =
 (\nub y)(\lb P_1\{M/x\}\rb \mid \lb P_2\{M/x\}\rb)$ \` by
 definition\\
$(\nub y)(\lb P_1\rb\{\monad{\lb M \rb_c}/x\} \mid \lb
P_2\rb\{\monad{\lb M \rb_c}/x\})$ \` by i.h.\\
$\lb (\nub y)(P_1 \mid P_2) \rb\{\monad{\lb M \rb_c}/x\} =
  (\nub y)(\lb P_1\rb\{\monad{\lb M \rb_c}/x\} \mid \lb
P_2\rb\{\monad{\lb M \rb_c}/x\})$ \` by definition
\end{tabbing}

\item[Case:] $P = z\langle M_0 \rangle.P_0$ by $\rgt\exists$

\begin{tabbing}
$\lb z\langle M_0\rangle.P_0\{M/x\}\rb =
  z\langle \monad{\lb M_0\{M/x\}\rb_d}\rangle.\lb P_0\{M/x\}\rb$ \`
 by definition\\
$z\langle \monad{\lb M_0\rb_d\{\monad{\lb M \rb_c}/x\}}\rangle.\lb
P_0\rb\{\monad{\lb M \rb_c}/x\}$ \` by i.h.\\
$\lb z\langle M_0\rangle.P_0\rb\{\monad{\lb M \rb_c}/x\} =
z\langle \monad{\lb M_0\rb_d\{\monad{\lb M \rb_c}/x\}}\rangle.\lb
P_0\rb\{\monad{\lb M \rb_c}/x\}$ \` by definition\\
\end{tabbing}

\item[Case:] $P = x \leftarrow M_0 \leftarrow \ov{u_j};\ov{y_i};P_0$

\begin{tabbing}
$\lb P\{M/x\}\rb = (\nub c)(\lb M_0\{M/x\} \rb_c \mid \ov{c}\langle v_1\rangle.(\ov{u_1}\langle a_1\rangle.[a_1\leftrightarrow v_1] \mid
\dots \mid$\\
$            \ov{c}\langle d_1 \rangle.([y_1\leftrightarrow
      d_1] \mid \dots \mid \ov{c}\langle d_n\rangle.([y_n
                       \leftrightarrow d_n] \mid \lb P_0\{M/x\}\rb)
                       \dots )$ \` by definition\\
$(\nub c)(\lb M_0\rb_c\{\monad{\lb M \rb_c}/x\} \mid \ov{c}\langle v_1\rangle.(\ov{u_1}\langle a_1\rangle.[a_1\leftrightarrow v_1] \mid
\dots \mid$\\
$            \ov{c}\langle d_1 \rangle.([y_1\leftrightarrow
      d_1] \mid \dots \mid \ov{c}\langle d_n\rangle.([y_n
                       \leftrightarrow d_n] \mid \lb P_0\rb\{\monad{\lb M \rb_c}/x\})
                       \dots )$ \` by i.h.\\
$= \lb P\rb\{\monad{\lb M \rb_c}/x\}$ \` by definition                       
\end{tabbing}
Remaining process cases are straightforward by i.h.
\end{description}
\end{proof}

\begin{lemma}[Compositionality -- Reflection in Equality]
~
\begin{enumerate}
\item $\Psi \vdash \lb K\{M/x\}\rb = \lb K \rb
  \{\monad{\lb M \rb_c}/x\}$
\item $\Psi \vdash \lb K_1 \{\tau / t \}\rb = \lb K_1\rb
  \{\lb \tau \rb / t\}$
\item $\Psi \vdash \lb K_1\{A/x\}\rb = \lb K_1\rb\{\lb
  A \rb /x\}$
\item $\Psi \vdash \lb \tau\{M/x\}\rb  = \lb \tau\rb\{\monad{\lb M \rb_c}/x\} :: \lb K
\{M/x\}\rb$ 
\item $\Psi \vdash \lb A\{M/x\}\rb  = \lb A\rb\{\monad{\lb M \rb_c}/x\} :: \lb K
\{M/x\}\rb$ 
\item $\Psi ; \Ga ; \D  \vdash \lb M \{N/x\}\rb_z = \lb M\rb_z\{\monad{\lb N \rb_y}/x\} :: z{:}\lb
    A\{N/x\}\rb$ 
  \item $\Psi ; \Ga ; \D \vdash \lb P\{M/x\}\rb =
      \lb P \rb\{\monad{\lb M \rb_c}/x\}::
      z{:}\lb A \{M/x\}
     \rb$ 
\end{enumerate}
\end{lemma}

\begin{proof}
{\bf (1-3)} is identical to corresponding statements in Lemma~\ref{lem:comp}.

{\bf (4)} $\Psi \vdash \lb \tau\{M/x\}\rb  = \lb \tau\rb\{\monad{\lb M
  \rb_c}/x\} :: \lb K \{M/x\}\rb$ 
\begin{description}

\item[Case:]  $\tau = \Pi y{:}\tau'.\sigma$
\begin{tabbing}
$\lb \Pi y{:}\tau'.\sigma\{M/x\}\rb = 
 \lb \Pi y{:}\tau'\{M/x\}.\sigma\{M/x\}\rb =
  \forall y{:}\{\lb\tau'\{M/x\}\rb\}.
  \lb\sigma\{M/x\}\rb  $ \` by definition\\
$\lb\Pi y{:}\tau'.\sigma\rb\{\{\lb M \rb_c\}/x\} =
(\forall y{:}\{\lb\tau'\rb\}.
\lb\sigma\rb)\{\{\lb M \rb_c\}/x\} =
\forall y{:}\{\lb\tau'\rb\{\{\lb M \rb_c\}/x\}\}.
\lb\sigma\rb\{\{\lb M \rb_c\}/x\}$\\ \` by definition\\
$\Psi \vdash \{\lb\tau'\{M/x\}\rb\} = \{\lb\tau'\rb\{\{\lb M
\rb_c\}/x\}\} :: \type$ \` by i.h. and $\m{TEq}\{\}$\\
$\Psi , y{:}\{\lb\tau'\{M/x\}\rb\} \vdash \lb\sigma\{M/x\}\rb =
\lb\sigma\rb\{\{\lb M \rb_c\}/x\} :: \stype$ \` by i.h.\\
$\Psi \vdash \forall y{:}\{\lb\tau'\{M/x\}\rb\}.
  \lb\sigma\{M/x\}\rb = \forall y{:}\{\lb\tau'\rb\{\{\lb M \rb_c\}/x\}\}.
\lb\sigma\rb\{\{\lb M \rb_c\}/x\} :: \stype$ \` by $\m{STEq}\forall$
\end{tabbing}

\item[Case:] $\tau = \lambda y{:}\tau'.\sigma$

  \begin{tabbing}
    $\lb \lambda y{:}\tau'.\sigma\{M/x\}\rb =
    \lambda y{:}\monad{\lb \tau'\{M/x\}\rb}.\lb\sigma\{M/x\}\rb$ \` by
    definition\\
    $\lb \lambda y{:}\tau'.\sigma\rb\{\monad{\lb M \rb_c}/x\} =
    \lambda y{:}\monad{\lb \tau'\rb\}\{\monad{\lb M
        \rb_c}/x\}}.\lb\sigma\rb\{\monad{\lb M \rb_c}/x\}$ \` by
    definition\\
$\Psi \vdash \monad{\lb
      \tau'\{M/x\}\rb} = \monad{\lb \tau'\rb\}\{\monad{\lb M
        \rb_c}/x\}} :: \type$ \` by i.h. and $\m{TEq}\{\}$\\
$\Psi ,y{:}\monad{\lb
      \tau'\{M/x\}\rb} \vdash \lb\sigma\{M/x\}\rb =
      \lb\sigma\rb\{\monad{\lb M \rb_c}/x\} :: \lb K\{M/x\}\rb$ \` by i.h.\\

    $\Psi \vdash \lambda y{:}\monad{\lb
      \tau'\{M/x\}\rb}.\lb\sigma\{M/x\}\rb = 
  \lambda y{:}\monad{\lb \tau'\rb\}\{\monad{\lb M
        \rb_c}/x\}}.\lb\sigma\rb\{\monad{\lb M \rb_c}/x\} ::  \Pi
    x{:}\monad{\lb\tau'\{M/x\}\rb}.\lb K\{M/x\}\rb$
\\ \` by $\m{STEq}\lambda$
  \end{tabbing}

\item[Case:] $\tau = \tau'\, M'$

  \begin{tabbing}
$\lb \tau'\, M' \{M/x\} \rb = \lb \tau'\{M/x\}\rb \, \monad{\lb
  M'\{M/x\}\rb_d} $ \` by definition\\
$\lb \tau'\, M'\rb\{\monad{\lb M \rb_c}/x\} =
(\lb \tau'\rb\{\monad{\lb M \rb_c}/x\})\,\monad{\lb
  M'\rb_d\{\monad{\lb M \rb_c}/x\}}$ \` by definition\\
$\Psi \vdash \lb \tau'\{M/x\}\rb = \lb \tau'\rb\{\monad{\lb M
  \rb_c}/x\} :: \Pi y{:}\monad{\lb \tau''\{M/x\}\rb}.\lb K \{M/x\}\rb$
\` by i.h\\
$\Psi \vdash \monad{\lb M'\{M/x\}\rb_d} = \monad{\lb
  M'\rb_d\{\monad{\lb M \rb_c}/x\}} : \monad{\lb \tau''\{M/x\}\rb}$ \` by
i.h. and $\m{TEq}\{\}$\\
$\Psi \vdash \lb \tau'\{M/x\}\rb \, \monad{\lb
  M'\{M/x\}\rb_d} =$\\
$\qquad (\lb \tau'\rb\{\monad{\lb M \rb_c}/x\})\,\monad{\lb
  M'\rb_d\{\monad{\lb M \rb_c}/x\}} :: \lb K \{M/x\}\rb\{\monad{\lb
  M'\{M/x\}\rb_d}/y\}$ \` by $\m{STEqApp}$
  \end{tabbing}

\item[Case:] $\tau = \lambda t :: K'.\tau'$

  \begin{tabbing}
$\lb \lambda t :: K'.\tau'\{M/x\}\rb = \lambda t :: \lb K'\{M/x\}\rb
. \lb \tau'\{M/x\}\rb$ \` by definition\\
$\lb \lambda t :: K' . \tau'\rb\{\monad{\lb M \rb_c}/x\} =
 \lambda t :: \lb K'\rb\{\monad{\lb M \rb_c}/x\} . \lb
 \tau'\rb\{\monad{\lb M \rb_c}/x\}$ \` by definition\\
$\Psi \vdash \lb K'\{M/x\}\rb = \lb K'\rb\{\monad{\lb M \rb_c}/x\}$ \`
by i.h.\\
$\Psi , t :: \lb K'\{M/x\}\rb \vdash  \lb \tau'\{M/x\}\rb = \lb
 \tau'\rb\{\monad{\lb M \rb_c}/x\} :: \lb K''\{M/x\}\rb$ \` by i.h.\\
$\Psi \vdash \lambda t :: \lb K'\{M/x\}\rb
. \lb \tau'\{M/x\}\rb =$\\
$\qquad\lambda t :: \lb K'\rb\{\monad{\lb M \rb_c}/x\} . \lb
 \tau'\rb\{\monad{\lb M \rb_c}/x\} ::  \Pi t :: \lb K'\{M/x\}\rb.\lb
 K''\{M/x\}\rb$
 \` by $\m{STEqT}\lambda$
  \end{tabbing}

\item[Case:] $\tau = \tau'\,\sigma$

  \begin{tabbing}
$\lb \tau'\,\sigma\{M/x\}\rb =
\lb\tau'\{M/x\}\rb\,\lb\sigma\{M/x\}\rb$ \` by definition\\
$\lb \tau' \, \sigma\rb\{\monad{\lb M \rb_c}/x\} =
\lb \tau'\rb\{\monad{\lb M \rb_c}/x\} \,\lb \sigma\rb \{\monad{\lb M
  \rb_c}/x\}$ \` by definition\\
$\Psi \vdash \lb\tau'\{M/x\}\rb = \lb \tau'\rb\{\monad{\lb M
  \rb_c}/x\} ::  \Pi t : \lb K\{M/x\}\rb.\lb K'\{M/x\}\rb$ \` by i.h.\\
$\Psi \vdash \lb\sigma\{M/x\}\rb = \lb \sigma\rb \{\monad{\lb M
  \rb_c}/x\} ::  \lb K\{M/x\}\rb$ \` by i.h.
  \end{tabbing}

Remaining cases are identical.
\end{description}
\end{proof}

\begin{lemma}[Preservation of Equality]~
  
  \begin{enumerate}
\item If $\Psi \vdash K_1 = K_2$ then $\monad{\lb \Psi \rb} \vdash \lb
  K_1 \rb = \lb K_2\rb$
\item If $\Psi \vdash \tau_1 = \tau_2 :: K$ then $\monad{\lb \Psi \rb}
  \vdash \lb \tau_1\rb = \lb \tau_2 \rb :: \lb K \rb$
\item If $\Psi \vdash A = B :: K$ then $\monad{\lb \Psi \rb} \vdash
  \lb A \rb = \lb B \rb :: \lb K \rb$
\item If $\Psi \vdash M = N : \tau$ then $\monad{\lb \Psi \rb} ; \cdot
  ; \cdot \vdash \lb M \rb_z = \lb N \rb_z :: z{:}\lb \tau\rb$
\item If $\Psi ; \Ga ; \D \vdash P = Q :: z{:}A$ then $\monad{\lb \Psi
    \rb} ; \lb \Ga \rb ; \lb \D \rb\vdash \lb P \rb = \lb Q\rb ::
  z{:}\lb A \rb$
\end{enumerate}
\end{lemma}

\begin{proof}
  By induction on the given judgment.

  \begin{description}

  \item[Case:] $\m{KEqR}$, $\m{KEqS}$, $\m{KEqT}$ and $\m{KEq}\Pi_2$

  Immediate by i.h.

\item[Case:] $\m{KEq}\Pi_1$

  \begin{tabbing}
$\Psi \vdash \tau = \sigma :: \type$ and $\Psi , x{:}\tau \vdash K_1 =
K_2$ \` by inversion\\
$\monad{\lb\Psi\rb} \vdash \lb \tau \rb = \lb \sigma\rb :: \stype$ \`
by i.h.\\
$\monad{\lb\Psi\rb} \vdash \monad{\lb \tau\rb} = \monad{\lb \sigma\rb}
:: \type$ \` by $\m{TEq}\monad{}$
$\monad{\lb\Psi\rb} , x{:}\monad{\lb \tau\rb} \vdash \lb K_1 \rb = \lb
K_2 \rb$ \` by i.h.\\
$\monad{\lb\Psi\rb} \vdash \Pi x{:}\monad{\lb \tau\rb}.\lb K_1\rb =
 \Pi x{:}\monad{\lb \sigma\rb}.\lb K_2\rb$ \` by $\m{KEq}\Pi_1$
\end{tabbing}

\noindent {\bf (2)}

\item[Case:] $\m{TEqR}$, $\m{TEqT}$, $\m{TEqS}$

  Immediate by i.h.

\item[Case:] $\m{TEq}\Pi$

  \begin{tabbing}
$\Psi \vdash \tau = \tau' :: \type$ and $\Psi , x{:}\tau \vdash \sigma
= \sigma' :: \type$ \` by inversion\\
$\monad{\lb \Psi \rb} \vdash \lb \tau\rb = \lb \tau'\rb :: \stype$ \`
by i.h.
$\monad{\lb \Psi \rb} \vdash \monad{\lb \tau\rb} = \monad{\lb
  \tau'\rb} :: \type$ \` by  $\m{TEq}\monad{}$\\
$\monad{\lb\Psi\rb} , x{:}\monad{\lb \tau\rb} \vdash \lb \sigma \rb =
\lb \sigma' \rb :: \stype$ \` by i.h.\\
$\monad{\lb\Psi\rb} \vdash \forall x{:}\monad{\lb \tau\rb}.\lb \sigma
\rb = \forall x{:}\monad{\lb \tau'\rb}.\lb \sigma'
\rb :: \stype$ \` by $\m{STEq}\forall$
\end{tabbing}

\item[Case:] $\m{TEq}\lambda$

  \begin{tabbing}
$\Psi \vdash \tau = \tau' :: \type$ and $\Psi , x{:}\tau \vdash \sigma
= \sigma' :: K$ \` by inversion\\
$\monad{\lb\Psi\rb} \vdash \lb \tau \rb = \lb \tau'\rb :: \stype$ \`
by i.h.\\
$\monad{\lb \Psi \rb} \vdash \monad{\lb \tau\rb} = \monad{\lb
  \tau'\rb} :: \type$ \` by  $\m{TEq}\monad{}$\\
$\monad{\lb\Psi\rb}, x{:}\monad{\lb \tau\rb} \vdash \lb \sigma \rb =
\lb \sigma' \rb :: \lb K \rb$ \` by i.h.\\
$\monad{\lb\Psi\rb} \vdash \lambda x{:}\monad{\lb \tau\rb} . \lb
\sigma \rb = \lambda x{:}\monad{\lb \tau'\rb}.\lb \sigma'\rb :: \Pi
x{:}\monad{\lb \tau\rb}.\lb K\rb$ \` by $\m{STEq}\lambda$
  \end{tabbing}

\item[Case:] $\m{TEq}T\lambda$

  \begin{tabbing}
$\Psi \vdash K = K'$ and $\Psi , t :: K \vdash \tau = \sigma :: K''$
\` by inversion\\
$\monad{\lb \Psi \rb} \vdash \lb K \rb = \lb K'\rb$ \` by i.h.\\
$\monad{\lb \Psi \rb} , t :: \lb K \rb \vdash \lb \tau \rb = \lb
\sigma \rb :: \lb K'' \rb$ \` by i.h.\\
$\monad{\lb \Psi \rb} \vdash \lambda t :: \lb K \rb . \lb \tau\rb =
\lambda t :: \lb K' \rb .\lb \sigma\rb :: \Pi t :: \lb K\rb . \lb
K''\rb$ \` by $\m{STEq}T\lambda$
  \end{tabbing}

\item[Case:] $\m{TEqApp}$

  \begin{tabbing}
$\Psi \vdash \tau = \sigma :: \Pi x{:} \tau'.K$ and $\Psi \vdash M = N
: \tau'$ \` by inversion\\
$\monad{\lb \Psi \rb} \vdash \lb \tau \rb = \lb \sigma \rb :: \Pi x{:}
\monad{\lb \tau'\rb}.\lb K\rb$ \` by i.h.\\
$\monad{\lb \Psi \rb} ; \cdot ; \cdot \vdash \lb M \rb_z = \lb N \rb_z :: z{:}\lb
\tau'\rb$ \` by i.h.\\
$\monad{\lb \Psi \rb} \vdash \monad{\lb M \rb_z} = \monad{\lb N \rb_z}
: \monad{\lb \tau'\rb}$ \` by $\m{TEq}\monad{}$\\
  $\monad{\lb \Psi \rb} \vdash \lb \tau \rb \, \monad{\lb M \rb_z} =
    \lb \sigma \rb\, \monad{\lb N \rb_z} :: \lb K\rb\{\monad{\lb M
      \rb_z}/x\}$ \` by $\m{STEqApp}$\\
  $\monad{\lb \Psi \rb} \vdash \lb \tau \rb \, \monad{\lb M \rb_z} =
    \lb \sigma \rb\, \monad{\lb N \rb_z} :: \lb K \{M/x\}\rb$ \` by
    compositionality and conversion
  \end{tabbing}

\item[Case:] $\m{TEqTApp}$

  \begin{tabbing}
$\Psi \vdash \tau = \tau' :: \Pi t :: K_1 . K_2$ and $\Psi \vdash
\sigma = \sigma ' :: K_1$ \` by inversion\\
$\monad{\lb \Psi \rb} \vdash \lb \tau\rb = \lb \tau' \rb :: \Pi t ::
\lb K_1 \rb . \lb K_2\rb$ \` by i.h.\\
$\monad{\lb \Psi \rb} \vdash \lb \sigma \rb = \lb \sigma' \rb :: \lb
K_1\rb$ \` by i.h.\\
$\monad{\lb \Psi \rb} \vdash \lb \tau \rb \, \lb \sigma \rb = \lb
\tau' \rb \, \lb \sigma'\rb :: \lb K_2\rb\{\lb \sigma\rb / t\}$ \` by
$\m{STEqTApp}$\\
$\monad{\lb \Psi \rb} \vdash \lb \tau \rb \, \lb \sigma \rb = \lb
\tau' \rb \, \lb \sigma'\rb :: \lb K_2\{ \sigma / t\}\rb$ \` by
compositionality and conversion 
  \end{tabbing}

\item[Case:] $\m{TEq}\beta$

  \begin{tabbing}
$\Psi , x{:}\tau \vdash \sigma :: K$ and $\Psi \vdash M : \tau$ \` by
inversion\\
$\monad{\lb \Psi \rb} , x{:} \monad{\lb \tau\rb} \vdash \lb \sigma\rb
:: \lb K \rb$ \` by type preservation of the encoding\\
$\monad{\lb \Psi \rb} ; \cdot ; \cdot \vdash \lb M \rb_c :: c{:}\lb
\tau\rb$ \` by type preservation of the encoding\\
$\monad{\lb \Psi \rb} \vdash \monad{\lb M \rb_c} : \monad{\lb\tau\rb}$
\` by $\{\}I$\\
$\monad{\lb \Psi \rb} \vdash (\lambda x{:}\monad{\lb\tau\rb}.\lb
\sigma\rb)\,\monad{\lb M \rb_c} = \lb \sigma\rb\{ \monad{\lb M \rb_c} /x\} :: \lb K \rb\{
\monad{\lb M \rb_c}/x\}$ \` by $\m{STEq}\beta$\\
$\monad{\lb \Psi \rb} \vdash (\lambda x{:}\monad{\lb\tau\rb}.\lb
\sigma\rb)\,\monad{\lb M \rb_c} = \lb \sigma \{M/x\}\rb
 :: \lb K \{M/x\} \rb$ \\\` by
compositionality and conversion 
\end{tabbing}

\item[Case:] $\m{TEq}T\beta$

  \begin{tabbing}
$\Psi \vdash \sigma :: K$ and $\Psi , t :: K \vdash \tau :: K'$ \` by
inversion\\
$\monad{\lb \Psi \rb} \vdash \lb \sigma \rb :: \lb K \rb$ \` by type
preservation of the encoding\\
$\monad{\lb \Psi \rb} , t :: \lb K \rb \vdash \lb \tau \rb :: \lb
K'\rb$ \` by type preservation of the encoding\\
$\monad{\lb \Psi \rb} \vdash (\lambda t :: \lb K \rb.\lb
\tau\rb)\,\lb\sigma\rb = \lb \tau \rb \{ \lb \sigma \rb / t\} :: \lb
K'\rb\{\lb \sigma \rb / t\}$ \` by $\m{STEq}T\beta$\\
$\monad{\lb \Psi \rb} \vdash (\lambda t :: \lb K \rb.\lb
\tau\rb)\,\lb\sigma\rb = \lb \tau  \{ \sigma  / t\}\rb :: \lb
K'\{ \sigma  / t\}\rb$ \`  by
compositionality and conversion
  \end{tabbing}
  
\item[Case:] $\m{TEq}\eta$

\begin{tabbing}
$\Psi \vdash \sigma :: \Pi x{:}\tau.K$ and $x\not\in fv(\sigma)$ \` by
inversion\\
$\monad{\lb\Psi\rb} \vdash \lb \sigma \rb :: \Pi x{:}\monad{\lb
    \tau\rb}.\lb K \rb$ \` by type preservation of the encoding\\
$\monad{\lb\Psi\rb} \vdash \lambda x{:}\monad{\lb\tau\rb}.\lb
\sigma\rb\,x = \lb \sigma\rb :: \Pi x{:}\monad{\lb\tau\rb}.\lb K \rb$
\` by $\m{STEq}\eta$\\
$\monad{\lb \Psi \rb}\vdash \lambda x{:}\monad{\lb \tau\rb}.\lb
\sigma\rb\,\monad{c\leftarrow (y \leftarrow x ; [y\leftrightarrow c])} = \lb \sigma\rb :: \Pi
x{:}\monad{\lb\tau\rb}.\lb K \rb$ \\\` by $\m{STEqT}$, $\m{STEqApp}$ and $\m{TMEq}\{\}\eta$
\end{tabbing}

\item[Case:] $\m{TEqT}\eta$

  \begin{tabbing}
$\Psi \vdash \tau :: \Pi t :: K_1.K_2$ and $t \not\in fv(\tau)$ \` by
inversion\\
$\monad{\lb \Psi \rb} \vdash \lb \tau \rb :: \Pi t :: \lb K_1 \rb
. \lb K_2 \rb$ \` by type preservation of the encoding\\
$\monad{\lb \Psi \rb} \vdash \lambda t :: \lb K \rb .\lb\tau\rb\,t =
\lb \tau\rb :: \Pi t :: \lb K\rb.\lb K'\rb$ \` by $\m{STEqT}\eta$

\end{tabbing}

\item[Case:] $\m{TEq}\{\}$

  \begin{tabbing}
    $\forall i , j . \Psi \vdash A_i = B_i :: \stype$,
    $\Psi \vdash C_j = D_j :: \stype$ and $\Psi \vdash A = B ::
    \stype$ \` by inversion\\
    $\monad{\lb \Psi \rb} \vdash \lb C_j \rb = \lb D_j \rb :: \stype$
    \` by i.h.\\
    $\monad{\lb \Psi \rb} \vdash \lb A_i\rb = \lb B_i\rb :: \stype$ \`
    by i.h.\\
    $\monad{\lb \Psi \rb}\vdash \lb A \rb = \lb B \rb :: \stype$ \` by
    i.h.\\
    $\monad{\lb \Psi \rb} \vdash \ov{\bang \lb C_j\rb} \lolli \ov{\lb A_i\rb} \lolli \lb A \rb =
     \ov{\bang \lb D_j\rb} \lolli \ov{\lb B_i\rb} \lolli \lb B \rb ::
     \stype$ \` by $\m{STEq}\lolli$ and
     $\m{STEq}\bang$
   \end{tabbing}

   \noindent {\bf (3)}

   All cases are identical to those of {\bf (2)}.

   \noindent {\bf (4)}

 \item[Case:] $\m{TMEqR}$

   \begin{tabbing}
     $\Psi \vdash M : \tau$ \` by inversion\\
     $\monad{\lb \Psi \rb} ; \cdot ;\cdot \vdash \lb M \rb_z ::
     z{:}\lb \tau\rb$ \` by type preservation of the encoding\\
     $\monad{\lb \Psi \rb} ; \cdot ; \cdot \vdash \lb M \rb_z = \lb M
     \rb_z :: z{:}\lb\tau\rb$ \` by $\m{PEqR}$
   \end{tabbing}

 \item[Case:] $\m{TMEqS}$ $\m{TMEqT}$

   Immediate by i.h. and the corresponding definitional equality rules
   for processes.

 \item[Case:] $\m{TMEq}\lambda$

   \begin{tabbing}
     $\Psi \vdash \lambda x{:}\tau.M : \Pi x{:}\tau.\sigma$,
     $\Psi \vdash \lambda x{:}\tau'.N : \Pi x{:}\tau'.\sigma'$,
     $\Psi \vdash \Pi x{:}\tau.\sigma = \Pi x{:}\tau'.\sigma' ::
     \type$\\
     and $\Psi , x{:}\tau \vdash M = N : \sigma$ \` by inversion\\
     $\monad{\lb \Psi \rb} , x{:}\monad{\lb \tau \rb} ; \cdot ; \cdot \vdash
     \lb M \rb_z = \lb N \rb_z :: z{:}\lb \sigma\rb$ \` by i.h.\\
     $\monad{\lb \Psi \rb} \vdash \forall x{:}\monad{\lb \tau\rb}.\lb
     \sigma\rb = \forall x{:}\monad{\lb \tau'\rb}.\lb \sigma'\rb ::
     \stype$ \` by i.h.\\
     $\monad{\lb \Psi \rb} \vdash z(x).\lb M\rb_z = z(x').\lb N \rb_z
     :: z{:} \forall x{:}\monad{\lb \tau\rb}.\lb \sigma\rb$ \` by $\m{PEq}\rgt\forall$

   \end{tabbing}

\item[Case:] $\m{TMEqApp}$

\begin{tabbing}
$\Psi \vdash M = M' : \Pi x{:}\tau.\sigma$ and $\Psi \vdash N = N' :
\tau$ \` by inversion\\
$\monad{\lb\Psi\rb} ; \cdot ; \cdot \vdash \lb M \rb_x = \lb M'\rb_x
:: x{:}\forall x{:}\monad{\lb \tau\rb}.\lb \sigma\rb$ \` by i.h.\\
$\monad{\lb \Psi \rb} ; \cdot ; \cdot \vdash \lb N \rb_y = \lb N' \rb_y :: y{:}\lb
\tau\rb$ \` by i.h.\\
$\monad{\lb \Psi \rb} \vdash \monad{\lb N \rb_y} = \monad{\lb N'\rb_y}
: \monad{\lb \tau\rb}$ \` by $\m{TMEq}\{\}$\\
$\monad{\lb \Psi \rb} ; \cdot ; \cdot \vdash (\nub x)(\lb M \rb_x \mid
x\langle  \monad{\lb N \rb_y}\rangle.[x\leftrightarrow z]) =$\\
$\qquad (\nub x)(\lb M' \rb_x \mid
x\langle  \monad{\lb N' \rb_y}\rangle.[x\leftrightarrow z]) :: z{:}\lb
\sigma\rb\{\monad{\lb N \rb_y}/x\}$ \` by $\m{PEq}\cut$,
$\m{PEq}\lft\forall$, $\m{PEqID}$\\
$\monad{\lb \Psi \rb} ; \cdot ; \cdot \vdash (\nub x)(\lb M \rb_x \mid
x\langle  \monad{\lb N \rb_y}\rangle.[x\leftrightarrow z]) =$\\
$\qquad (\nub x)(\lb M' \rb_x \mid
x\langle  \monad{\lb N' \rb_y}\rangle.[x\leftrightarrow z]) :: z{:}\lb
\sigma\{N/x\}\rb$ \` by compositionality and conversion\\
\end{tabbing}

\item[Case:] $\m{TMEq}\beta$

\begin{tabbing}
$\Psi \vdash \lambda x{:}\tau.M : \Pi x{:}\tau.\sigma$ and $\Psi
\vdash N : \tau$ \` by inversion\\
$\monad{\lb \Psi \rb} ;\cdot;\cdot\vdash y(x).\lb M \rb_y :: y{:}\forall
x{:}\monad{\lb \tau\rb}.\lb \sigma\rb$ \` by type preservation of the
encoding\\
$\monad{\lb \Psi \rb};\cdot;\cdot \vdash \lb N \rb_w :: w{:}\lb \tau\rb$ \` by
type preservation of the encoding\\
$\monad{\lb \Psi \rb} \vdash \monad{\lb N \rb_w} : \monad{\lb\tau\rb}$
\` by $\{\}I$\\

To show: $\monad{\lb\Psi\rb} ; \cdot ; \cdot \vdash \lb (\lambda
x{:}\tau.M)\,N\rb_z = \lb M\{N/x\} \rb_z :: z{:} \lb \sigma\{N/x\}\rb$\\
S.T.S: $\monad{\lb\Psi\rb} ; \cdot ; \cdot \vdash (\nub y)(y(x).\lb M \rb_y
\mid y\langle \monad{\lb N \rb_w} \rangle.[y\leftrightarrow z]) = \lb
M\{N/x\} \rb_z :: z{:}\lb \sigma\{N/x\}\rb$\\\\

$\monad{\lb\Psi\rb} ; \cdot ; \cdot \vdash (\nub y)(y(x).\lb M \rb_y
\mid y\langle \monad{\lb N \rb_w} \rangle.[y\leftrightarrow z]) ::
z{:}\lb\sigma\rb\{\monad{\lb N \rb_w}/x\}$ \\\` by above,
$\m{id}$,$\lft\forall$ and $\m{cut}$\\
$\tra{}\tra{} \lb M\rb_z\{\monad{\lb N \rb_w}/x\}$ \` by operational semantics\\
$\monad{\lb\Psi\rb} ; \cdot ; \cdot \vdash \lb M\rb_z\{\monad{\lb N
  \rb_w}/x\} :: z{:}\lb\sigma\rb\{\monad{\lb N \rb_w}/x\}$ \` by type
preservation\\
$\monad{\lb\Psi\rb} ; \cdot ; \cdot \vdash (\nub y)(y(x).\lb M \rb_y
\mid y\langle \monad{\lb N \rb_w} \rangle.[y\leftrightarrow z]) = \lb
M\{N/x\} \rb_z :: z{:}\lb \sigma\{N/x\}\rb$ \\\` by above,
$\m{PEqRed}$, compositionality and conversion
\end{tabbing}

\item[Case:] $\m{TMEq}\eta$

\begin{tabbing}
$\Psi \vdash M : \Pi x{:}\tau.\sigma$ and $x \not\in {fv}(M)$ \` by
inversion\\
$\monad{\lb \Psi \rb} ; \cdot ; \cdot \vdash \lb M \rb_y ::
z{:}\forall x{:}\monad{\lb \tau\rb}.\lb \sigma\rb$ \` by type
preservation of the encoding\\
$\monad{\lb \Psi \rb} ; \cdot ; \cdot \vdash
z(x).(\nub y)(\lb M \rb_y \mid y\langle x\rangle.[y\leftrightarrow z]) =$\\
\qquad $z(x).(\nub y)(\lb M \rb_y \mid y\langle \monad{c\leftarrow
  (y\leftarrow x ; [y\leftrightarrow c]) 
\leftarrow \cdot} \rangle.[y\leftrightarrow z]) :: z{:}\forall x{:}\monad{\lb \tau\rb}.\lb \sigma\rb$\\
\` by $\m{PEqCut}$, $\m{PEq}\lft\forall$, $\m{PEqID}$,
$\m{TMEq}\{\}\eta$ and $\m{PEqR}$\\

To show: $\monad{\lb \Psi \rb} ; \cdot ; \cdot \vdash
z(x).(\nub y)(\lb M \rb_y \mid y\langle \monad{c\leftarrow
  (y\leftarrow x ; [y\leftrightarrow c]) 
\leftarrow \cdot} \rangle.[y\leftrightarrow z])
=$\\
$\qquad \lb M \rb_z :: z{:}\forall x{:}\monad{\lb \tau\rb}.\lb \sigma\rb$\\\\

$\monad{\lb \Psi \rb} ; \cdot ; \cdot \vdash
z(x).(\nub y)(\lb M \rb_y \mid y\langle x\rangle.[y\leftrightarrow z])
= (\nub y)(\lb M \rb_y \mid z(x).y\langle x\rangle.[y\leftrightarrow
z]) :: z{:}\forall x{:}\monad{\lb \tau\rb}.\lb \sigma\rb$
\\\` by $\m{PEqCC}\forall$\\
$\monad{\lb \Psi \rb} ;\cdot ;\cdot \vdash
(\nub y)(\lb M \rb_y \mid z(x).y\langle x\rangle.[y\leftrightarrow
z]) = (\nub y)(\lb M \rb_y \mid [y\leftrightarrow z]) :: z{:}\forall
x{:}\monad{\lb \tau\rb}.\lb \sigma\rb$ \\\` by $\m{PEq}\forall\eta$\\

$(\nub y)(\lb M \rb_y \mid [y\leftrightarrow z]) \tra{} \lb M \rb_z$
\` by the operational semantics\\
$ \monad{\lb \Psi \rb} ;\cdot ;\cdot \vdash \lb M \rb_z :: z{:}\forall
x{:}\monad{\lb \tau\rb}.\lb \sigma\rb$ \` by type preservation\\
$ \monad{\lb \Psi \rb} ;\cdot ;\cdot \vdash 
(\nub y)(\lb M \rb_y \mid [y\leftrightarrow z]) = \lb M \rb_z :: z{:}\forall
x{:}\monad{\lb \tau\rb}.\lb \sigma\rb$ \` by $\m{PEqRed}$ \\
$\monad{\lb \Psi \rb} ; \cdot ; \cdot \vdash
z(x).(\nub y)(\lb M \rb_y \mid y\langle \monad{c\leftarrow
  (y\leftarrow x ; [y\leftrightarrow c]) 
\leftarrow \cdot} \rangle.[y\leftrightarrow z])
= \lb M \rb_z :: z{:}\forall x{:}\monad{\lb \tau\rb}.\lb \sigma\rb$\\
\` by the above reasoning and $\m{PEqT}$
\end{tabbing}

\item[Case:] $\m{TMEq}\{\}$   

\begin{tabbing}
$\Psi ; \ov{u_j {:}B_j} ; \ov{d_i{:}A_i} \vdash P = Q :: c{:}A$ \` by
inversion\\
$\monad{\lb \Psi \rb} ;  \ov{u_j {:}\lb B_j\rb} ; \ov{d_i{:}\lb
  A_i\rb} \vdash \lb P \rb = \lb Q \rb :: c{:}\lb A \rb$ \` by i.h.\\
$\monad{\lb \Psi \rb} ; \cdot ; \cdot \vdash c(u_0).\dots
c(u_j).c(d_0).\dots c(d_n).\lb P \rb =$\\
$\qquad c(u_0).\dots
c(u_j).c(d_0).\dots c(d_n).\lb Q \rb :: c{:} \ov{\bang \lb B_j\rb}
\lolli
\ov{\lb A_i\rb} \lolli \lb A \rb$ \` by $\m{PEq}\rgt\lolli$,
$\m{PEq}\lft\bang$

\end{tabbing}

\item[Case:] $\m{TMEq}\{\}\eta$

\begin{tabbing}
$\Psi \vdash M : \monad{ \ov{u_j {:}B_j} ; \ov{d_i{:}A_i} \vdash
  c{:}A}$ \` by inversion\\
$\monad{\lb \Psi \rb} ; \cdot ; \cdot \vdash \lb M \rb_z :: \ov{\bang \lb B_j\rb}
\lolli
\ov{\lb A_i\rb} \lolli \lb A \rb$\\
To show: $\monad{\lb \Psi \rb} ; \cdot ; \cdot \vdash
c(u_0).\dots c(u_j).c(d_0).\dots c(d_n).\lb z \leftarrow M ;
\ov{u_j};\ov{d_i};[z\leftrightarrow c]\rb$\\
$\qquad\qquad = \lb M \rb_c :: c{:}  \ov{\bang \lb B_j\rb}
\lolli
\ov{\lb A_i\rb} \lolli \lb A \rb$\\
S.T.S: $\monad{\lb \Psi \rb} ; \cdot ; \cdot \vdash
c(u_0).\dots c(u_j).c(d_0).\dots c(d_n).  
(\nub z)(\lb M \rb_z \mid$\\
$\qquad\ov{z}\langle v_1\rangle.(\ov{u_1}\langle a_1\rangle.[a_1\leftrightarrow v_1] \mid
\dots \mid  \ov{z}\langle d_1 \rangle.([y_1\leftrightarrow
      d_1] \mid \dots \mid \ov{z}\langle d_n\rangle.([y_n
                       \leftrightarrow d_n] \mid [z\leftrightarrow c])
                       \dots )$\\
$\qquad = \lb M \rb_c :: c{:}  \ov{\bang \lb B_j\rb}
\lolli
\ov{\lb A_i\rb} \lolli \lb A \rb$\\
\` by $\m{PEqCC}\lolli$ and $\m{PEqCC}\bang$ and $\m{PEq}\lolli\eta$
and $\m{PEq}\bang\eta$, $\m{PEqR}$ and $\m{PEqRed}$
 
\end{tabbing}

\item[Case:] $\m{PEqR}$

\begin{tabbing}
$\Psi ; \Ga ; \D \vdash P :: z{:}A$ \` by inversion\\
$\monad{\lb \Psi\rb} ; \lb \Ga \rb ; \lb \D\rb \vdash \lb P \rb ::
z{:}\lb A \rb$ \` by type preservation of the encoding\\
$\monad{\lb \Psi\rb} ; \lb \Ga \rb ; \lb \D\rb \vdash \lb P \rb = \lb
P \rb :: z{:} \lb A \rb$ \` by $\m{PEqR}$
\end{tabbing}
  
\item[Case:] $\m{PEqS}$ and $\m{PEqT}$

  Straightforward by i.h.

\item[Case:] $\m{PEqRed}$

  \begin{tabbing}
$\Psi ; \Ga ; \D \vdash P :: z{:}A$, $P \tra{}^* Q$ and $\Psi ; \Ga ;
\D \vdash Q :: z{:}A$ \` by inversion\\
$\monad{\lb \Psi \rb} ; \lb \Ga \rb ; \lb \D \rb \vdash \lb P \rb ::
z{:} \lb A \rb$ \` by type preservation of the encoding\\
$\monad{\lb \Psi \rb} ; \lb \Ga \rb ; \lb \D \rb \vdash \lb Q \rb ::
z{:} \lb A \rb$ \` by type preservation of the encoding\\
$\lb P \rb \tra{}^* \lb Q \rb$ \` by operational correspondence\\
$\monad{\lb \Psi \rb} ; \lb \Ga \rb ; \lb \D \rb \vdash \lb P \rb =
\lb Q \rb :: z{:}\lb A \rb$ \` by $\m{PEqRed}$
\end{tabbing}

\item[Case:] $\m{PEq}\{\}E$

  \begin{tabbing}
    $\Psi \vdash M = N: \monad{\ov{u_j{:}B_j};\ov{d_i{:}A_i}\vdash
      c{:}A}$,
    $\Psi ; \Ga ; \D , c{:}A \vdash Q = Q' :: z{:}C$, $\ov{u_j{:}B_j}
    \subseteq \Ga$ and $\ov{d_i{:}A_i} = \D'$\\ \` by inversion\\
    $\monad{\lb\Psi\rb} ; \cdot ; \cdot \vdash \lb M \rb_y = \lb N
    \rb_ y :: y{:}\ov{\bang \lb B_j\rb} \lolli \ov{\lb A_i\rb} \lolli
    \lb A\rb$ \` by i.h.\\
    $\monad{\lb\Psi\rb} ; \lb \Ga \rb ; \lb \D \rb, c{:}\lb A \rb \vdash \lb Q \rb =
    \lb Q' \rb :: z{:}\lb C \rb$ \` by i.h.\\
    We conclude by $\m{PEqCut}$, (repeated) $\m{PEq}\lft\lolli$,
    $\m{PEq}\lft\bang$ and  $\m{PEqID}$.
  \end{tabbing}

All other process cases follow fairly straightforwardly by i.h.
  
\end{description}
\end{proof}

\begin{lemma}[Preservation of Typing]
 \label{lem:pwf}~
\begin{enumerate}
\item If $\Psi \vdash$ then $\lb \Psi\rb\vdash$ and $\{\lb \Psi \rb\} \vdash$.
\item If $\Psi \vdash K$ then $\{\lb \Psi \rb\} \vdash \lb K\rb$
\item If $\Psi \vdash \tau :: K$ then $\{\lb \Psi \rb\} \vdash
  \lb\tau\rb :: \lb K\rb$
\item If $\Psi \vdash A :: K$ then $\{\lb \Psi \rb\} \vdash \lb A \rb
  :: \lb K\rb$
\item If $\Psi \vdash M : \tau$ then $\{\lb \Psi \rb\} ; \cdot
  ; \cdot \vdash \lb M \rb_z :: z{:}\lb \tau\rb$
\item If $\Psi ; \Ga ; \D \vdash P :: z{:}A$ then
 $\{\lb\Psi\rb\} ; \lb\Ga\rb  ; \lb\D\rb \vdash \lb P\rb :: z{:}\lb
 A \rb$
\end{enumerate}
\end{lemma}

\begin{proof}
  By induction on the given judgement.
\noindent {\bf (1)} is immediate by induction.
  \begin{description}

\item[Case:] $\tau = \Pi x{:}\tau'.\sigma$

  \begin{tabbing}
    $\Psi \vdash \tau' :: \type$ and $\Psi , x{:}\tau' \vdash \sigma ::
    \type$ \` by inversion\\
    $\monad{\lb \Psi \rb} \vdash \lb \tau' \rb :: \stype$ \` by i.h.\\
    $\monad{\lb \Psi \rb}, x{:}\monad{\lb\tau'\rb} \vdash \lb \sigma
    \rb :: \stype$ \` by i.h.\\
    $\monad{\lb \Psi \rb} \vdash \forall x{:}\monad{\lb\tau'\rb}.\lb
    \sigma\rb :: \stype$ \` by $\{\}$ and $\forall$ rules
  \end{tabbing}

\item[Case:] $\tau = \monad{\ov{u_j{:}B_j};\ov{d_i{:}B_i} \vdash
    c{:}A}$

Straightforward by induction.

\item[Case:] $\tau = \lambda x{:}\tau'.\sigma$

  \begin{tabbing}
    $\Psi \vdash \tau' :: \type$ and $\Psi , x{:}\tau' \vdash \sigma ::
    \type$ \` by inversion\\
    $\monad{\lb \Psi \rb} \vdash \lb \tau' \rb :: \stype$ \` by i.h.\\
    $\monad{\lb \Psi \rb}, x{:}\monad{\lb\tau'\rb} \vdash \lb \sigma
    \rb :: \stype$ \` by i.h.\\
    $\monad{\lb \Psi \rb} \vdash \lambda x{:}\monad{\lb\tau'\rb}.\lb
    \sigma\rb :: \stype$ \` by $\{\}$ and $\lambda$ rules
  \end{tabbing}

\item[Case:] $\tau = \tau'\, M$

  \begin{tabbing}
$\Psi \vdash \tau' :: \Pi x{:}\sigma . K$ and $\Psi \vdash M : \sigma$
\` by inversion\\
$\monad{\lb \Psi \rb} \vdash \lb \tau'\rb :: \Pi x{:}\monad{\lb \sigma
  \rb}.\lb K \rb$ \` by i.h.\\
$\monad{\lb \Psi \rb} ; \cdot ; \cdot \vdash \lb M \rb_c :: c{:}\lb
\sigma\rb$ \` by i.h.\\
$\monad{\lb \Psi \rb} \vdash \monad{\lb M \rb_c} :: c{:}\monad{\lb
  \sigma\rb}$ \` by $\{\}I$\\
$\monad{\lb \Psi \rb} \vdash \lb \tau'\rb \, \monad{\lb M \rb_c} ::
\lb K \rb \{\monad{\lb M \rb_c}/x \}$ \` by application
well-formedness rule\\
$\monad{\lb \Psi \rb} \vdash \lb \tau'\rb \, \monad{\lb M \rb_c} ::
 \lb K \{{ M}/x \}\rb$ \` by compositionality
\end{tabbing}


\item[Case:] $\tau = \lambda t :: K . \tau'$

\begin{tabbing}
$\Psi , t :: K \vdash \tau' :: K_2$ \` by inversion\\
$\monad{\lb \Psi \rb} , t :: \lb K \rb \vdash \lb \tau' \rb :: \lb
K_2\rb$ \` by i.h.\\
$\monad{\lb \Psi \rb} \vdash \lambda t :: \lb K \rb . \lb \tau'\rb ::
\Pi t :: \lb K \rb . \lb K' \rb$ \` by $T\lambda$ well-formedness rule
\end{tabbing}

\item[Case:] $\tau = \tau' \, \sigma$

\begin{tabbing}
$\Psi \vdash \tau' :: \Pi t {::} K_1 . K_2$ and $\Psi \vdash \sigma ::
K_1$ \` by inversion\\
$\monad{\lb \Psi \rb} \vdash \lb \tau' \rb :: \Pi t {::}\lb K_1
\rb. \lb K_2\rb$ \` by i.h.\\
$\monad{\lb \Psi \rb} \vdash \lb \sigma \rb :: \lb K_1\rb$ \` by
i.h.\\
$\monad{\lb \Psi \rb} \vdash \lb \tau'\rb \, \lb \sigma\rb :: \lb
K_2\rb\{\lb K_1\rb/t\}$ \` by $Tapp$ well-formedness rule\\
$\monad{\lb \Psi \rb} \vdash \lb \tau'\rb \, \lb \sigma\rb :: 
\lb K_2\{ K_1 /t\}\rb$ \` by compositionality
\end{tabbing}
  
\item[Case:] $\tau = \tau'$ by conversion rule

\begin{tabbing}
$\Psi \vdash \tau' :: K$ \` by inversion\\
$\Psi \vdash K = K'$ \` by inversion\\
$\monad{\lb \Psi \rb} \vdash \lb \tau' \rb :: \lb K\rb$ \` by i.h.\\
$\monad{\lb \Psi \rb} \vdash \lb K \rb = \lb K' \rb$ \` by
preservation of equality\\
$\monad{\lb \Psi \rb} \vdash \lb \tau'\rb :: \lb K'\rb$ \` by
conversion rule
\end{tabbing}

\item[Case:] $A = \one$

Immediate from the definition.

\item[Case:] $A = \bang A'$

Immediate by i.h and $\bang$ well-formedness rule.
  
\item[Case:] $A = A_1 \lolli A_2$

Immediate by i.h. and $\lolli$ well-formedness rule.

\item[Case:] $A = A_1 \tensor A_2$

Immediate by i.h. and $\tensor$ well-formedness rule.
  
\item[Case:] $A = \forall x{:}\tau.A_0$

  \begin{tabbing}
  $\Psi \vdash \tau :: \type$ and $\Psi, x{:}\tau \vdash A_0 ::
  \stype$ \` by inversion\\
  $\monad{\lb\Psi\rb} \vdash \lb \tau \rb :: \stype$ \` by i.h.\\
  $\monad{\lb \Psi \rb} , x {:} \monad{\lb \tau \rb} \vdash \lb A_0 \rb ::
  \stype$ \` by i.h.\\
  $\monad{\lb \Psi \rb} \vdash \forall x{:}\monad{\lb \tau \rb}.\lb
  A_0 \rb :: \stype$ \` by $\forall$ well-formedness rule
  \end{tabbing}

\item[Case:] $A = \exists x{:}\tau.A_0$

  \begin{tabbing}
  $\Psi \vdash \tau :: \type$ and $\Psi, x{:}\tau \vdash A_0 ::
  \stype$ \` by inversion\\
  $\monad{\lb\Psi\rb} \vdash \lb \tau \rb :: \stype$ \` by i.h.\\
  $\monad{\lb \Psi \rb} , x {:} \monad{\lb \tau \rb} \vdash \lb A_0 \rb ::
  \stype$ \` by i.h.\\
  $\monad{\lb \Psi \rb} \vdash \exists x{:}\monad{\lb \tau \rb}.\lb
  A_0 \rb :: \stype$ \` by $\exists$ well-formedness rule
  \end{tabbing}
  
\item[Case:] $A = \with\{\ov{l_i {:}B_i}\}$

  Immediate by i.h. and $\with$ well-formedness rule.

\item[Case:] $A = \oplus\{\ov{l_i {:}B_i}\}$
  
  Immediate by i.h. and $\oplus$ well-formedness rule.

\item[Case:] $A = \lambda x{:}\tau.A'$

  \begin{tabbing}
$\Psi \vdash \tau :: \type$ and $\Psi , x{:}\tau \vdash A' :: K$
\` by inversion\\
$\monad{\lb\Psi\rb} \vdash \lb \tau \rb :: \stype$ \` by i.h.\\
$\monad{\lb \Psi \rb} , x{:}\monad{\lb \tau \rb} \vdash \lb A'\rb ::
\lb K \rb$ \` by i.h.\\
$\monad{\lb\Psi\rb} \vdash \lambda x{:}\monad{\lb \tau \rb} . \lb
A'\rb :: \Pi x{:}\monad{\lb \tau \rb}.\lb K \rb$ \` by $S\lambda$
well-formedness rule\\
$\monad{\lb\Psi\rb} \vdash \lambda x{:}\monad{\lb \tau \rb} . \lb
A'\rb :: \lb \Pi x{:} \tau . K \rb$ \` by compositionality
  \end{tabbing}
  
\item[Case:] $A = A_0\, M$

  \begin{tabbing}
    $\Psi \vdash A_0 :: \Pi x{:}\tau.K$ and
    $\Psi \vdash M : \tau$ \` by inversion\\
    $\monad{\lb\Psi\rb} \vdash \lb A_0\rb :: \Pi
    x{:}\monad{\lb\tau\rb}.\lb K \rb$ \` by i.h.\\
    $\monad{\lb\Psi\rb} ; \cdot ; \cdot \vdash \lb M \rb_c :: c{:}\lb \tau\rb$
    \` by i.h.\\
    $\monad{\lb\Psi\rb} \vdash \monad{\lb M \rb_c} :
    \monad{\lb\tau\rb}$ \` by $\{\}$\\
    $\monad{\lb\Psi\rb} \vdash \lb A_0\rb \, \monad{\lb M \rb_c} ::
    \lb K\rb\{ \monad{\lb M \rb_c} /x\}$ \` by $S$app well-formedness
    rule\\
    $\monad{\lb\Psi\rb} \vdash \lb A_0\rb \, \monad{\lb M \rb_c} ::
    \lb K \{M/x\} \rb$ \` by compositionality
  \end{tabbing}

\item[Case:] $A= \lambda t :: K.A'$

  \begin{tabbing}
$\Psi , t :: K \vdash A' :: K_2$ and $\Psi \vdash K_1$ \` by
inversion\\
$\monad{\lb \Psi \rb} , t :: \lb K \rb \vdash \lb A' \rb :: \lb
K_2\rb$ \` by i.h.\\
$\monad{\lb \Psi \rb} \vdash \lb K_1 \rb$ \` by i.h.\\
$\monad{\lb \Psi \rb} \vdash \lambda t :: \lb K_1 \rb . \lb A' \rb ::
\Pi t :: \lb K_1 \rb . \lb K_2\rb$ \` by $S\Pi$ well-formedness rule\\
$\monad{\lb \Psi \rb} \vdash \lambda t :: \lb K_1 \rb . \lb A' \rb ::
\lb \Pi t :: K_1 . K_2 \rb$ \` by compositionality\\

  \end{tabbing}
  
\item[Case:] $A = A' \, B$

  \begin{tabbing}
$\Psi \vdash A' :: \Pi t :: K_1 . K_2$ and $\Psi \vdash B :: K_1$ \`
by inversion\\
$\monad{\lb \Psi \rb} \vdash \lb A' \rb :: \Pi t :: \lb K_1 \rb . \lb
K_2 \rb$ \` by i.h.\\
$\monad{\lb \Psi \rb} \vdash \lb B \rb :: \lb K_1 \rb$ \` by i.h.\\
$\monad{\lb \Psi \rb} \vdash \lb A' \rb \, \lb B \rb :: \lb
K_2\rb\{\lb B\rb / x\}$ \` by S$app$ well-formedness rule\\
$\monad{\lb \Psi \rb} \vdash \lb A' \rb \, \lb B \rb :: \lb
K_2 \{ B /x \}\rb$ \` by compositionality
  \end{tabbing}

\item[Case:] $A = A'$ by conversion rule

  \begin{tabbing}
    $\Psi \vdash A' :: K$ and $\Psi \vdash K = K'$ \` by inversion\\
    $\monad{\lb \Psi \rb} \vdash \lb A' \rb :: \lb K\rb$ \` by i.h.\\
    $\monad{\lb \Psi \rb} \vdash \lb K \rb = \lb K' \rb$ \` by
    preservation of equality\\
    $\monad{\lb \Psi \rb} \vdash \lb A'\rb :: \lb K'\rb$ \` by
    conversion rule 
  \end{tabbing}

\item[Case:] $M = \lambda x{:}\tau.M'$

  \begin{tabbing}
    $\Psi , x{:}\tau \vdash M : \sigma$ \` by inversion\\
    $\monad{\lb \Psi \rb} , x{:}\monad{\lb \tau\rb} ; \cdot ; \cdot
    \vdash \lb M \rb_z :: z{:}\lb \sigma\rb$ \` by i.h.\\
    $\monad{\lb \Psi \rb} ; \cdot ; \cdot \vdash z(x).\lb M \rb_z ::
    z{:}\forall x{:}\monad{\lb \tau\rb}.\lb \sigma\rb$ \` by
    $\rgt\forall$\\

  \end{tabbing}
  
\item[Case:] $M = M' \, N$

  \begin{tabbing}
$\Psi \vdash M' : \Pi x{:}\tau.\sigma$ and $\Psi \vdash N : \tau$ \`
by inversion\\
$\monad{\lb \Psi \rb} ;\cdot ;\cdot \vdash \lb M'\rb_x :: x{:}\forall
x{:}\monad{\lb \tau\rb}.\lb \sigma\rb$ \` by i.h.\\
$\monad{\lb \tau\rb} ; \cdot ; \cdot \vdash \lb N \rb_y :: y{:}\lb
\tau\rb$ \` by i.h.\\
$\monad{\lb \tau\rb} \vdash \monad{\lb N \rb_y} : \monad{\lb\tau\rb}$
\` by $\{\}I$\\
$\monad{\lb \Psi \rb} ;\cdot ;\cdot \vdash (\nub x)(\lb M'\rb_x \mid
x\langle\monad{\lb N \rb_y} \rangle.[x\leftrightarrow z]) :: z{:} \lb
\sigma\rb\{\monad{\lb N \rb_y} / x\}$ \` by $\cut$, $\lft\forall$ and $\m{id}$\\
$\monad{\lb \Psi \rb} ;\cdot ;\cdot \vdash (\nub x)(\lb M'\rb_x \mid
x\langle\monad{\lb N \rb_y} \rangle.[x\leftrightarrow z]) :: z{:}
\lb \sigma\{N/x\}\rb$ \` by compositionality
\end{tabbing}

\item[Case:] $M = x$
\begin{tabbing}
  $\Psi , x{:}\tau \vdash x{:}\tau$ \` by assumption\\
  $\monad{\lb \Psi \rb} , x{:}\monad{\lb \tau\rb} ; \cdot ; \cdot
  \vdash y\leftarrow x ; [y\leftrightarrow z] :: z{:}\lb \tau\rb$ \`
  by $\monad{}E$ and $\m{id}$ rules
  
\end{tabbing}

\item[Case:] $M = \monad{c \leftarrow P \leftarrow \ov{u_j};\ov{d_i}}$

  \begin{tabbing}
  $\Psi \vdash  \monad{c \leftarrow P \leftarrow \ov{u_j};\ov{d_i}} :
  \monad{\ov{u_j{:}B_j} ; \ov{d_i{:}A_i} \vdash c{:}A}$ \` by
  assumption\\
  $\Psi ; \ov{u_j{:}B_j} ;  \ov{d_i{:}A_i} \vdash P :: c{:}A$ \` by
  inversion\\
  $\monad{\lb \Psi \rb} ; \ov{u_j{:}\lb B_j\rb} ; \ov{d_i{:}\lb
    A_i\rb} \vdash \lb P \rb :: c{:}\lb A \rb$  \` by i.h.\\
  $\monad{\lb \Psi \rb} ; \cdot ; \cdot \vdash
  z(u_0).\dots.z(u_j).z(d_0). \dots . z(d_n).\lb P\rb :: z{:}
  \ov{\bang\lb B_j\rb} \lolli \ov{\lb A_i\rb} \lolli \lb A\rb$ \\\` by
  $\lft{\bang}$, $\m{copy}$ and $\rgt\lolli$ (repeated)\\

  \end{tabbing}
  
\item[Case:] $M = M'$ by conversion rule
\begin{tabbing}
$\Psi \vdash M' : \sigma$ \` by inversion\\
$\Psi \vdash \sigma = \tau :: \type$ \` by inversion\\
$\monad{\lb \Psi \rb};\cdot ;\cdot \vdash \lb  M' \rb :: z{:}\lb
\sigma\rb$ \` by i.h.\\
$\monad{\lb \Psi \rb} \vdash \lb \sigma\rb = \lb \tau\rb :: \stype$ \`
by preservation of equality\\
$\monad{\lb \Psi \rb};\cdot ;\cdot \vdash \lb  M' \rb ::
z{:}\lb\tau\rb$ \` by conversion rule
\end{tabbing}
 
\item[Case:] $P = z\langle M \rangle.P_1$ by $\rgt\exists$

\begin{tabbing}
$\Psi \vdash M : \tau$ and $\Psi ; \Ga ; \D \vdash P_1 ::
z{:}A\{M/x\}$ \` by inversion\\
$\monad{\lb \Psi \rb} ; \cdot ; \cdot \vdash \lb M\rb_y :: y{:}\lb \tau\rb$ \` by
i.h.\\
$\monad{\lb \Psi \rb} \vdash \monad{\lb M\rb_y } : \monad{\lb
  \tau\rb}$ \` by $\monad{}I$\\ 
$\monad{\lb \Psi \rb} ; \lb \Ga \rb ; \lb \D \rb \vdash \lb P_1 \rb ::
z{:}\lb A\{M/x\}\rb$ \` by i.h.\\
$\monad{\lb \Psi \rb} ; \lb \Ga \rb ; \lb \D \rb \vdash \lb P_1 \rb ::
z{:}\lb A \rb\{\monad{\lb M \rb_y}/x\}$ \` by compositionality\\
$\monad{\lb \Psi \rb} ; \lb \Ga \rb ; \lb \D \rb \vdash z\langle
\monad{\lb M\rb_y }\rangle.\lb P_1 \rb :: z{:}\exists x{:}\monad{\lb
  \tau\rb}.\lb A\rb$ \` by $\rgt\exists$
\end{tabbing}

All other process cases follow straightforwardly by i.h. (and
compositionality/preservation of equality when needed).
 
\end{description}
\end{proof}


\begin{theorem}[Operational Correspondence -- Completeness]~
  \begin{enumerate}
  \item Let $\Psi ; \Ga ; \D \vdash P :: z{:}A$. If $P \tra{} P'$ then
    $\lb P \rb \tra{}^+ Q$ with
    $\monad{\lb\Psi\rb} ; \lb\Ga\rb ; \lb\D\rb \vdash Q = \lb P'\rb ::
    z{:}A$
\item Let $\Psi \vdash M : \tau$. If $M \tra{} M'$ then
  $\lb M \rb_z \tra{} N$ with
  $\monad{\lb \Psi \rb} ; \cdot ; \cdot \vdash N = \lb M '\rb_z
  ::z{:}\lb \tau\rb$
\end{enumerate}
\end{theorem}

\begin{proof}
  By induction on the reduction relation.

  \begin{description}

\item[Case:] $(\nub x)(x\langle M \rangle.P \mid x(y).Q) \tra{} (\nub
  x)(P \mid Q\{M/y\})$ 

\begin{tabbing}
$\lb (\nub x)(x\langle M \rangle.P \mid x(y).Q) \rb = 
(\nub x)(x\langle \monad{\lb M_c\rb} \rangle.\lb P\rb \mid x(y).\lb
Q\rb)$ \` by definition\\
$\tra{} (\nub x)(\lb P \rb \mid \lb Q\rb\{ \monad{\lb M_c\rb}/y\})$  \`
by operational semantics\\
$\lb (\nub  x)(P \mid Q\{M/y\}) \rb = (\nub x)(\lb P \rb \mid \lb Q
\{M/y\}\rb)$ \` by definition\\
$\Psi ; \Ga ; \D \vdash (\nub x)(\lb P \rb \mid \lb Q\rb\{ \monad{\lb
  M_c\rb}/y\}) = (\nub x)(\lb P \rb \mid \lb Q
\{M/y\}\rb) :: z{:}C$\\ \` by compositionality, type preservation and $\m{PEqCut}$
\end{tabbing}    

\item[Case:] $c\leftarrow M \leftarrow \ov{u_j};\ov{d_i};Q \tra{}
  c\leftarrow M' \leftarrow \ov{u_j};\ov{d_i};Q $ with $M \tra{} M'$

Straightforward by i.h.

\item[Case:] $c\leftarrow \monad{c\leftarrow P \leftarrow
    \ov{u_j};\ov{d_i}} \leftarrow \ov{u_j};\ov{d_i};Q \tra{} (\nub
  c)(P \mid Q)$

\begin{tabbing}
$\lb c\leftarrow \monad{c\leftarrow P \leftarrow
    \ov{u_j};\ov{d_i}} \leftarrow \ov{u_j};\ov{d_i};Q \rb \tra{}^+ 
(\nub c)(\lb P \rb \mid \lb Q \rb)$ \\\` by definition and operational
semantics\\ 
$\lb (\nub c)(P \mid Q)\rb = (\nub c)(\lb P \rb \mid \lb Q \rb)$ \` by
definition\\
We conclude by $\m{PEqR}$.
\end{tabbing}

  \item[Case:] $(\lambda x{:}\tau.M)\,N \tra{} M\{N/x\}$
    \begin{tabbing}
$\lb (\lambda x{:}\tau.M)\,N \rb_z = (\nub y)(y(x).\lb M \rb_y \mid
y\langle \monad{\lb N \rb_c}\rangle.[y\leftrightarrow z])$ \` by
definition\\
$\tra{} (\nub y)(\lb M \rb_y\{\monad{\lb N \rb_c}/x\} \mid
[y\leftrightarrow z])
\tra{} \lb M \rb_z\{\monad{\lb N \rb_c}/x\}$ \` by operational
semantics\\
$\monad{\lb \Psi \rb} ;\cdot ;\cdot \vdash \lb M \rb_z\{\monad{\lb N
  \rb_c}/x\} = \lb M\{N/x\}\rb_z :: z{:}\lb \sigma\{N/x\}\rb$\\ \` by
compositionality and type preservation

\end{tabbing}

\item[Case:] $M\, N \tra{} M' \, N$ with $M \tra{} M'$

\begin{tabbing}
$\lb M\, N \rb_z = (\nub x)(\lb M \rb_x \mid x\langle \monad{\lb
  N\rb_c}\rangle.[x\leftrightarrow z])$ \` by definition\\
$\lb M \rb_x \tra{} M_0$ with $\monad{\lb \Psi \rb} ; \cdot ; \cdot
\vdash M_0 = \lb M'\rb_z :: z{:}A$ \` by i.h.\\
$(\nub x)(\lb M \rb_x \mid x\langle \monad{\lb
  N\rb_c}\rangle.[x\leftrightarrow z]) \tra{}
(\nub x)(M_0 \mid  x\langle \monad{\lb
  N\rb_c}\rangle.[x\leftrightarrow z])$ \\\` by the operational
semantics\\
$=(\nub x)(\lb M'\rb_x \mid  x\langle \monad{\lb
  N\rb_c}\rangle.[x\leftrightarrow z]) :: z{:}A$ \` by type
preservation, $\m{PEqCut}$ and $\m{PEqR}$
\end{tabbing}

  \end{description}
\end{proof}


\end{document}